\documentclass[11pt,letterpaper]{article}
\usepackage[margin=1in]{geometry}

\newlength{\abstractwidth}
\usepackage{graphicx}
\usepackage{amsmath, nccmath}
\usepackage{amsfonts,amsthm, amssymb,amsbsy}
\usepackage{latexsym}
\usepackage{color}
\usepackage{graphicx}
\usepackage{dsfont}
\usepackage{mathbbol}
\usepackage{authblk}
\usepackage{setspace}
\usepackage{subcaption}
\usepackage{simpler-wick}
\usepackage[titletoc,page]{appendix}
\usepackage{booktabs}
\usepackage{algorithmic}
\usepackage[ruled,vlined]{algorithm2e}
\usepackage{multirow,tabularx}
\graphicspath{{figures/}}

\DeclareSymbolFontAlphabet{\amsmathbb}{AMSb}%
\usepackage{braket}
\usepackage{mathtools}
\usepackage{tensor}
\usepackage{mathrsfs}
\usepackage{comment}
\usepackage{cite}
\newcommand{\expval}[1]{\braket{#1}}

\usepackage{hyperref}
\definecolor{darkred}{rgb}{0.8,0.1,0.1}
\hypersetup{colorlinks=true, linkcolor=darkred, citecolor=blue, linktoc=page}
\usepackage[capitalize]{cleveref}

\numberwithin{equation}{section}

\newtheorem{thm}{Theorem}[section]
\newtheorem{assumption}[thm]{Assumption}
\newtheorem{lemma}[thm]{Lemma}

\newtheorem{rem}[thm]{Remark}
\newtheorem{prop}[thm]{Proposition}
\newtheorem{cor}[thm]{Corollary}

\newtheorem{definition}[thm]{Definition}

\theoremstyle{definition}
\newtheorem{problem}{Problem}

\newcommand{\EE}{ {\amsmathbb E}}
\newcommand{\PP}{ {\amsmathbb P}}
\newcommand{\E}{ {\amsmathbb E}}
\newcommand{\Sph}{ {\amsmathbb S}}

\newcommand{\complex}{ {\amsmathbb C} }

\AtBeginDocument{}

\def\cA{{\cal A}}

\def\cV{{\cal V}}

\def\Tr{{\rm Tr}}
\def\det{{\rm det \,}}
\def\half{{1\over 2}}

\def\p{\partial}

\newcommand{\Id}{\amsmathbb{I}}

\def\a{\alpha}
\def\b{\beta}
\def\g{\gamma}
\def\d{\delta}
\def\s{\sigma}

\def\eps{\epsilon}
\def\ep{\varepsilon}

\def\no{\nonumber}

\DeclareMathOperator{\poly}{poly}
\DeclareMathOperator{\diag}{diag}
\newcommand{\vect}[1]{\boldsymbol{#1}}

\newcommand{\bv}{{\vect{b}}}

\newcommand{\pv}{{\vect{p}}}

\newcommand{\sv}{{\vect{s}}}
\newcommand{\xv}{{\vect{x}}}

\newcommand{\zv}{{\vect{z}}}
\newcommand{\param}{\beta,\gamma,\sigma}
\newcommand{\paramv}{\betav,\gammav}
\newcommand{\gammav}{{\vect{\gamma}}}
\newcommand{\etav}{{\vect{\eta}}}
\newcommand{\betav}{{\vect{\beta}}}
\renewcommand{\pmb}{\vect}

\newcommand{\RR}{\amsmathbb{R}}

\newcommand{\abs}[1]{\left|#1\right|}

\begin{document}

\title{CV-QAOA: Efficient Low-Depth Quantum Optimization of Continuous Variables}

\author[1]{Sriram Bharadwaj}
\author[2,3,4]{Di Luo}
\author[2]{Leo Zhou}

\affil[1]{\small \it Mani L. Bhaumik Institute for Theoretical Physics, Department of Physics and Astronomy,
\authorcr \small \it  University of California, Los Angeles, CA 90095, USA} 
\affil[2]{\small \it Electrical and Computer Engineering Department,
University of California, Los Angeles, CA 90095, USA}
\affil[3]{\small \it Department of Physics, Tsinghua University, Beijing 100084, China}
\affil[4]{\small \it Institute of Advanced Study, Tsinghua University, Beijing 100084, China}

\date{\vspace{-40pt}}
\maketitle

\begin{center}
{\tt \small sbharadwaj@physics.ucla.edu, diluo@tsinghua.edu.cn, leoxzhou@ucla.edu}

\bigskip
\bigskip
\date{\large October 5, 2026}
\end{center}
\bigskip

\AtBeginEnvironment{abstract}{\onehalfspacing}

\begin{abstract}
{\noindent 
We study a Continuous-Variable Quantum Approximate Optimization Algorithm (CV-QAOA) for high-dimensional continuous optimization. Our formulation extends an earlier CV-QAOA proposal~\cite{verdon2019qaoa} with a variationally optimized initial state and recovers the convergence guarantees of Quantum Hamiltonian Descent (QHD) in the high-depth limit. 
We prove rigorous performance guarantees of CV-QAOA on several families of cost functions. First, we show $d$-step CV-QAOA minimizes any $d$-dimensional strictly convex quadratic function with $2d$ quantum queries. We then analyze a family of nonconvex ``Rotated Double Well'' (RDW) functions with $2^d$ local minima introduced by Leng \emph{et al.}~\cite{leng2023nonconvex}. While prior work showed QHD reaches its global minimum with $\tilde{O}(d^3)$ queries, we prove that 1-step CV-QAOA solves RDW with just two quantum queries. 
Although general-purpose classical solvers need superpolynomial time for RDW~\cite{leng2023nonconvex} and structure-awareness reduces the cost to polynomial time~\cite{herman2025mechanisms}, we show that the 1-step CV-QAOA protocol can be efficiently dequantized, and that a gradient-aligned line search succeeds with $O(d)$ queries, nearly matching the information-theoretic $\Omega(d/\log d)$ query lower bound. 
To move beyond the dequantizable regime, we introduce a ``Rotated Square Well'' (RSW) problem, whose globally flat landscape suppresses useful local gradient information. 
For this family, we show that an adiabatic evolution simulated by CV-QAOA can reach the global minimum using $d^{o(1)}$ queries. On the other hand, any classical algorithm that learn the hidden rotation in RSW provably requires $\Omega(d^2/\log d)$ queries, a bound we nearly match with an explicit $\Theta(d^2\log d)$-query classical algorithm.
Numerical simulations on deflected corrugated spring and Easom functions illustrate the promising performance of CV-QAOA on more general problems.
}
\end{abstract}

\newpage

\tableofcontents

\clearpage

\section{Introduction}

Continuous optimization problems, arising naturally from mathematical modeling of real-world systems, are ubiquitous across physics, applied mathematics, operations research, computer science, and engineering. Such problems are concerned with minimizing a cost function $C(\xv)$ over $d$ real-valued variables $\xv\in\RR^d$. While many convex problems admit efficient classical methods, high-dimensional nonconvex landscapes remain difficult because local information may be misleading: gradient-based and second-order methods can converge to suboptimal critical points, while other methods like simulated annealing and basin hopping become
costly when the number of local minima grows superpolynomially with $d$. Quantum algorithms offer a different approach by evolving wavefunctions over the search space rather than following a single classical trajectory. Most of the existing quantum methods for continuous optimization often quantize classical procedures or accelerate classical subroutines~\cite{brandao2017sdp,van2017sdp,chakrabarti2020convex,zhang2021saddle}, and therefore do not necessarily change the underlying optimization dynamics. Quantum Hamiltonian Descent (QHD) or Real-space Adiabatic Algorithm (RsAA)~\cite{leng2023QHD, herman2025mechanisms} provides a more intrinsically quantum alternative by quantizing gradient-flow dynamics, but its finite-time behavior is difficult to analyze rigorously beyond special cases.

In this paper, we develop and analyze a Continuous-Variable Quantum Approximate Optimization Algorithm (CV-QAOA) that may be applied to any continuous optimization problem. While our CV-QAOA can be understood as a slightly generalized variant of a previous proposal~\cite{verdon2019qaoa}, here we provide a framework to rigorously prove performance guarantees for optimizing various high-dimensional functions. 
For example, we prove that $d$-step CV-QAOA can find the minimum of any $d$-dimensional convex quadratic function.
Similar to the discrete setting where QAOA is shown to have advantages over the quantum adiabatic algorithm in performance and analyzability~\cite{montanaro2024quantum}, we expect CV-QAOA to be a more powerful and versatile framework than QHD/RsAA.

Indeed, we show that CV-QAOA has an advantage over QHD on certain nonconvex optimization problems.
Specifically, we show that 1-step CV-QAOA can minimize a family of nonconvex $d$-dimensional quartic functions that we call the ``Rotated Square Well'' (RDW) problem, previously studied in Ref.~\cite{leng2023nonconvex}. This implies CV-QAOA can find the global minimum with $O(1)$ quantum queries to the cost function, 
improving over the known performance of the QHD~\cite{leng2023nonconvex} that solves the same family with $\tilde{O}(d^3)$ quantum queries.
As a complementary result, we provide an information-theoretic lower bound of $\Omega(d/\log d)$ queries are required by any classical algorithms with $O(\log d)$-bit precision black-box access to the cost function.
Although previous empirical studies~\cite{leng2023nonconvex} have suggested that state-of-the-art classical solvers require super-polynomial run time (in $d$), we show that various classical algorithm, including one based on dequantizing the 1-step CV-QAOA, can solve the RDW problem in $O(\poly(d))$ time.
In fact, we exhibit a line-search algorithm that achieves this with $O(d)$ queries, due to the local gradient at the origin revealing substantial information about global minimum. These classical results complement a recent work \cite{herman2025mechanisms} that independently uncovered efficient classical algorithms for the RDW problem.

Towards a stronger quantum advantage for CV-QAOA beyond the dequantizable regime, we consider another family of optimization problems that we call the ``Rotated Square Well'' (RSW).
The landscape of the RSW family is globally flat almost everywhere, so local gradients provide little information.
This RSW family not only evades the classical dequantization method that solves the RDW problem but also yields a stronger query separation over the classical strategies analyzed here. 
Specifically, we show that RsAA/QHD solves this problem in runtime $T=O\bigl(\log^{8} d\bigr)$, and that CV-QAOA can simulate its dynamics with $d^{o(1)}$ phase queries to the cost function. 
On the other hand, the best classical algorithm we found is based on a random-ray tomography that uses $O(d^2\log d)$ queries, which nearly match the classical information-theoretic lower bound of ${\Omega}(d^2/\log d)$ queries for any classical algorithm based on reconstructing the hidden rotated structure.

\subsection{A Continuous-Variable QAOA}

We now describe the CV-QAOA, a general-purpose quantum continuous-optimization algorithm inspired by both the QAOA~\cite{farhi2014quantum}, and Quantum Hamiltonian Descent (QHD) \cite{leng2023QHD, chakrabarti2025speedups}. We remark that our theoretical analysis of the QHD algorithm is often performed in the adiabatic regime, in which case it is also referred to as the Real-space Adiabatic Algorithms (RsAA)\cite{herman2025mechanisms}. 

Our task is to optimize a continuous-variable cost function $C: \RR^d\to\RR$. Using canonical quantization, we can promote the cost function to an operator $\hat{C}=C(\hat\xv)$ acting on $d$ bosonic modes, where $\hat{\xv} = (\hat{x}_1,\ldots, \hat{x}_d)$ and $\hat{x}_i = (\hat{a}_i^\dag + \hat{a}_i)/\sqrt{2}$ can be written in terms of the bosonic creation and annihilation operators $\hat{a}_i$ and $\hat{a}_i^\dag$. 
We also define the driver Hamiltonian
\begin{equation}
    \hat{B} = \frac{1}{2}\hat{\pmb{p}}^2 = \frac12 \sum_{j=1}^d \hat{p}_j^2, \qquad  \text{where} \qquad \hat{p}_j=i(\hat{a}_j^\dag-\hat{a}_j)/\sqrt{2}.
\end{equation}
Let $\gammav=(\gamma_1,\ldots,\gamma_p)$ and $\betav=(\beta_1,\ldots,\beta_p)\in \RR^p$ be variational parameters. 
The $p$-step CV-QAOA starts with an initial Gaussian state of the $d$ bosons parameterized by $\sigma$,
\begin{equation}\label{initalSt}
    \ket{\sigma}, \quad \text{where} \quad \braket{\xv|\sigma} = \frac{1}{(2\pi\sigma^2)^{d/4}} \exp\left(-\frac{1}{4\sigma^2} \xv^\top\xv\right),
\end{equation}
and then applies $p$ layers of alternative unitary operations generated by the driver Hamiltonians $\hat{B}$ and $\hat{C}$ to prepare and optimize:
\begin{align}\label{QAOAansatz}
    \ket{\paramv,\sigma} &= e^{-i\beta_p \hat{B}}e^{-i\gamma_p \hat{C}} \cdots
    e^{-i\beta_2 \hat{B}}e^{-i\gamma_2 \hat{C}}
    e^{-i\beta_1 \hat{B}}e^{-i\gamma_1 \hat{C}}
    \ket{\sigma},&\braket{\hat C}_{\betav, \gammav,\s} &= \bra{\paramv,\sigma}\hat C \ket{\paramv,\sigma}.
\end{align}
We remark that the above ansatz for CV-QAOA is similar to a variant of the QAOA previously proposed for continuous problems, where the width $\sigma$ of the initial Gaussian state is not optimized~\cite{verdon2019qaoa}. 
Here, our proposal is motivated by the fact that as the number of layers $p$ increases, CV-QAOA can approximate QHD to arbitrary precision, with an overhead that scales polynomially with the precision parameters and the Fock space truncation limits.
As a result, we expect CV-QAOA can retain or even surpass the quantum speedups observed in QHD,
while also being potentially easier to implement and more amenable to rigorous theoretical analysis. We now introduce the query model that will be used to characterize classical and quantum complexity.
\begin{definition}[Query model]\label{def:query}
Consider any cost function $C:\RR^d\to\RR$. For any positive integer $b=O(\log d)$, let $C^b(\xv)$ be a $b$-bit representation of $C(\xv)$.
Then a \emph{classical query} upon input $\xv$ returns $C^b(\xv)$ as a bit string.
A \emph{quantum query} is a single application of the reversible oracle
\begin{equation}
    A^b_C \ket{\xv}\ket{z} = \ket{\xv}\ket{z + C^b(\xv)}.
\end{equation}
A \emph{quantum phase query} to the cost function with parameter $\gamma$ is unitary $U_C^b(\gamma)\ket{\xv} = e^{-i\gamma C^b(\xv)}\ket{\xv}$.
\end{definition}

In this model, $p$-step CV-QAOA makes exactly $p$ quantum phase queries.
At $b$-bit precision, the $p$-step QAOA implemented with this query model has error at most $2^{-b} \sum_{j=1}^p |\gamma_i|$ in operator norm compared to using an idealized infinitely-precise $C(\xv)$. The restriction of $b=O(\log d)$ here is partly motivated  
so that we only consider "reasonable" algorithm parameters $\sum_i |\gamma_i| \le \poly(d)$; nevertheless, in principle $b=O(\poly(d))$ is possible in a digital implementation.
Note that a quantum phase query $U_C(\gamma)$ can be implemented by two standard quantum queries: apply $A_C^b$ with an zeroed value register, rotate the value register by the known phase $e^{-i\gamma C(\xv)}$, and then apply $(A_C^b)^{-1}$ to erase the value bits.

\begin{figure}[h!]
    \centering
    \includegraphics[width=\linewidth]{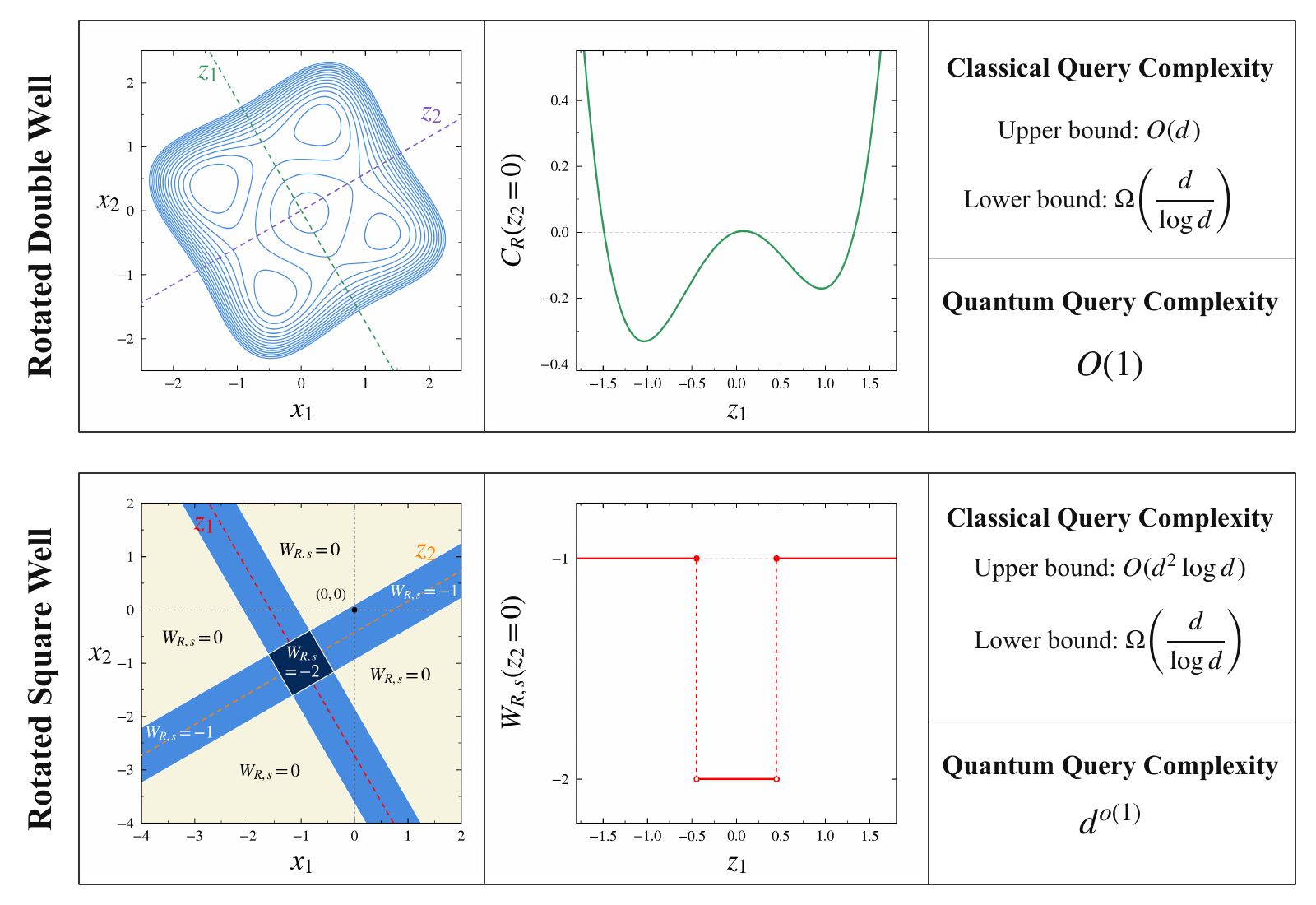}
    \caption{
Summary of the CV-QAOA framework and the main quantum advantages established in this work.
The figure compares two complementary optimization regimes controlled by the initial width parameter $\sigma$.
The top panels show the rotated double well family: a two-dimensional slice of the rotated high-dimensional landscape and a hidden-coordinate slice revealing the underlying one-dimensional double well structure.
The bottom panels show the rotated square well family: a two-dimensional view and a hidden-coordinate slice, illustrating a flat staircase-like landscape whose global optimum lies in the overlap of the hidden wells.
The displayed classical query complexities are upper bounds: $O(d)$ for the rotated double well problem
(Proposition~\ref{prop:RDW-O(d)-bound}) and $O(d^2\log d)$ for the rotated square well problem (Proposition~\ref{prop:rsw-random-ray-query-complexity}).
The corresponding information-theoretic lower bounds are $\Omega(d/\log d)$ for every classical algorithm, and $\Omega(d^2/\log d)$ for those algorithms that reconstruct the hidden rotation in either family.
These examples show that CV-QAOA can exploit both localized $(\sigma\to0)$ and highly delocalized $(\sigma\to\infty)$ quantum states to obtain provable query complexity separations.
}
    \label{fig:summary}
\end{figure}

\subsection{Summary of results}
In this subsection, we give an overview of our main results via three (informal) theorems, which we will unpack in the rest of this paper. Figure~\ref{fig:summary} offers a pictorial summary of our key results on non-convex optimization problems.

For our first theorem, we consider the analytically simple case of $d$-dimensional strictly-convex quadratic functions.
The result highlights the two qualitatively different regimes in which CV-QAOA is found to have good performance, depending on whether the initial state $\ket{\sigma}$ is \textit{delocalized} (i.e. $\s\rightarrow\infty$) or \textit{localized} ($\s\rightarrow 0$).

\begin{thm}[Informal: depth $d$ suffices for strictly convex quadratics]\label{thm:inf-quad}
For any $d$-dimensional strictly convex quadratic with Hessian $Q\succ0$,
$d$-step CV-QAOA can achieve an arbitrarily small optimality gap
from a Gaussian initial state of any width $\sigma>0$.
Specifically, consider any choice of parameters $\paramv$ satisfying $\beta_k\gamma_k=1/q_k$, where $\{q_k\}_{k=1}^d$ are the
eigenvalues of $Q$, then for sufficiently small $P(\paramv)=\sum_{i<j}|\beta_i\gamma_j|$,
\[
\abs{\braket{\hat C}_{\paramv,\sigma} - \min_{\xv}C(\xv)}
\le K_Q\left[
(\sigma^2+\|\xv_0\|^2)P^2
+d^2\left(\frac{\hat\beta}{\sigma}\right)^2
\right],
\]
where $\xv_0$ is the minimizer, $\hat\beta=\max_k|\beta_k|$, and
$K_Q$ depends only on $Q$.
Separating the scales of the layer parameters makes $P$ arbitrarily
small, while scaling all $\beta_k$ down at fixed $\beta_k\gamma_k$
reduces $\hat\beta/\sigma$ without changing $P$.
Thus, with the first term controlled, both delocalized initial
states ($\sigma\to\infty$ at fixed $\hat\beta$) and localized initial
states ($\sigma\to0$ with $\hat\beta/\sigma\to0$) allow the gap to vanish.
\end{thm}

We remark that the optimal initial width $\sigma$ is highly problem dependent. If we were to follow the intuition guided by the adiabatic theorem (as in RsAA), one might expect that a delocalized initial state with a large spread, which minimizes $\braket{\hat B}$, is the favorable option to converge to the unknown global minimum. 
For strictly convex quadratics, however, Theorem~\ref{thm:inf-quad} shows that CV-QAOA does not single out either regime: both a \emph{delocalized} initial state ($\sigma\to\infty$ at fixed $\hat\beta$, matching the adiabatic intuition) and a \emph{localized} one ($\sigma\to0$ with $\hat\beta\ll\sigma$) can yield an optimality gap arbitrarily small.
Indeed, we find that using a localized initial state allows one to minimize a large class of non-convex real-analytic functions while a delocalized initial state fails to minimize for any choice of CV-QAOA parameters, as we observe explicitly for the Styblinski-Tang function in Section~\ref{sec:1DST}.
Importantly, using a localized initial state works well in certain high-dimensional non-convex problems, and Theorem~\ref{infthm2} below gives a sufficient criterion under which this is effective.

\begin{thm}[Informal: CV-QAOA for $d$-dimensional real-analytic functions]\label{infthm2}
Consider any $d$-dimensional real analytic function $C:\RR^d\to \RR$.
Suppose its gradient satisfy $ \nabla{C}(\vect{0})=\lambda \xv_*$ for any $|\lambda|> 1/\poly(d)$ and any global minimizer $\xv_*=\arg\min_\xv C(\xv)$. Then the one-step CV-QAOA can minimize $C(\xv)$ with a localized initial state $\ket{\sigma}$ centered at the origin.
\end{thm}

In particular, for $d=1$, the gradient condition is essentially $|C'(0)|\neq 0$, which is easily satisfied for many real-analytic function $C$; this is investigated in detail in Theorem~\ref{thm:1DrealAnalytic}. The formal version of the above applied to the general case can be found in Theorem \ref{thm:GradPointingToMin}.

The above theorem has an important consequence for the so-called ``Rotated Square Well" (RDW) problem analyzed in \cite{leng2023nonconvex}, 
which is a family of $d$-dimensional quartic functions constructed from summing independent 1D quartic double wells and performing a hidden rotation to mask the separability. 
Here, we consider a generalized family of such problems that are defined below.

\begin{problem}[Optimization of a rotated real-analytic function]
\label{prob:rotated-analytic}
Let \(w:\RR\to\RR\) be a real-analytic function, and define the separable \(d\)-dimensional cost with \(R\in \mathbf{O}(d)\) being an orthogonal matrix:
\begin{align}
\label{eq:C-R-general}
    C_{R}(x_1,\dots,x_d)=\sum_{k=1}^d w\left(\sum_{\ell=1}^dR_{k,\ell}x_\ell\right).
\end{align}
The \emph{rotated real-analytic optimization problem} is to minimize \(C_R(\xv)\) over \(\xv\in\RR^d\), given oracle access to \(C_R\), without knowledge of the hidden rotation \(R\).
In particular, when $w(z)$ is a quartic double well function (e.g., $w(z)=z^4-16z^2 + 5z$), we call this problem the \emph{Rotated Double Well (RDW)}.
\end{problem}

The RDW problem was initially proposed by \cite{leng2023nonconvex} to be classically hard due to the presence of $2^d$ local minima, slowing down general-purpose solvers. Ref.\cite{leng2023nonconvex} also argued for hardness on the basis that learning the rotation from the quartic polynomial coefficients is a special case of an NP-hard problem~\cite{hillar2013tensor}. 
On the other hand, Theorem~\ref{infthm2} implies that CV-QAOA can solve the rotated family in~\eqref{eq:C-R-general} including the RDW, provided $w$ has a global minimizer
$x_0$ and $w'(0)\ne0$. Indeed,
$\xv_*=x_0R^\top\vect{1}
=(x_0/w'(0))\nabla C_R(\vect{0})$,
which verifies the theorem's gradient-alignment hypothesis.
The formal result is presented in Corollary~\ref{thm:doublewell}.
Beyond solving rotated real-analytic problems, the above theorem also explains precisely when 1-step CV-QAOA is able to find optimal parameters that drives the optimality gap to zero. 
Since this theorem uses a point-like initial state, $\sigma\to0$, we show in Proposition~\ref{prop:dequant} that the resulting 1-step CV-QAOA state is well-approximated by a Gaussian state and is therefore classically simulable. Thus, this regime is effectively dequantized and should be viewed as a quantum-inspired classical algorithm.

While the above-mentioned dequantization mechanism works for $1$-step CV-QAOA with a point-like initial state, it is important to emphasize that higher-step CV-QAOA is not necessarily dequantizable. Indeed, we consider such concrete instances of quantum advantage with $2$-step CV-QAOA in Section~\ref{sec:beyondClassSim}.

The query complexity of CV-QAOA, QHD, and the classical algorithms
considered in~\cite{herman2025mechanisms,leng2023nonconvex} is
summarized in Table~\ref{tab:query_complexity}.
We also mention two classical approaches arising from our analysis.
For RDW with $w'(0)\ne0$, the gradient at the origin is proportional to the global minimizer, yielding an $O(d)$-query algorithm by line search (Proposition~\ref{prop:RDW-O(d)-bound}).
The second approach, dequantized $1$-step CV-QAOA, requires
estimates of the gradient and Hessian at the origin.
For RDW, the Hessian is known: $\nabla^2C_R(\vect{0}) = w'(0)\Id$, so $O(d)$ value queries suffice.
For more general real-analytic costs, $O(d^2)$ value
queries suffice to estimate both derivatives, with oracle
precision chosen for the desired simulation accuracy
(Remark~\ref{rem:dequant-query-cost}).
By comparison, $1$-step CV-QAOA with a point-like initial state
finds the global minimum of RDW with $O(1)$ quantum queries,
improving on the previous $\tilde{O}(d^3)$ query upper bound
for QHD~\cite{leng2023nonconvex}.
\begin{table}[h]
    \centering
    \begin{tabular}{l | l | c}
        \toprule
        \textbf{Type} & \textbf{Algorithm} & \textbf{Query Complexity} \\
        \midrule
        \multirow{4}{*}{Classical}
        & General-purpose solvers (e.g., Gurobi, SGD) \cite{leng2023nonconvex}
        & $\mathrm{superpoly}(d)$ \\
        & Hessian-learning \cite{herman2025mechanisms}
        & $\tilde{\Omega}(d^a)$, $a\ge2$ \\
        & Convexity-honing \cite{herman2025mechanisms}
        & $\tilde{\Omega}(d^b)$, $b\ge1$ \\
        & Line search along $\nabla C_R(\vect{0})$
          (Proposition~\ref{prop:RDW-O(d)-bound})
        & $\Theta(d)$ \\
        \midrule
        \multirow{2}{*}{Quantum}
        & QHD \cite{leng2023nonconvex}
        & $\tilde{O}(d^3)$ \\
        & CV-QAOA
        & $O(1)$ \\
        \bottomrule
    \end{tabular}
    \caption{Query complexity comparison of classical and quantum
    algorithms for the rotated double well (RDW) problem.
    The constants $a$ and $b$ depend on the choice of the
    one-dimensional function in the cited analyses \cite{herman2025mechanisms}; the lower bounds $a\ge 2$ and $b\ge1$ are based on information-theoretic counting arguments.
    }
    \label{tab:query_complexity}
\end{table}

At this stage, we ask whether there exist problems where CV-QAOA  has a larger quantum advantage relative to classical approaches and improve upon the factor-of-$d$ quantum advantage observed on the RDW problem.
A natural class of problems that can evade the classical attack offered by \cref{infthm2} (and other classical algorithms in Table~\ref{tab:query_complexity}) has a globally flat landscape with an exponentially small gradient that does not carry any information about the global minimum.
The simplest such instance, which we analyze in depth, is the so-called rotated square well problem:
\begin{problem}[Rotated Square Well]
Let $R\in \textbf{O}(d)$ be an $d$-dimensional rotation and $\vect s\in\RR^d$. The \emph{Rotated Square Well (RSW)} problem is to find a global minimum of the following cost function
\begin{align}\label{eq:rot-sq-well}
    W_{R,\sv}(\xv) = -\sum_{j=1}^d\mathbf{1}\{|(R\xv+\sv)_j|<\d\},
\end{align}
with only oracle access to $W_{R,\sv}$ and without knowledge of $R, \sv$.
\end{problem}

On this problem, we prove that there exists a black-box query complexity separation between classical and quantum optimization algorithms, such as CV-QAOA. The key technical result needed to achieve the separation is summarized in the following informal theorem (formalized in Theorem~\ref{thm:qaa_log_runtime} and Corollary~\ref{cor:queryComplexity}). This allows us to achieve a nearly superpolynomial-in-dimension separation between quantum algorithms and any classical rotation-learning algorithms in the black-box setting, with the quantum side requiring only $d^{o(1)}$ queries.

\begin{thm}[Informal: runtime and query complexity of CV-QAOA on rotated square-well]\label{thm:inf-thm-3}
For the Rotated Square Well problem specified by a unknown rotation $R\in \mathbf{O}(d)$ and bounded shift $\|\vect s\|_\infty = O(1)$, a RsAA/QHD algorithm with runtime $T$ or a CV-QAOA with query complexity $Q$ can find the global minimum of $W_{R,\sv}$ with constant success probability, where
\begin{align}
    T &= O\left(\log^{8} d\right), & Q&=d^{o(1)}.
\end{align}
In particular, the runtime is polylogarithmic and the query complexity is subpolynomial in $d$. 
\end{thm}
To prove the above theorem, we utilize the fact that CV-QAOA can be used simulate the RsAA/QHD.
We start by analyzing RsAA/QHD with an optimized schedule and solve the rotated square well problem exactly to obtain the scaling of the adiabatic runtime $T$ (see Theorem~\ref{thm:qaa_log_runtime}). Following that, we use higher-order product formulas to obtain the query complexity (see Corollary~\ref{cor:queryComplexity}) from the runtime $T$.

On the classical side, we give an information-theoretic lower bound (see Proposition~\ref{prop:rsw-info-lower-bound}) on classical algorithms that try to learn the rotation $R$, which forces a $\Omega(d^2/
\log d)$ lower bound, essentially because of degree-of-freedom counting. Then we explicitly construct a classical $\Theta(d^2\log d)$-query rotation-learning algorithm that nearly matches the lower bound. 
Beyond classical algorithms that learn the rotation $R$, we also tested various general-purpose classical optimizers, including path integral Monte Carlo, which is a standard classical method to simulate RsAA.
All of the test algorithms, including PIMC, exhibit superpolynomial scaling in our simulation up $d=60$ (see Section~\ref{sec:numerics-RSW}).
We also analyze a structure-aware adaptive random walk algorithm, and find that $\Theta(d^3/\log d)$  queries is needed to solve the RSW problem (see Section~\ref{subsubsec:do-not-learn-rotation}).
Hence, this suggests that learning the rotation may indeed be the optimal strategy for this problem, which we leave open as a conjecture.

Finally, we numerically explore low-depth CV-QAOA with a uniform initial state on the two-dimensional Easom and deformed deflected corrugated spring functions (Section~\ref{sec:beyondClassSim}). On the Easom function,
whose gradient is exponentially small away from the minimum,
CV-QAOA achieves consistently higher success probability than RsAA/QHD at matched integrated runtime. On the deformed deflected corrugated spring, where the gradient at the origin is not aligned with the global minimizer, $2$-step CV-QAOA achieves a $39.93\%$ probability of returning a point within unit distance of the minimum, compared with $6.99\%$ for the linear-ramp adiabatic algorithm and $3.14\%$ for QHD. These numerical results illustrate the potential of low-depth CV-QAOA on more general nonconvex landscapes and motivate extending our theoretical analysis beyond the families considered here.

\subsection{Related work}
\label{sec:related}

The CV-QAOA we propose in this paper can be understood as an algorithmic approach inspired by both the Quantum Hamiltonian Descent (QHD)~\cite{leng2023QHD} and the Quantum Approximate Optimization Algorithm (QAOA)~\cite{farhi2014quantum}. In this section, we explain what is novel about our approach and compare them to previous results.

We first describe QHD in more detail, which is proposed~\cite{leng2023QHD} as a truly quantum continuous optimization algorithm formed by quantizing the classical gradient-flow dynamics through a path integral formalism. 
The QHD algorithm involves evolving a system of $d$ bosons with the following time-dependent Hamiltonian
\[
\hat{H}(t) = e^{\varphi(t)} \hat{B} + e^{\chi(t)} \hat{C},
\]
Here $e^{\varphi(t)}$ and $e^{\chi(t)}$ are time-dependent damping parameters controlling the energy flow in the system. A key strategy of QHD is to choose $\varphi(t)$ and $\chi(t)$ to satisfy $e^{\varphi(t)-\chi(t)}\to 0$ as $t\to\infty$, so that kinetic energy is gradually drained out of the system.

Notably, QHD has been proven to achieve (sub)exponential quantum speedups over certain black-box optimization problem~\cite{leng2025subexponential}.
While asymptotic convergence to the global minimum have been derived in the infinite time limit~\cite{leng2023QHD}, there is no general analysis of the performance for nonconvex functions in the polynomial time regime.
Hence, the performance of QHD on more general objective functions in finite time is often heuristic, guided by numerical simulations~\cite{leng2023QHD}.

Manifestly, QHD is similar to a quantum adiabatic algorithm that interpolates between the kinetic Hamiltonian $\hat{B}$ and the potential Hamiltonian $\hat{C}$, although key differences remain~\cite{leng2023QHD}.
In the discrete-optimization setting,  the QAOA was proposed with a key feature that it is at least as powerful as the adiabatic algorithm due to its capability to simulate the latter~\cite{farhi2014quantum}.
Thus, the similarity between QHD and the adiabatic algorithm inspires the proposal of a discrete-time version of QHD. This is the CV-QAOA described in equation~\eqref{QAOAansatz}, a variational quantum algorithm given by a two-driver ansatz utilizing the two Hamiltonians $\hat{B}$ and $\hat{C}$.
As the number of layers $p$ increases, the CV-QAOA can approximate QHD to arbitrary precision via standard arguments with Trotter decomposition.

We remark that a previous work~\cite{verdon2019qaoa} has introduced a continuous-variable variant of the QAOA that is essentially identical to equation~\eqref{QAOAansatz}. However, they did not optimize the choice of initial state $\ket{\sigma}$, nor did they provide any rigorous performance guarantee. As we find in our study of the quartic Styblinski-Tang cost function (see Section~\ref{sec:1DST}) that was also studied in Ref.~\cite{verdon2019qaoa}, the optimization of the initial state is crucial to the success of 1-step CV-QAOA to achieve an arbitrarily-well approximation of the global minimum. In contrast, Ref.~\cite{verdon2019qaoa} did not achieve the same level of performance despite optimizing their ansatz up to $p=3$ steps. 

Quantum gradient methods for convex optimization achieve
polylogarithmic dimension dependence using approximate
function-value oracles, including a $\tilde O(\kappa)$ query
bound for smooth strongly convex objectives with condition
number $\kappa$~\cite[Theorems~4.2 and~4.3]{augustino2025fastconvex}.
Complementary lower bounds show that constant-accuracy
optimization requires
$\tilde{\Omega}(\min\{\sqrt{\kappa},d\})$ quantum gradient
queries in the worst case, even for
quadratics~\cite[Theorem~6.1]{augustino2026nearoptimal}.
Our work complements these results with an analytic $d$-step
CV-QAOA construction for strictly convex quadratics, achieving
arbitrarily small optimality gap with depth and phase-query
bounds independent of $\kappa$
(Section~\ref{sec:ClassicalComplexity}).

Furthermore, we note that PIMC has been widely used to classically simulate quantum annealing, including for continuous real-space optimization problems such as double-well landscapes~\cite{stella2006monte,santoro2006optimization}.
More recently, \cite{herman2025mechanisms} suggests that PIMC may can suffer from sampling and mixing bottlenecks even when the corresponding adiabatic quantum evolution (RsAA) remains efficient.
Our results show that CV-QAOA outperforms PIMC in the rotated square well problem, which advances the state of the art and shows that substantially more effective optimization dynamics can be realized. Beyond achieving higher-quality solutions, it introduces a distinct optimization paradigm that opens new possibilities for solving challenging high-dimensional continuous optimization problems.

\subsection{Outline of the paper}

We begin in Section~\ref{sec:Quadratics1D} by analyzing the performance of CV-QAOA on quadratic functions in arbitrary dimension. In Section~\ref{sec:RealAnalyticD}, we extend this analysis to real-analytic functions in arbitrary dimension, including rotated separable problems such as the RDW, more general real-analytic objective functions, and the dequantization of one-step CV-QAOA with a point-like initial state. Then we turn to the rotated square well problem, which is designed to evade this dequantization mechanism and for which we prove polylogarithmic quantum runtime and subpolynomial query complexity. In Section~\ref{sec:ClassicalComplexity}, we compare these quantum results with corresponding classical complexity bounds, including information-theoretic query lower bounds for rotated real-analytic problems and $\Theta(d^3/\log d)$ hitting-time bounds for a natural family of random-walk algorithms on the rotated square well. Finally, Section~\ref{sec:Numerics} presents numerical simulations validating the analytical results and exploring the performance of CV-QAOA beyond the regimes covered by our rigorous theorems. The appendices contain the detailed proofs and supplementary technical estimates.

\section{Quadratic functions in arbitrary dimension}\label{sec:Quadratics1D}
In this section, we analyze the performance of CV-QAOA on $d$-dimensional convex quadratic functions as a warm-up for the non-convex problems studied in this work. We formalize Theorem~\ref{thm:inf-quad}, which shows that $d$-step CV-QAOA can find the global minimum of any $d$-dimensional strictly convex quadratic function to arbitrary accuracy from a Gaussian initial state of any width $\sigma>0$, with appropriately chosen circuit parameters. We first state the precise result in Theorem~\ref{thm:quad}, analyze the CV-QAOA transformation of position and momentum operators in Lemma~\ref{lem:px}, and then sketch its proof. We defer the derivation of error bounds to Appendix~\ref{app:Quadratics}.
\begin{thm}[Depth \(d\) suffices for strictly convex quadratics]
\label{thm:quad}
Given any $d$-dimensional strictly convex quadratic, i.e. $C(\xv) =\frac12
\xv^\top Q \xv + \bv^\top \xv$ where $Q\succ 0$ has eigenvalues
$q_1,\dots,q_d$ and minimizer $\xv_0=-Q^{-1}\bv$, and an initial
$d$-dimensional Gaussian state $\ket{\sigma}$ as in \eqref{initalSt},
consider $d$ steps of CV-QAOA as in \eqref{QAOAansatz} with parameters
$(\betav,\gammav,\sigma)$. Whenever the pairwise products of the parameters
are sufficiently small, in the sense that $\hat q\, P \le \tfrac12$ below,
the optimality gap satisfies
\begin{align}
\left|\braket{\hat C}_{\betav, \gammav,\s}
-\min_{\xv} C(\xv)\right|
\le
4\Tr \,Q\left[\left({\sigma^2}+\|\xv_0\|^2\right)
\left[
\sqrt{\zeta}\max_{1\le k\le d}\abs{1-\lambda_k q_k}
+
\hat q P
\right]^2
+
\frac{d^2\a^2\hat\beta^2}{\sigma^2}\right],
\end{align}
where $\lambda_k=\beta_k\gamma_k$ and $\cA_k=\Id-\lambda_k Q$,
{$\sqrt{\zeta}=\sum_{k=1}^d\prod_{j\neq k}\abs{1-\lambda_j q_k}$},
$\hat q=\max\{1,(\Tr Q)^d\}$,
$P=\sum_{1\le i<j\le d}|\beta_i|\,|\gamma_j|$,
$\hat\beta=\max_{1\le j\le d}|\beta_j|$,
$\a = \max\{\|\prod_{k=1}^{j-1}\cA_k \|_1\;|\; j=1,\dots,d\}$ with the empty
product read as $\Id$, and one may take $K=4\Tr(Q)$.
In particular, the explicit sequential choice $\lambda_k=1/q_k$, with
$\beta_1\to0,\dots,\beta_{d-1}\to0$ taken in this order at fixed $\lambda_k$
and $\hat\beta/\sigma\to0$ thereafter, drives the optimality gap to zero.
\end{thm}

Before we sketch the proof of Theorem~\ref{thm:quad}, we will need the following lemma which allows us to compute expectation values of polynomials of the operators $(\hat{x},\hat{p})$. This was given in Ref.~\cite{verdon2019qaoa} but we provide a full proof here.
\begin{lemma}[One-layer evolution]\label{lem:px}
Let $B=\hat p^2/2$, $C=h(\hat x)$ for a real-valued real-analytic
function $h$, and $U=e^{-i\beta B}e^{-i\gamma C}$. Then
\begin{align}
U^\dagger\hat p\,U
&=\hat p-\gamma h'(\hat x),&U^\dagger\hat x\,U
&=\hat x+\beta\hat p-\beta\gamma h'(\hat x).
\end{align}
\end{lemma}
\begin{proof}
We compute the mixer and cost conjugations separately.
The canonical relation $[\hat x,\hat p]=i$ gives
\[
[B,\hat x]
=\frac12\bigl(\hat p[\hat p,\hat x]+[\hat p,\hat x]\hat p\bigr)
=-i\hat p,
\qquad
[B,\hat p]=0.
\]
For the cost, $[C,\hat x]=0$ because $C=h(\hat x)$.
To compute $[C,\hat p]$, use $\hat p=-i\partial_x$: acting on a
smooth function $f$ of compact support,
\[
[C,\hat p]f
=-ihf'+i(hf)'
=ih'f.
\]
Thus $[C,\hat p]=ih'(\hat x)$. Since $h'(\hat x)$ commutes with
$C$, the relevant double commutators vanish:
\[
[B,[B,\hat x]]=0,
\qquad
[C,[C,\hat p]]=0.
\]
The Baker-Campbell-Hausdorff expansion therefore terminates after
the first commutator, giving
\begin{align*}
e^{i\beta B}\hat x e^{-i\beta B}
&=\hat x+i\beta[B,\hat x]
=\hat x+\beta\hat p,&
e^{i\gamma C}\hat p e^{-i\gamma C}
&=\hat p+i\gamma[C,\hat p]
=\hat p-\gamma h'(\hat x).
\end{align*}
Meanwhile, mixer conjugation leaves $\hat p$ unchanged and cost
conjugation leaves $\hat x$ unchanged. For
$U=e^{-i\beta B}e^{-i\gamma C}$, we consequently obtain
\begin{align*}
U^\dagger\hat p\,U
&=e^{i\gamma C}\hat p e^{-i\gamma C}
=\hat p-\gamma h'(\hat x),&
U^\dagger\hat x\,U
&=e^{i\gamma C}(\hat x+\beta\hat p)e^{-i\gamma C}=\hat x+\beta\bigl(\hat p-\gamma h'(\hat x)\bigr).
\end{align*}
These are the claimed identities; see also
Ref.~\cite{verdon2019qaoa}.
\end{proof}

The proof generalizes directly to higher dimensions. In what follows, we will drop the hats on the operators $(\hat{x}, \hat{p})$. Now we sketch the proof of Theorem~\ref{thm:quad}.
\paragraph{Proof sketch of Theorem~\ref{thm:quad}.}
CV-QAOA is analytically tractable because both the cost and mixer Hamiltonians generate linear transformations of position and momentum operators due to Lemma~\ref{lem:px}. On quadratics, the CV-QAOA state therefore remains Gaussian throughout the circuit, and we can follow its center and spread explicitly. We wish to find circuit parameters such that an initially broad wavepacket becomes sharply concentrated around the minimizer. In coordinates centered at the minimizer and aligned with the eigenvectors of $Q\succ0$, write
\begin{align}
     C-C_{\min}=\frac12\zv^\top D\zv,
    \qquad D=\mathrm{diag}(q_1,\dots,q_d),
\end{align}
and denote the conjugate momentum by $\pv_z$. After $n$ layers, $U_n=e^{-i\beta_n B}e^{-i\gamma_n C}\cdots e^{-i\beta_1 B}e^{-i\gamma_1 C}$, the position operators take the form $U_n^\dagger\zv U_n=R^{(n)}\zv+S^{(n)}\pv_z$, where $R^{(n)}$ and $S^{(n)}$  are determined recursively -- see Eq.~\eqref{recursionRS}. Thus, the final position depends on both the initial position and the initial momentum. Evaluating the cost in the initial Gaussian state gives
\begin{align}\label{eqExpMain}
    \braket{\hat C}_{\betav,\gammav,\sigma}-C_{\min}
    &=
    \frac12\bra{\sigma}
    \zv^\top R^{(d)\top}D R^{(d)}\zv
    \ket{\sigma}
    +
    \frac{1}{8\sigma^2}
    \Tr\bigl(S^{(d)\top}D S^{(d)}\bigr).
\end{align}
To bound these two contributions, set $\lambda_k=\beta_k\gamma_k=1/q_k$.
The leading part of $R^{(d)}$ is $\prod_{k=1}^d(\Id-\lambda_kD)$,
which vanishes because each eigendirection is annihilated by one factor.
The recurrence relations bound the remaining terms and the momentum
coefficient by
\[
\|R^{(d)}\|\le c_R\hat q P,
\qquad
\|S^{(d)}\|\le c_S d\a\hat\beta,
\qquad
P=\sum_{i<j}|\beta_i\gamma_j|,
\]
for absolute constants $c_R,c_S>0$, provided $\hat qP\le1/2$, with $\hat q,\a,\hat\beta$ as defined in
Theorem~\ref{thm:quad}. Since the initial position has covariance
$\sigma^2\Id$ and displacement of norm $\|\xv_0\|$ from the minimizer,
the position contribution is bounded by a constant times
$\Tr(D)(\sigma^2+\|\xv_0\|^2)(\hat qP)^2$. The momentum contribution
is bounded by a constant times
$\Tr(D)d^2\a^2\hat\beta^2/\sigma^2$. Together, these give
\[
\braket{\hat C}-C_{\min}
\le
\Tr(D)\left[
\frac{c_R^2}{2}(\sigma^2+\|\xv_0\|^2)(\hat qP)^2
+\frac{c_S^2}{8}\frac{d^2\a^2\hat\beta^2}{\sigma^2}
\right].
\]
At fixed $\lambda_j=1/q_j$, we have
$|\beta_i\gamma_j|=|\beta_i/\beta_j|/q_j$, so a hierarchy
$|\beta_1|\ll\cdots\ll|\beta_d|$ makes $P$ arbitrarily small.
Scaling all $\beta_j$ down together, while keeping the $\lambda_j$
fixed, then reduces $\hat\beta/\sigma$ without changing $P$.
Thus both contributions can be made arbitrarily small for any
initial width $\sigma>0$.

A complete proof with the explicit error bounds is given in Appendix~\ref{app:Quadratics}.
\section{Real-analytic functions in arbitrary dimension}\label{sec:RealAnalyticD}

In this section, we analyze the performance of $1$-step CV-QAOA with a point-like initial state on  $d$-dimensional real-analytic functions. First, building on a detailed study of the one-dimensional case, we identify a broad class of high-dimensional nonconvex problems for which the algorithm can achieve an arbitrarily small optimality gap. This class includes the rotated separable real-analytic functions introduced in Problem~\ref{prob:rotated-analytic}, where an unknown rotation obscures the separable structure of the cost. It also includes more general real-analytic functions whose global minimum is aligned with the gradient at the origin, formalizing Theorem~\ref{infthm2}. Second, we show that $1$-step CV-QAOA with a point-like initial state can be efficiently dequantized, which leads to a search for quantum advantages beyond this regime (see Sec.~\ref{sec:RSW}).

\subsection{One-dimensional real-analytic functions}\label{sec:1DST}

We begin with arbitrary one-dimensional real-analytic functions. This case is rich enough to exhibit the basic mechanism behind \(1\)-step CV-QAOA in the point-like regime $\s\rightarrow0$, and serves as a building block for the higher-dimensional nonconvex problems studied in this section. 
We will assume the function growth at most polynomial at large $|x|$ for technical reasons.
When the initial Gaussian is sharply localized at the origin, we show that the \(1\)-step CV-QAOA expectation value is governed by an effective displacement determined by the gradient \(C'(0)\). In this sense, the algorithm behaves like a quantum line search along the distinguished direction selected by gradient of the cost at the center of the initial Gaussian state. We show the following theorem.
\begin{thm}[1-step CV-QAOA for 1D real-analytic functions]\label{thm:1DrealAnalytic}
Let $C:\RR\to\RR$ be real analytic with $C'(0)\neq 0$. 
Assume that, for every integer $k\ge0$, there are constants $K_k,m_k<\infty$ such that $|C^{(k)}(x)|\le K_k(1+|x|)^{m_k}$ for all $x\in\RR$. 
Then, for every fixed $\lambda\in\RR$,
\begin{equation}
\label{eq:1d-main-bound}
\bra{\beta,\lambda/\beta,\sigma}\hat C\ket{\beta,\lambda/\beta,\sigma}
=
C\bigl(-\lambda C'(0)\bigr)
+
O\left(\frac{\beta^2}{\sigma^2}+\sigma^2\right).
\end{equation}
In particular, if $x_* \in \arg\min_x C(x)$ and $\lambda=-{x_*}/{C'(0)}$, then $C\bigl(-\lambda C'(0)\bigr) = \min_x C(x)$.
\end{thm}
The content of Theorem~\ref{thm:1DrealAnalytic} is as follows. In the regime \(\sigma\to 0\), the initial state becomes sharply localized near the origin. To leading order in the CV-QAOA parameters, the state after one CV-QAOA layer behaves as though the initial wavepacket had been displaced from \(x=0\) to the effective point
\[
x_{\mathrm{eff}}=-\lambda C'(0).
\]
Thus the parameter \(\lambda=\beta\gamma\) selects a point along the line determined by the local gradient, and the optimization problem reduces to choosing \(\lambda\) so that this effective displacement lands near the global minimizer. This is the basic mechanism that will reappear later in higher-dimensional settings.

Note that the theorem requires \(C'(0)\neq 0\). If \(C'(0)=0\), then the leading-order displacement vanishes, and \(1\)-step CV-QAOA with a point-like initial state can no longer move effectively away from the origin at finite \(\lambda\). Nonetheless, provided a neighborhood of the origin contains a point where the derivative is not exponentially small, one can displace the initial state to such a point and still utilize $1$-step CV-QAOA. However, if the global cost landscape has exponentially vanishing gradients everywhere, then this strategy fails. This obstruction allows us to construct classically hard problems for which $1$-step CV-QAOA  with a point-like initial state is not sufficient (see Sec.~\ref{sec:RSW}).

\paragraph{Example: 1D Styblinski-Tang.} To concretely illustrate how the mechanism behind Theorem \ref{thm:1DrealAnalytic}, we now analyze the one-dimensional Styblinski-Tang function. This is a useful test case because it is nonconvex, has multiple local minima, and yet is sufficiently explicit that the \(1\)-step CV-QAOA expectation value can be computed analytically.
\begin{align}
    w(x)=\frac12\left(x^4-16x^2+5x\right).
    \label{ST-1D}
\end{align}
This function has a global minimum at \(x\approx -2.90353\), with value \(w(x)\approx -39.1662\), and a secondary local minimum at \(x\approx 2.7468\), with value \(w(x)\approx -25.0294\). It therefore provides a simple setting in which one can ask whether a shallow quantum circuit can distinguish the true minimum from a nearby false one. Using Lemma~\ref{lem:px}, the expectation value of the cost operator \(\hat w=w(\hat x)\) in the \(1\)-step CV-QAOA state can be evaluated exactly:
\begin{equation}
\begin{aligned}
\label{wOptST}
\braket{\hat{w}}_{\param} 
&= \braket{\sigma| e^{i\gamma C} e^{i\beta B}  w(\hat{x}) e^{-i\beta B} e^{-i\gamma C}|\sigma}
 = \braket{\sigma| w(\hat{x} - \beta\gamma w'(\hat{x}) + \beta \hat{p})|\sigma} \\
&= \frac{3\beta^4}{32\,\sigma^4}
+ \frac{\tfrac{75}{16}\,\beta^4\gamma^2 \;-\; 2\beta^2}{\sigma^2}
\nonumber\\
&\quad+ \frac{\beta}{32}\Bigl[
    \beta^3\gamma^2\,(625\gamma^2+6144)
  + 768\,\beta^2\gamma
  - 8\,\beta\,(200\gamma^2-3)
  - 200\,\gamma
\Bigr]\\
&\quad+ \sigma^2\Bigl[
    48\,\beta^4\gamma^2\,(100\gamma^2-1)
  + \beta^3\,(600\gamma^3-3\gamma)
  - \tfrac{8117}{4}\,\beta^2\gamma^2
  - 256\,\beta\gamma
  - 8
\Bigr]\\
&\quad+ \sigma^4\Bigl[
    \beta^4\,(94704\gamma^4+9\gamma^2)
  + 24351\,\beta^3\gamma^3
  + 3840\,\beta^2\gamma^2
  + 192\,\beta\gamma
  + \tfrac32
\Bigr]\\
&\quad- 15\,\beta\gamma\,\sigma^6\Bigl[
    16309\,\beta^3\gamma^3
  + 3072\,\beta^2\gamma^2
  + 224\,\beta\gamma
  + 4
\Bigr]\\
&\quad+ 1260\,\beta^2\gamma^2\,\sigma^8\,(16\beta\gamma+1)^2
- 15120\,\beta^3\gamma^3\,\sigma^{10}\,(16\beta\gamma+1)
+ 83160\,\beta^4\gamma^4\,\sigma^{12}.
\end{aligned}
\end{equation}
The expression is cumbersome, but its detailed form is less important than the lesson it contains. For each fixed \(\sigma\), we may optimize over \((\beta,\gamma)\), and the result is shown in Fig.~\ref{fig:AnalyticalSTplot}. The striking feature is that the delocalized regime \(\sigma\gg 1\) is highly suboptimal, while the true minimum is approached in the opposite limit \(\sigma\to 0\). Thus, unlike the quadratic case discussed earlier, nonconvex one-dimensional analytic functions are naturally optimized by a sharply localized initial state.

\begin{figure}[h]
    \centering
    \includegraphics[width=0.75\linewidth]{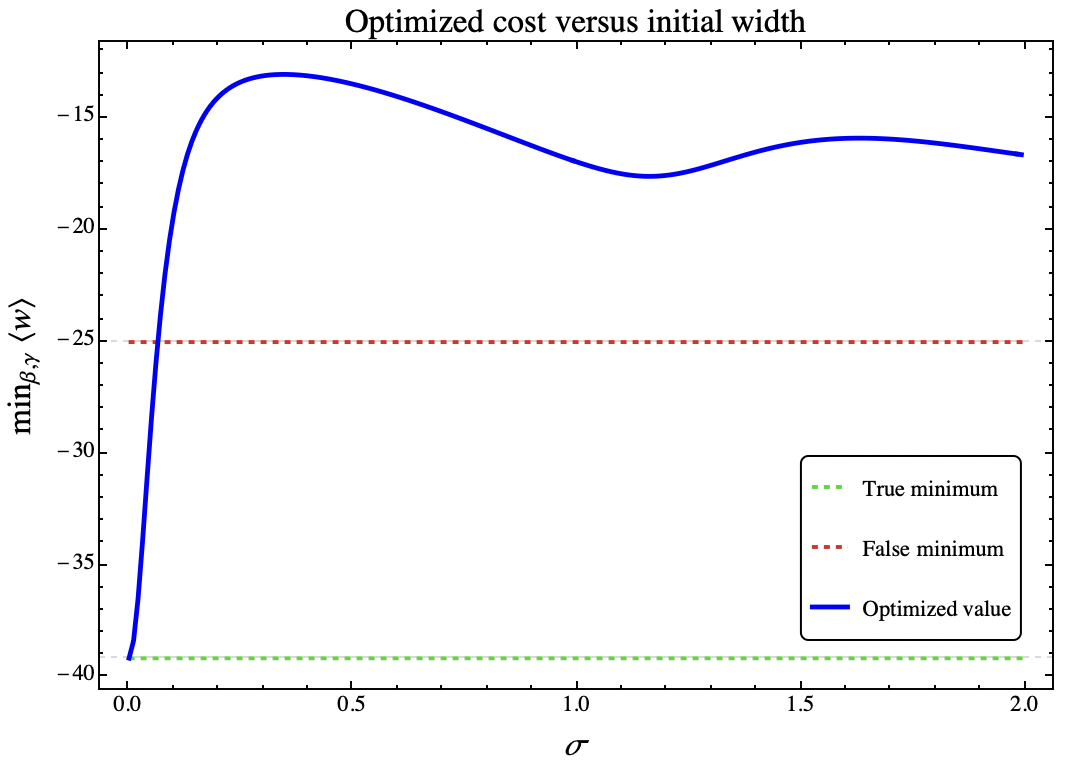}
    \caption{The optimized analytical expectation value \(\min_{\beta,\gamma}\expval{w}\) as a function of the initial width \(\sigma\) for the one-dimensional Styblinski-Tang function. The red dashed line marks the false minimum, while the green dashed line marks the true minimum.}
    \label{fig:AnalyticalSTplot}
\end{figure}
This behavior is exactly what Theorem~\ref{thm:1DrealAnalytic} predicts. In the point-like regime, one should keep the product \(\beta\gamma\) finite while sending \(\sigma\to 0\), so that the circuit induces a finite effective displacement of the localized wavepacket. Guided by this, we set
\begin{align}
    \beta=\alpha\sigma
\end{align}
and optimize in the small-\(\sigma\) regime. One finds
\begin{align}
    \gamma^*(\alpha,\sigma)\label{gOpt}
    &= \arg \min_\gamma \braket{\hat{w}}_{\gamma,\alpha\sigma,\sigma}
    = \frac{1}{\sigma}\Bigl[\frac{1.161}{\alpha} + \mathcal{O}(\alpha^0)\Bigr] + \mathcal{O}(\sigma^0), \\
    \braket{\hat{w}}_{\gamma^*,\alpha\sigma,\sigma}
    &\approx -39.1662 + 8.646\,\alpha^2 + 3315.47\,\sigma^2 + \cdots.
    \label{fOpt}
\end{align}
Hence the optimized expectation value approaches the true global minimum \(-39.1662\) as \(\sigma\to 0\) and \(\alpha\to 0\). In particular, \(1\)-step CV-QAOA does not merely settle at the nearby false minimum; it reaches the correct one.

The Styblinski-Tang example therefore makes the theorem intuitive. In the point-like regime, the QAOA layer effectively translates the initial wavepacket to a location determined by the local slope at the origin, and by tuning \(\lambda=\beta\gamma\) one can steer this effective displacement toward the global minimum. For Styblinski-Tang, this simple mechanism is already enough to overcome the presence of a false minimum. The same idea will underlie our treatment of more general quartic and higher-dimensional rotated problems. A complete proof of Theorem~\ref{thm:1DrealAnalytic} is deferred to Appendix~\ref{app:Quadratics}.

\subsection{Which $d$-dimensional problems can 1-step CV-QAOA solve?}

We now turn to real-analytic optimization problems in higher dimensions, ask the question: when can \(1\)-step CV-QAOA still succeed for a general high-dimensional real-analytic objective? The theorem below identifies a concrete sufficient condition. Namely, \(1\)-step CV-QAOA works whenever the global minimizer lies on the line determined by the gradient at the origin. In this regime, the point-like initial state again causes the circuit to behave effectively as a line search along a distinguished direction selected by local gradient information. This main result is stated in the theorem below, and a detailed proof is in Appendix~\ref{app:PfGradMin}.

\begin{thm}[$1$-step CV-QAOA for $d$-dimensional real-analytic problems with $\nabla C(\vect{0})\propto \arg\min C$]
\label{thm:GradPointingToMin}
Let \(C:\RR^d\to\RR\) be real analytic and assume that every partial derivative of $C$ has at most polynomial growth on $\RR^d$. Suppose a global minimizer \(\xv_*\in\arg\min_{\xv} C(\xv)\) and a finite \(\lambda_*\in\RR\) satisfy
\begin{align}
    \xv_*=-\lambda_* \nabla C(\vect{0}).
\end{align}
Then the \(1\)-step CV-QAOA with \(\gamma=\lambda_*/\beta\) satisfies
\begin{align}
    \bra{\beta,\lambda_*/\beta,\sigma}\hat C
    \ket{\beta,\lambda_*/\beta,\sigma}
    =
    \min_{\xv} C(\xv)
    +
    O_{C,d,\lambda_*}\left(
        \sigma^2+\frac{\beta^2}{\sigma^2}
    \right).
\end{align}
Here the limit is $\sigma\to0$ and $\beta/\sigma\to0$,
with $C$, $d$, and $\lambda_*$ fixed.
Hence \(1\)-step CV-QAOA minimizes \(C\) up to arbitrarily small error.
\end{thm}

We now specialize this general perspective to a concrete family that will be important throughout the rest of the paper and that directly connects our analysis to earlier work. Start from a separable real-analytic function, for which the preferred coordinates are known, and then hide that structure by composing with an unknown orthogonal transformation. From the point of view of the ambient coordinates, the resulting objective is no longer separable and may appear substantially more difficult to optimize. 

A particularly important special case of Problem~\ref{prob:rotated-analytic} is the Rotated Double Well (RDW) problem, obtained by choosing
\begin{align}
    w(x)=x^4+ax^3+bx^2+cx
\end{align}
to be a quartic function with two non-degenerate local minima. This problem was previously studied in Ref.~\cite{leng2023nonconvex}. The point of the example is that the hidden rotation destroys the manifest product structure of the objective, even though the underlying one-dimensional geometry is unchanged.

For this rotated family, the hypothesis of
Theorem~\ref{thm:GradPointingToMin} is satisfied whenever
$w'(0)\ne0$ and $w$ obeys its growth condition. Indeed, if $x_0\in\arg\min_x w(x)$, then
\begin{align}
\label{eq:RDW-global-min}
    \xv_* = x_0 R^\top \vect{1},
    \qquad
    \nabla C_R(\vect{0}) = w'(0)R^\top \vect{1},
\end{align}
so the global minimizer and the gradient at the origin are parallel.
This same alignment underlies the $O(d)$-query classical algorithm
for RDW established in Proposition~\ref{prop:RDW-O(d)-bound}
in Section~\ref{sec:classical-RDW}. Moreover, in this rotated setting one can say more than the general theorem alone: because both the mixer Hamiltonian and the initial Gaussian state are rotationally invariant, the variational objective is in fact independent of the hidden rotation \(R\). This yields the following sharpened corollary.

\begin{cor}[$1$-step CV-QAOA for rotated real-analytic functions]
\label{thm:doublewell}
Consider Problem~\ref{prob:rotated-analytic}, and assume that \(w:\RR\to\RR\) satisfies the growth condition of Theorem~\ref{thm:1DrealAnalytic} and that \(w'(0)\neq 0\). Let $\lambda = -x_*/w'(0)$. Then the \(1\)-step CV-QAOA objective is independent of the hidden rotation \(R\), and setting $\gamma = \lambda/\beta$ we have
\begin{align}
    \braket{\param|\hat{C}_R|\param}
    =
    \min_{\xv} C_R(\xv)
    +
    O\left(d\sigma^2+d(\beta/\sigma)^2\right).
\end{align}
In particular, \(1\)-step CV-QAOA minimizes the rotated real-analytic problem up to arbitrarily small optimality gap, with a choice of parameters that is independent of \(R\).
\end{cor}

\begin{proof}
To prove Corollary~\ref{thm:doublewell}, we proceed in two steps. First, we show that the CV-QAOA objective function \(F_R(\param):=\braket{C_R}_{\param}\) is independent of the choice of \(R\in\mathbf{O}(d)\), i.e. \(F_R=F_{\amsmathbb{I}}\). Second, we use the one-dimensional result of Theorem~\ref{thm:1DrealAnalytic}.

We start by establishing that \(F_R(\param)=F_{\Id}(\param)\) for every \(R\in\mathbf{O}(d)\). The mixer Hamiltonian \(B=\pmb{p}^\top\pmb{p}/2\) enjoys manifest \(\mathbf{O}(d)\) rotational symmetry, i.e. for any \(R\in\mathbf{O}(d)\), the transformation \(\pmb{p}\mapsto R\pmb{p}\) leaves \(B\) invariant. Moreover, this transformation also leaves the QAOA state invariant, since the argument of the Gaussian \(\xv^\top\xv\) is manifestly \(\mathbf{O}(d)\)-symmetric. These statements can be rephrased in terms of a unitary representation \(\cV:\mathbf{O}(d)\rightarrow \mathrm{Op}(\mathcal{H}^d)\), where \(\mathrm{Op}(\mathcal{H}^d)\) is the space of linear operators acting on the Hilbert space of \(d\) bosons, with the convention that \(\cV(R)^\dagger\,\xv\,\cV(R)=R\xv\). Then for every \(R\in\mathbf{O}(d)\),
\begin{align}
    \cV(R)^\dag ~C_{\Id}(\xv)~\cV(R)= C_R(\xv),
    \qquad
    \cV(R) ~B~\cV(R)^\dag =B,
    \qquad
    \cV(R)\ket{\sigma} = \ket{\sigma}.
\end{align}
Putting these observations together, we find
\begin{align}
F_R(\param) 
    &= \bra{\sigma}e^{i\gamma C_R(\xv)}e^{i\beta B}C_R(\xv) e^{-i\beta B} e^{-i\gamma C_R(\xv)}\ket{\sigma}\nonumber\\
    &= \bra{\sigma} \cV(R)^\dag e^{i\gamma C_\Id(\xv)} \cV(R)e^{i\beta B} \cV(R)^\dag C_\Id(\xv) \cV(R) e^{-i\beta B} \cV(R)^\dag e^{-i\gamma C_\Id(\xv)} \cV(R)\ket{\sigma}\nonumber\\
    &=\bra{\sigma}e^{i\gamma C_\Id(\xv)}e^{i\beta B}C_\Id(\xv) e^{-i\beta B} e^{-i\gamma C_\Id(\xv)}\ket{\sigma}
    = F_\Id(\param).
\end{align}
This proves the claim.

As a result, to compute the CV-QAOA cost expectation \(F_R\) for any \(R\in\mathbf{O}(d)\), it is sufficient to work with \(R=\amsmathbb{I}\). But \(F_{\Id}=\braket{C_{\Id}}\) is simply a sum of \(d\) independent copies of the one-dimensional expectation value \(\expval{w}_{\beta,\gamma,\sigma}\), so by Theorem~\ref{thm:1DrealAnalytic},
\begin{equation}\label{rotInv}
    F_R(\param)=F_\Id(\param)=d\,\expval{w}_{\beta,\gamma,\sigma}.
\end{equation}
Hence,
\begin{align}
    \min_{\param}F_R(\param)
    &= d\,\min_{\param}\expval{w}_{\beta,\gamma,\sigma}
     = d\left[\min_x w(x)+O\left(\sigma^2+(\beta/\sigma)^2\right)\right]\no\\
    &= \min_{\xv} C_R(\xv)+O\left(d\sigma^2+d(\beta/\sigma)^2\right).
\end{align}
This completes the proof.
\end{proof}

\subsection{Dequantization of 1-step CV-QAOA with point-like initial state}
\label{sec:dequantization}

Despite its improved performance over QHD for problems such as the RDW, \(1\)-step CV-QAOA with a point-like initial state (\(\sigma\ll 1\)) can in fact be dequantized by an efficient classical simulation. As a consequence, the same mechanism that allows \(1\)-step CV-QAOA to solve the RDW also yields an \(O(\poly(d))\)-time classical algorithm for that problem. 

The dequantization is based on a quadratic truncation of the cost unitary in the CV-QAOA ansatz. Let
\begin{align}
    \tilde{C}(\xv)&=\sum_{\ell_1+ \dots+\ell_d\le 2}c(\ell_1, \dots, \ell_d)x_1^{\ell_1}\dots x_d^{\ell_d}
    \qquad \text{and} \qquad
    &\ket{\b,\g,\s\rangle\!}&=e^{-i\b B}e^{-i\g \tilde{C}}\ket{\s},
\end{align}
where the double-ket indicates evolution under the truncated cost operator \(\tilde{C}\) rather than the full cost \(C\).
Since both \(\tilde{C}\) and \(B\) are quadratic Hamiltonians, the output state \(\ket{\b,\g,\s\rangle\!}\) is Gaussian and therefore efficiently classically simulable. In the point-like regime, the full CV-QAOA state and its cost expectation can both be approximated to arbitrarily small error by this quadratic truncation. More precisely, we have the following result; see Appendix~\ref{app:PfPropDequant} for the proof.

\begin{prop}[Dequantizing 1-step CV-QAOA with point-like initial state]\label{prop:dequant}
For any parameters $(\beta,\gamma,\sigma)$ in the point like regime of $\sigma \to 0$ and $\beta/\sigma\to 0$, we have
\begin{align*}
    \big\|\ket{\param} - |\param\rangle\!\rangle\big\| &\le  O\left(|\g| (d\s)^3\right),\\
    \big|\braket{\param|C|\param}-\langle\!\langle\param|C|\param\rangle\!\rangle\big| &\le O\left(d\sigma^2 + d\beta^2/\sigma^2\right).
\end{align*}
\end{prop}
\begin{rem}[Query cost of dequantized CV-QAOA]
\label{rem:dequant-query-cost}\normalfont
The quadratic approximation in Proposition~\ref{prop:dequant}
is determined by $\nabla C(0)$ and $\nabla^2 C(0)$, up to an
irrelevant additive constant. These derivatives can be estimated
by central differences using $2d^2+1$ value queries.
Suppose $\{\vect{e}_i\}_{i=1}^d$ denotes the Cartesian basis. For a step size $h>0$, it suffices to query
$0$, $\pm h\vect{e}_i$, and
$h(s\vect{e}_i+t\vect{e}_j)$ for $i<j$ and
$s,t\in\{-1,1\}$. For analytic $C$, choosing the step size and oracle precision according to local derivative bounds attains any prescribed coefficient accuracy. This gives an $O(d^2)$ query upper bound for constructing the quadratic approximation for generic real-anlaytic objectives. In particular, when the objective function is a rotated real-analytic function $C_R$, the hessian is $\nabla^2 C_R (\vect{0}) = w'(0)\Id$. Thus, only the gradient $\nabla C_R(\vect{0})$ needs to be estimated, which brings the query cost down to $O(d)$.
\end{rem}

Our analysis shows that $1$-step CV-QAOA with a point-like initial state can solve a broad class of nonconvex problems, but that this regime admits efficient classical simulation. Its success depends on the alignment of the initial gradient with a global minimizer. Problems that violate this condition therefore require us to go beyond the point-like, one-step setting, either through greater circuit depth or a delocalized initial state. Our dequantization argument does not extend to these regimes, which we investigate through the rotated square well problem in Section~\ref{sec:RSW} and numerical studies in Section~\ref{sec:Numerics}.

\section{The Rotated Square Well problem}
\label{sec:RSW}
In this section, we analyze the $d$-dimensional rotated square well problem, whose landscape is flat almost everywhere, making it a candidate for classical hardness that evades the gradient-based attacks discussed above. We first study an adiabatic evolution with an analytic schedule chosen to suppress the adiabatic error exponentially, allowing us to establish a polylogarithmic runtime $T=O(\log^8 d)$ for fixed target error. We then bound the number of phase queries needed to simulate this evolution using CV-QAOA with a harmonic mixer, obtaining a subpolynomial query complexity of $d^{o(1)}$. Together with the classical lower bound established in Section~\ref{sec:ClassicalComplexity}, this gives a polynomial-versus-subpolynomial query separation between our quantum algorithm and classical algorithms that reconstruct the hidden rotated structure. 

Define
\begin{align}
    w(z)&=-\mathbf 1_{\{|z|<\delta\}},
    &
    W(\xv)\equiv W_{R,\sv}(\xv)&=\sum_{j=1}^d w((R\xv+\sv)_j),
\end{align}
where $\delta>0$ is fixed independently of the dimension and target accuracy, $R\in\mathbf O(d)$ is an unknown rotation, and $\sv$ is the unknown shift introduced in \eqref{eq:rot-sq-well}, with $\|\sv\|_\infty\le B$ for a fixed $B>0$. We bound the probability of reaching $V_a=\{\xv:\|R\xv+\sv\|_\infty\le a\}$, where $a>\delta$ is also fixed independently of the dimension and target accuracy. By sampling sufficiently many points independently from the final-state probability distribution, a concentration argument shows that the mean of the samples lies in the intersection of the $d$ wells with failure probability exponentially small in the number of samples (Proposition~\ref{prop:rsw-averaging}).

We first study the evolution on $\RR^d$ under
\begin{align}
    H(s)&=(1-s)A+sW,
    &
    A&=-\frac12\Delta+\frac{|\xv|^2}{8\sigma^4},
\end{align}
where $\sigma>0$ is the initial Gaussian width, fixed independently of the dimension and target accuracy. We take $0\le s\le s_f=1-\nu^2$, with $0<\nu<1$, to keep the spectral gap from vanishing at the end of the evolution. The initial state is the known Gaussian ground state of $A$,
\begin{equation}
    \braket{\xv|\sigma}
    =\frac{1}{(2\pi\sigma^2)^{d/4}}
    e^{-|\xv|^2/(4\sigma^2)}.
\end{equation}
Under the change of coordinates $\zv=R\xv+\sv$, the Laplacian is unchanged and $|\xv|^2=|\zv-\sv|^2$. The transformed Hamiltonian therefore decomposes exactly into one-dimensional factors
\begin{align}
    h_j(s)&=(1-s)A_j+sw(z),
    &
    A_j&=-\frac12\partial_z^2+\frac{(z-s_j)^2}{8\sigma^4}.
\end{align}
The initial Gaussian becomes the product of the ground states $\phi_j(z)=(2\pi\sigma^2)^{-1/4}e^{-(z-s_j)^2/(4\sigma^2)}$. Although the factors generally have different shifts $s_j$, they follow the same schedule, and their evolutions remain independent. Let $\psi_j(z)$ be the final wavefunction of the $j$th one-dimensional evolution, initialized in $\phi_j$, and let $\PP_d(T)$ denote the probability that the final position lies in $V_a$. The product structure gives
\begin{equation}
    \PP_d(T)
    =\prod_{j=1}^d\int_{-a}^a|\psi_j(z)|^2\,dz.
\end{equation}
It is thus enough to bound the failure probability of each factor uniformly over the allowed shifts $|s_j|\le B$. Smoothing is treated by comparison with this infinite-volume evolution in Appendix~\ref{app:QAA}.

We use the exponentially accurate adiabatic theorem of Hagedorn and Joye~\cite{HagedornJoye2002}, which bounds the adiabatic error schematically as $\varepsilon_{\mathrm{ad}}(T)\le K e^{-\Gamma T}$ for constants $K,\Gamma>0$ that depend on the dimension $d$ and the spectral gap (see Theorem~\ref{thm:HJ-original} for details and Lemma~\ref{lem:gamma-nu-M} for bounds on these constants). The theorem requires an analytic schedule and a nondegenerate ground state separated from the rest of the spectrum by a uniform positive gap.
\begin{assumption}\label{HJ-assumptions}
The operators $\widehat H(u)$, $u\in[0,1]$, are self-adjoint on a common dense domain $D\subset\mathcal H$ and extend to an analytic family of type A on a complex neighborhood of $[0,1]$. For every $u\in[0,1]$, the eigenvalue $E(u)$ is simple, and there is a constant $d_0>0$ such that
\begin{equation}
    \operatorname{dist}\left(
    E(u),\operatorname{spec}(\widehat H(u))\setminus\{E(u)\}
    \right)\ge d_0
    \qquad (0\le u\le1).
\end{equation}
\end{assumption}
The cutoff $s_f=1-\nu^2$ keeps the spectral gap bounded below by
$\gamma_*\nu^2>0$ throughout the evolution, where $\gamma_*>0$ is
independent of the allowed shifts and the dimension
(see Lemma~\ref{lem:gap-all-s}). As explained in \cite{HagedornJoye2002}, the initial state in the superadiabatic approximation is not guaranteed to coincide with the ground state of the initial mixer Hamiltonian. To achieve this agreement and apply the exponentially accurate adiabatic theorem, we choose a schedule whose derivatives vanish at the endpoints (see Lemma~\ref{lem:endpoint-flat-schedule}). Specifically, writing physical time as $t=Tu$, we take
\begin{equation}
    s_M(u)=s_f\,
    \frac{\int_0^u v^M(1-v)^M\,dv}
         {\int_0^1 v^M(1-v)^M\,dv},
    \qquad 0\le u\le1.
\end{equation}
This polynomial is increasing on $[0,1]$, satisfies $s_M(0)=0$ and $s_M(1)=s_f$, and has $s_M^{(m)}(0)=s_M^{(m)}(1)=0$ for $1\le m\le M$. Taking $M\ge N+1$, where $N$ is the truncation order of the superadiabatic approximation, ensures that its initial state is exactly our Gaussian and its final state is proportional to the final ground state. The adiabatic theorem, therefore, bounds the error in preparing the final ground state, up to phase.

\begin{thm}[Rotated square well runtime]
\label{thm:qaa_log_runtime}
For every target failure probability $\varepsilon_{\mathrm{tar}}\in(0,1/3)$, the evolution under $H(s_M(t/T))$, initialized in $\ket{\sigma}$, admits a schedule and runtime independent of $R$ and $\sv$ such that
\begin{equation}
    \PP_d(T)\ge1-\varepsilon_{\mathrm{tar}},
    \qquad
    T=O\left(\log^8\frac{d}{\varepsilon_{\mathrm{tar}}}\right).
\end{equation}
\end{thm}
The success probability is controlled by two errors in each coordinate: leakage of the final ground state outside $[-a,a]$, and the adiabatic error in preparing that state. Together, these give $$\PP_d(T)\ge[1-K_{\mathrm{loc}}\nu e^{-c(a-\delta)/\nu}-K_{\mathrm{ad}}e^{-\Gamma T}]^d\ge1-d[K_{\mathrm{loc}}\nu e^{-c(a-\delta)/\nu}+K_{\mathrm{ad}}e^{-\Gamma T}],$$ where $K_{\mathrm{loc}},c>0$ are fixed constants, while $K_{\mathrm{ad}}$ and $\Gamma$ depend on the cutoff and schedule. It therefore suffices to make each error term $O(\varepsilon_{\mathrm{tar}}/d)$. Choosing $\nu=\Theta((a-\delta)/\log(d/\varepsilon_{\mathrm{tar}}))$ controls the leakage, while the analytic schedule suppresses the adiabatic error sufficiently at runtime $T=O(\log^8(d/\varepsilon_{\mathrm{tar}}))$. The complete parameter choices and bounds are given in Theorem~\ref{thm:rsw-harmonic-runtime}. We now translate the runtime into a query bound using the one-dimensional decomposition and higher-order product formulas. The query analysis is given in Theorem~\ref{thm:rsw-subpolynomial-queries}.
\begin{cor}[Query complexity for simulation via CV-QAOA]
\label{cor:queryComplexity}
\label{cor:rsw-query}
For every fixed $\varepsilon_{\mathrm{tar}}\in(0,1/3)$, a harmonic-mixer product formula initialized in $\ket{\sigma}$ reaches $V_a$ with probability at least $1-\varepsilon_{\mathrm{tar}}$, using
\begin{equation}
    Q(d,\varepsilon_{\mathrm{tar}})=d^{o(1)}
\end{equation}
phase queries to the smoothed potential $W^{\eta}$ in the oracle model of Definition~\ref{def:query}, for a suitable smoothing width and schedule.
\end{cor}
\noindent
To be more precise, the $Q = d^{o(1)}$ scaling here means $\lim_{d\to\infty} \log Q / \log d =0$.

\section{Complexity of classical algorithms}\label{sec:ClassicalComplexity}

In this section, we complement the results on CV-QAOA by analyzing the performance of classical algorithms on the corresponding optimization problems.
For convex quadratics, we recall the complexity of standard classical optimization methods and compare them with the gate complexity of CV-QAOA. For rotated and shifted real-analytic problems, we derive an information-theoretic $\Omega(d/\log d)$ lower bound on classical query complexity, clarifying what can and cannot be inferred from oracle access alone. Finally, focusing on the RSW problem, we analyze two natural families of classical algorithms: (1) an algorithm that attempts to learn the hidden rotation to recover the global minimum, with $\Theta(d^2\log d)$ query complexity, and (2) an adaptive random walk algorithm with $\Theta(d^3/\log d)$ hitting time and query complexity.

\subsection{Convex quadratic functions}

We begin with the classical complexity of optimizing convex quadratic functions. Consider a cost function of the form
\begin{equation}
    C(\xv) = \frac12 \xv^\top Q \xv + \bv^\top \xv,
\end{equation}
where \(Q\) is positive definite. Let
\begin{equation}
    \kappa = \|Q\|\,\|Q^{-1}\|
\end{equation}
denote the condition number of the Hessian, and let \(\xv_*=\arg\min_{\xv} C(\xv)\). We ask for the computational cost of producing an approximate minimizer \(\tilde{\xv}\) satisfying
\begin{equation}
    |C(\tilde{\xv})-C(\xv_*)|\le \epsilon.
\end{equation}
For standard classical methods, the required complexity is well known:
\begin{enumerate}
    \item Gradient descent: \(O(\kappa \log (\frac 1\eps))\) iterations,
    corresponding to \(O(\kappa d^2 \log (\frac 1\eps))\) floating-point
    operations (flops)~\cite{nesterov2018lectures,boyd2004convex}.
    \item Conjugate gradient: \(O(\sqrt{\kappa}\log \frac 1\eps))\) iterations,
    corresponding to \(O(\sqrt{\kappa}\, d^2 \log (\frac 1\eps))\)
    flops~\cite{trefethen1997numerical,golub2013matrix}.
    \item Newton's method: \(O(d^3)\) flops~\cite{boyd2004convex,golub2013matrix}.
\end{enumerate}
By comparison, we recall that \(d\)-step CV-QAOA finds the minimum of any \(d\)-dimensional strictly convex quadratic using a total of \(O(d^3)\) gates at optimal parameters, which can be obtained by knowing the eigenvalues of the Hessian. 

\subsection{Rotated and shifted real-analytic functions}
\label{sec:classical-RDW}

To obtain a concrete lower bound for classical algorithms on the rotated real-analytic problem, it is natural to study query complexity. The point of this subsection is to clarify what can be inferred from oracle access alone. We show below that any classical algorithm must make at least \(\Omega(d/\log d)\) queries in order to identify the global minimizer. Thus, even before one addresses runtime, there is already a nontrivial information-theoretic obstruction to solving these problems from black-box access to the cost function. We consider the family of instances
\begin{equation} \label{eq:RSW}
    C_{R, \vect{s}}(\xv) = \sum_{j=1}^d w((R\xv)_j+s_j),
\end{equation}
where \(w=w(x)\) is an arbitrary one-dimensional real-analytic cost function.

A particularly important example is obtained when \(w(x)=x^4+ax^3+bx^2+cx\) is a double-well quartic function. For many choices of \((a,b,c)\), for instance \((0,-16,5)\) as in the Styblinski-Tang function, the function \(w\) has two non-degenerate local minima. The corresponding high-dimensional function \(C_{R,\sv}(\xv)\) then has \(2^d\) local minima, only one of which is globally optimal. This proliferation of suboptimal minima makes the problem difficult for conventional classical methods, as was observed empirically in \cite{leng2023nonconvex}. It was further argued in \cite{leng2023nonconvex} that the difficulty of the problem stems from the hidden rotation \(R\): even when the functional form of \(C_{\Id,0}\) is known, recovering the relevant rotated coordinates is highly nontrivial. Indeed, learning \(R\) from the closed-form expression of \(C_{R,0}\) is equivalent to diagonalizing a fourth-order tensor, which is NP-hard in the worst case~\cite{hillar2013tensor}. The following proposition quantifies this obstruction at the level of query complexity.

\begin{prop}\label{prop:query}
Suppose $w$ is known and has a unique global minimizer $x_0\ne0$.
Fix $0<\epsilon<|x_0|$, and suppose each value query returns
a $b$-bit answer. Any classical algorithm using at most $q$
queries to $C_{R,0}$ and locating its minimizer within $\epsilon$
in each coordinate with constant success probability for every
unknown $R\in O(d)$ must satisfy $q=\Omega(d/b)$.
In particular, $b=O(\log d)$ implies $q=\Omega(d/\log d)$.
\end{prop}

\begin{proof}
The algorithm must, in particular, succeed when
$R=\operatorname{diag}(v)$ for any $v\in\{-1,1\}^d$.
The corresponding minimizers $x_0v$ are pairwise separated
by $2|x_0|>2\epsilon$ in the infinity norm, so an accurate
output uniquely identifies $v$. For each fixed choice of
the algorithm's randomness, its $q$ answers admit at most
$2^{bq}$ transcripts, each permitting success on at most
one $v$. Under uniform $v$, its success probability is
therefore at most $2^{bq-d}$. Averaging over the algorithm's
randomness preserves this bound. Writing $p>0$ for the
constant lower bound on its success probability, we obtain
$2^{bq-d}\ge p$, hence $bq\ge d+\log_2p=d-O(1)$.
\end{proof}

\paragraph{An efficient classical algorithm when $\sv=0$.}
For unshifted RDW, the alignment of the gradient at the origin with a global
minimizer yields a linear-query classical algorithm.

\begin{prop}\label{prop:RDW-O(d)-bound}
Suppose the quartic profile $w$ and $x_0\in\arg\min w$ are known, with
$w'(0)\ne0$. Given value queries to $C_{R,0}$ of absolute error at most
$\eta$, a deterministic algorithm makes fewer than $8d$ queries, all in
$[-1,1]^d$, and outputs $\widehat\xv$ satisfying
\[
 \|\widehat\xv-\xv_*\|_2
 \le 3\left|\frac{x_0}{w'(0)}\right|\eta.
\]
Thus, provided $|w'(0)|\ge 3|x_0|\eta/\varepsilon$, the algorithm
achieves accuracy $\varepsilon$ with oracle precision independent
of $d$ for a fixed profile.
\end{prop}

\begin{proof}
Set $\vect{g}=\nabla C_{R,0}(0)$, so
$\xv_*=(x_0/w'(0))\vect{g}$.
Let $N=2^{\lceil\log_2 d\rceil}<2d$ and let $V$ consist of the
first $d$ columns of an $N\times N$ matrix with entries $\pm1$
and mutually orthogonal columns. Then $V^\top V=NI_d$.
Each $q_k(t):=C_{R,0}(t\vect{v}_k)$ is quartic, so
\[
 q_k'(0)=\frac{8[q_k(1/2)-q_k(-1/2)]-[q_k(1)-q_k(-1)]}{6}.
\]
Replacing these four values by oracle answers gives an estimate $y_k$ of
$\vect{v}_k^\top\vect{g}$ with error at most $3\eta$.
Hence $\vect{y}=V\vect{g}+\vect{e}$ with
$\|\vect{e}\|_2\le3\eta\sqrt N$, and
\[
 \widehat\xv:=\frac{x_0}{w'(0)}\frac{V^\top\vect{y}}{N}
 \quad\Longrightarrow\quad
 \|\widehat\xv-\xv_*\|_2
 \le\left|\frac{x_0}{w'(0)}\right|
       \frac{\|\vect{e}\|_2}{\sqrt N}
 \le3\left|\frac{x_0}{w'(0)}\right|\eta.
\]
All $4N<8d$ query points $\pm\vect{v}_k/2,\pm\vect{v}_k$ lie in
$[-1,1]^d$.
\end{proof}

When $\sv\ne0$, the gradient
$\nabla C_{R,\sv}(0)=R^\top w'(\sv)$ need not be aligned with a global
minimizer. Nevertheless, the Hessian
\[
 \nabla^2 C_{R,\sv}(\xv)
 =R^\top\diag\bigl(w''(R\xv+\sv)\bigr)R
\]
retains the hidden coordinate structure. Hessian-based recovery of this
structure provides a different classical approach, reducing the problem to
one-dimensional optimization; see~\cite{herman2025mechanisms}.

\subsection{Rotated Square Well}

We now turn to the Rotated Square Well problem, which is as our main candidate for a classically hard instance. Unlike the rotated real-analytic examples discussed above, this landscape does not reveal useful global information through its local gradients. Away from the boundary of the well, the function is locally constant, so gradient-based information is essentially absent. For this reason, the dequantization mechanism available in the point-like 1-step CV-QAOA regime no longer applies, and it is natural to compare the quantum algorithm against more geometric classical search strategies.

Consistent with the ground-state (minimization) formulation used in the quantum analysis of Section~\ref{sec:RSW}, we treat the rotated square well as a minimization problem throughout. Consider the one-dimensional well
\begin{align}
    w(x) = \begin{cases}
	    -1, & \text{if } -\delta < x < \delta,\\
        0, & \text{otherwise}.
    \end{cases}
\end{align}
From this, define the \(d\)-dimensional function \(W_{R,s}:\RR^d\to\RR\) by choosing an arbitrary rotation \(R\in\mathbf{O}(d)\) and shift \(s\in\RR^d\), and setting
\begin{align}
    W_{R,s}(\xv) = \sum_{j=1}^d w\Big((R\xv+\sv)_j\Big).
\end{align}
By construction, \(\min W_{R,s}=-d\) for every instance \((R,s)\). The classical algorithm is assumed to have only black-box access to the landscape: it does not know the hidden rotation \(R\) or shift \(s\), but it does know the basic structure of the problem, namely the value of \(\delta\) and the fact that the landscape is built from nested slab-like regions of constant height. To formalize this, define
\begin{align} \label{eq:slab-def}
    V_k \coloneqq \{x\in\RR^d \mid {W_{R,s}(x)\le -k}\},
\end{align}
as the $k$-th slab region, for $k=0,\ldots, d$. Observe that
\begin{align}
    V_0 \supseteq V_1 \supseteq \cdots \supseteq V_d.
\end{align}
The task is to reach the innermost set \(V_d\), where the minimum value \(-d\) is attained.

In the quantum analysis of Section~\ref{sec:RSW}, success is defined as landing in the slightly larger region \(\{\xv:|a_j\cdot\xv+s_j|\le a\ \forall j\}\) for a fixed \(a\in(\delta,L)\): the final one-dimensional bound state carries a \(\Theta(\nu^3)\) fraction of its mass outside the well (Lemma~\ref{lem:deep-well}), so demanding \(\xv\in V_d\) exactly would reduce the quantum success probability to \(e^{-\Theta(d\nu^3)}=e^{-d/\mathrm{polylog}(d)}\) at the choice \(\nu=\Theta(1/\log d)\) of Theorem~\ref{thm:qaa_log_runtime}. The classical algorithms below are held to the stronger criterion of reaching \(V_d\) itself, which they meet, while the classical lower bound of Proposition~\ref{prop:rsw-info-lower-bound} concerns reconstructing the hidden arrangement --- the strategy employed by both classical algorithms we analyze. Since \(\delta<a\) are fixed constants independent of \(d\), this asymmetry does not affect the polynomial-versus-subpolynomial character of the comparison.

\subsubsection{Algorithms that learn the rotation}
\label{subsubsec:learn-rotation}

A natural classical strategy for the rotated square well problem is to first learn the hidden slab arrangement. This is quite different from gradient-based optimization: away from the slab boundaries the function is locally constant, so local differential information carries essentially no global signal. Nevertheless, the boundaries of the level sets encode the hidden rotation and shift. In this subsection we discuss two complementary facts. First, reconstructing the hidden arrangement from oracle queries has an information-theoretic cost of \(\Omega(d^2/\log d)\) queries --- for the sharp aggregate-count oracle, and equally for finite-precision access to the smoothed well queried by the quantum algorithm of Corollary~\ref{cor:queryComplexity}. Second, there is a concrete random-ray tomography algorithm which succeeds despite the flatness of the landscape, using \(\Theta(d^2\log d)\) oracle queries.

Throughout this subsection, let \(a_1^\top,\ldots,a_d^\top\) denote the rows of \(R\), so that
\[
W_{R,s}(x)
=
{-}\sum_{j=1}^d
\mathbf 1\{|a_j\cdot x+s_j|<\delta\}.
\]
We assume the shifts obey \(\|s\|_\infty\le B\), for some known \(B>0\). The innermost region \(V_d\) is the rotated cube
\[
V_d
=
\bigcap_{j=1}^d
\{x: |a_j\cdot x+s_j|<\delta\}.
\]
Thus learning the boundary hyperplanes of these slabs is sufficient to recover a point in \(V_d\).

\paragraph{Information-theoretic lower bound.}

We first prove a lower bound for algorithms that optimize by reconstructing the hidden slab arrangement to constant accuracy. The row ordering and row signs of \(R\) are not identifiable from the aggregate oracle: signed permutations of the rows of \(R\), with the corresponding signed permutation of \(s\), define the same unordered family of slabs and hence the same function \(W_{R,s}\). Therefore the identifiable rotation data naturally lives in the quotient \(O(d)/G_d\), where \(G_d\) is the signed permutation group.
\begin{prop}[Information-theoretic lower bound for reconstruction]
\label{prop:rsw-info-lower-bound}
Fix a sufficiently small constant accuracy $\varepsilon>0$.
Any randomized adaptive classical algorithm that reconstructs
the hidden slab arrangement of $W_{R,s}$ to accuracy
$\varepsilon$, with constant success probability and $b$-bit
oracle answers, requires
\[
q=\Omega\left(\frac{d^2}{b}\right)
\]
queries. For the aggregate-count oracle, whose responses lie
in $\{0,-1,\ldots,-d\}$, one may take
$b=\lceil\log_2(d+1)\rceil$. For responses in $[-d,0]$
rounded to resolution $\varepsilon_{\mathrm{prec}}$, one may
take $b=O(\log(2+d/\varepsilon_{\mathrm{prec}}))$.
Thus $\varepsilon_{\mathrm{prec}}\ge1/\poly(d)$ yields
$q=\Omega(d^2/\log d)$, including for the smoothed well used
in Corollary~\ref{cor:queryComplexity}.
\end{prop}

\begin{proof}
Fix $s=0$. The slab arrangement determines $R$ up to signed
permutations, so the identifiable parameter lies in $O(d)/G_d$,
where $|G_d|=2^d d!$. A standard volume-packing bound gives
a set of pairwise distances greater than $2\varepsilon$
with cardinality
$M\ge(c/\varepsilon)^{d(d-1)/2}/(2^d d!)$,
hence $\log M=\Omega(d^2)$ for sufficiently small fixed
$\varepsilon$. Choose the hidden instance uniformly from this packing.
For each fixed choice of the algorithm's randomness, its
$q$ queries admit at most $2^{bq}$ transcripts, each permitting
an $\varepsilon$-accurate reconstruction of at most one packing
element. Its success probability is therefore at most
$2^{bq}/M$. Averaging over the algorithm's randomness preserves
this bound. Writing $p>0$ for the constant lower bound on its
success probability, we obtain
\[
q\ge\frac{\log_2(pM)}{b}
=\Omega\left(\frac{d^2}{b}\right).
\]
\end{proof}

\paragraph{A concrete algorithm: random-ray tomography.}
The minimum region $V_d$ is bounded by the hyperplanes
$H_{j,\tau}=\{x:a_j\cdot x+s_j=\tau\delta\}$,
$j=1,\ldots,d$, $\tau\in\{+1,-1\}$.
Pairing the two parallel sheets of each slab determines its center
hyperplane $a_j\cdot x+s_j=0$. Solving $Rx+s=0$ then gives
$x_c=-R^\top s\in V_d$. Random-ray tomography identifies these hyperplanes by querying along
affine lines $\ell(t)=r+tv$, with
$v\sim\mathrm{Unif}(S^{d-1})$ and $\|r\|\le\rho$.
The restricted oracle is a step function whose jumps occur at
\[
t_{j,\tau}
=
\frac{\tau\delta-s_j-a_j\cdot r}{a_j\cdot v}.
\]
The numerator is bounded in magnitude by $\delta+B+\rho$, while
$|a_j\cdot v|$ is typically of order $d^{-1/2}$.
We therefore choose the scan radius
$T=C_T(\delta+B+\rho)\sqrt d$.
For sufficiently large $C_T$, each fixed hyperplane is crossed
within $[-T,T]$ with constant probability
(Appendix~\ref{app:random-ray-tomography}).

Algorithm~\ref{alg:random-ray-tomography} first records approximate
boundary points from these scans, then estimates their normals
to reconstruct the slab centers.
The reconstruction samples
$N_{\mathrm{pt}}=\lceil C_{\mathrm{pt}}d\log d\rceil$
recorded points and uses
$N_{\mathrm{dir}}=\lceil C_{\mathrm{dir}}d\rceil$
random probe directions per point, with probe radius $\eta>0$
and fixed constants $C_{\mathrm{pt}},C_{\mathrm{dir}}>0$.

\begin{algorithm}[H]
\caption{Random-ray tomography}
\label{alg:random-ray-tomography}
\begin{algorithmic}[1]
\REQUIRE Oracle $W_{R,s}$; dimension $d$; widths $\delta,B,\rho$;
probe radius $\eta$; budgets $N_{\mathrm{pt}},N_{\mathrm{dir}}$
\ENSURE Estimated center $\widehat x$ and its objective value

\STATE Set $M=\lceil C_Md\rceil$,
$T=C_T(\delta+B+\rho)\sqrt d$,
$\Delta t=c_\Delta/\sqrt d$; initialize $\mathcal P=\varnothing$.

\FOR{$m=1,\ldots,M$}
    \STATE Sample $v_m\sim\mathrm{Unif}(S^{d-1})$ and
    $\|r_m\|\le\rho$; set $\ell_m(t)=r_m+tv_m$.
    \STATE Query $\ell_m$ on the grid
    $t_q=-T+q\Delta t$, $q=0,\ldots,\lceil2T/\Delta t\rceil$.
    \STATE Append to $\mathcal P$ the midpoint of each grid
    interval whose oracle values differ by exactly one.
\ENDFOR

\FOR{$k=1,\ldots,N_{\mathrm{pt}}$}
    \STATE Sample $p_k$ uniformly from $\mathcal P$ and independent
    $z_i\sim\mathrm{Unif}(S^{d-1})$, $i=1,\ldots,N_{\mathrm{dir}}$.
    \STATE Query the symmetric probes and estimate
    \[
    \widehat n_k\propto
    \sum_{i=1}^{N_{\mathrm{dir}}}
    \bigl[W_{R,s}(p_k+\eta z_i)-W_{R,s}(p_k-\eta z_i)\bigr]z_i.
    \]
\ENDFOR

\STATE Cluster  normals up to sign into $d$ slab directions,
pair boundary sheets, and form midpoint hyperplanes.
\STATE Solve the resulting system
$\widehat R\widehat x+\widehat s=0$.
\RETURN $(\widehat x,W_{R,s}(\widehat x))$.
\end{algorithmic}
\end{algorithm}
Two facts make the query count transparent. First, \(M=\Theta(d)\) lines suffice for the scan: each fixed hyperplane is hit by each random line with constant probability, so after \(M=\Theta(d)\) independent lines every one of the \(2d\) hyperplanes carries between \(d\) and \(M\) recorded boundary points with high probability, which balances the sampling probabilities in the reconstruction stage. Second, the reconstruction must estimate a normal from \emph{every} one of the \(2d\) hyperplanes, and each drawn point lands on a given hyperplane with probability \(\Theta(1/d)\); covering all \(2d\) of them is a coupon-collector problem, which forces \(N_{\mathrm{pt}}=\Theta(d\log d)\). Both facts are proved in Appendix~\ref{app:random-ray-tomography}.
\paragraph{\(\Theta(d^2\log d)\) query complexity.}
We now summarize the query count of the random-ray tomography algorithm. The count is fixed by the construction: it is determined by the scan and probe budgets, independent of the geometric postprocessing.
\begin{prop}[Query complexity of random-ray tomography]
\label{prop:rsw-random-ray-query-complexity}
Run random-ray tomography with a line scan of \(M=\lceil C_M d\rceil\) random lines followed by the reconstruction stage, which probes \(N_{\mathrm{pt}}=\lceil C_{\mathrm{pt}}\,d\log d\rceil\) of the recorded boundary points, using \(N_{\mathrm{dir}}=\lceil C_{\mathrm{dir}}\,d\rceil\) probe directions at each. The algorithm makes
\[
Q_{\mathrm{ray}}
=
Q_{\mathrm{scan}}+Q_{\mathrm{rec}}+1
=
\Theta(d^2\log d),
\]
with \(Q_{\mathrm{scan}}=\Theta(d^2)\) grid queries along the \(M\) scanned lines, \(Q_{\mathrm{rec}}=2N_{\mathrm{pt}}N_{\mathrm{dir}}=\Theta(d^2\log d)\) probe queries in the reconstruction stage, and a single verification query. With probability \(1-O(1/d)\) the algorithm returns a point \(\widehat x\) with \(W_{R,s}(\widehat x)=-d\). The \(\log d\) factor is necessary: covering all \(2d\) boundary hyperplanes requires \(N_{\mathrm{pt}}=\Omega(d\log d)\), so \(Q_{\mathrm{ray}}=\Omega(d^2\log d)\).
\end{prop}
\begin{proof}
The scan issues \(M=\lceil C_M d\rceil=\Theta(d)\) lines (see Appendix~\ref{app:random-ray-tomography}), each queried at
\[
Q_{\mathrm{line}}
=
\left\lceil \frac{2T}{\Delta t}\right\rceil+1
\]
grid points. With \(T=C_T(\delta+B+\rho)\sqrt d\) and \(\Delta t=c_\Delta/\sqrt d\),
\[
Q_{\mathrm{line}}
=
\left\lceil
\frac{2C_T(\delta+B+\rho)}{c_\Delta}d
\right\rceil+1
=
\Theta(d),
\qquad\text{so}\qquad
Q_{\mathrm{scan}}
=
M Q_{\mathrm{line}}
=
\Theta(d^2).
\]
The reconstruction stage evaluates \(2N_{\mathrm{dir}}\) probe points at each of \(N_{\mathrm{pt}}\) sampled points, so
\[
Q_{\mathrm{rec}}
=
2N_{\mathrm{pt}}N_{\mathrm{dir}}
=
2\lceil C_{\mathrm{pt}}d\log d\rceil\lceil C_{\mathrm{dir}}d\rceil
=
\Theta(d^2\log d).
\]
Adding the single verification query, \(Q_{\mathrm{ray}}=\Theta(d^2)+\Theta(d^2\log d)+1=\Theta(d^2\log d)\). Correctness holds with probability \(1-o(1)\) by the coverage bound of Appendix~\ref{app:random-ray-tomography}, which also shows \(N_{\mathrm{pt}}=\Omega(d\log d)\) is necessary, giving the matching lower bound \(Q_{\mathrm{ray}}=\Omega(d^2\log d)\).
\end{proof}
Proposition~\ref{prop:rsw-random-ray-query-complexity} and Corollary~\ref{cor:queryComplexity} exhibit a superpolynomial separation in query complexity: random-ray tomography requires $\Theta(d^2\log d)$ queries, while CV-QAOA requires only $d^{o(1)}$ queries to the cost function. The quantum algorithm exploits the hidden product structure without reconstructing the rotation and shift.

\subsubsection{Algorithm without learning the rotation: Adaptive random walk}
\label{subsubsec:do-not-learn-rotation}
We next consider a different class of classical algorithms that do not attempt to reconstruct the hidden rotation or slab arrangement, but instead perform random walk trying to improve the objective function. 
Recall \eqref{eq:slab-def} where we defined $ V_k = \{x\in\RR^d \mid {W_{R,s}(x)\le -k}\}$ be the $k$-th slab.
We consider the following algorithm tries to climb the filtration of slabs
\(V_0\supseteq V_1\supseteq \cdots \supseteq V_d\), by moving from level \(k\) to level \(k+1\) with local moves, without explicitly learning the orientation of the slabs.

\paragraph{The algorithm.}
We consider the following monotone adaptive random walk. At time \(t\), let
\[
K_t:=-W_{R,s}(x_t)
\]
be the current level (the number of slabs containing $x_t$). The algorithm samples \(U_{t+1}\sim \mathrm{Unif}(\Sph^{d-1})\), proposes
\[
y_{t+1}=x_t+r_{K_t}U_{t+1},
\]
queries \(W_{R,s}(y_{t+1})\), and accepts the move only if it does not decrease the level (and thus the objective value is monotonically improving):
\[
x_{t+1}
=
\begin{cases}
y_{t+1}, & \text{ if } {W_{R,s}(y_{t+1})\le W_{R,s}(x_t)},\\
x_t, & \text{ if } {W_{R,s}(y_{t+1})> W_{R,s}(x_t)}.
\end{cases}
\]
Thus \(K_{t+1}\ge K_t\). The stopping time is
\[
\tau:=\inf\{t\ge 0:K_t=d\}.
\]
We adopt this monotonic strategy because allowing worse intermediate steps does not obviously improve convergence; a step with $K_{t+1}<K_t$ can land us in a level set whose volume is exponentially larger, i.e. $\mathrm{Vol}(V_{K_{t+1}})/\mathrm{Vol}(V_{K_t}) \approx (L/\d)^{K_{t} - K_{t+1}}$. Hence, we restrict to a family of greedy random walks that only accept steps which do not decrease the level.

The key algorithmic parameter is the step size. We choose it to retain the already-captured slab coordinates with probability bounded below by a constant. The analysis in Appendix~\ref{sec:RandomWalk} uses the \emph{typical-inside model}: at level $k$, the $k$ captured coordinates of $Rx+s$ are independently uniform on $[-\delta,\delta]$ and independent of the proposal direction $U\sim\mathrm{Unif}(\Sph^{d-1})$. Under this assumption, for $k\ge C\sqrt d$ with $C$ sufficiently large, the largest step-size scale compatible with constant retention probability is
\begin{equation}\label{optimal-step}
r_k=\Theta\left(\delta\frac{\sqrt d}{k}\right).
\end{equation}
A suitable choice on this scale therefore also keeps $\PP(x+r_kU\in V_k)$ bounded below by a constant (Proposition~\ref{lem:Vk_accept} and
Corollary~\ref{cor:rk_bottleneck}). The necessity statement concerns retention of all captured slabs; it does not exclude larger steps that preserve the level by exchanging occupied slabs.

\paragraph{$\Theta(d^3/\log d)$ hitting time.}

Even with this optimized step-size schedule the walk is slow, and it is slowest exactly at the end. At level \(k\), the algorithm has already entered \(k\) slabs and must find a way to enter one more. A random direction on \(\Sph^{d-1}\) has only an \(O(d^{-1/2})\)-sized component along any fixed transverse direction, so a \(d\)-dimensional proposal of size \(r_k\) produces an effective one-dimensional displacement of size
\(
r_{\mathrm{eff},k}
\asymp
\frac{r_k}{\sqrt d}
\asymp
\frac{\delta}{k}
\)
in each slab coordinate. Moving from \(V_k\) to \(V_{k+1}\) is therefore a one-dimensional first-passage problem with step size \(r_{\mathrm{eff},k}\), and since one-dimensional hitting times scale like the inverse square of the step size, a passage over a distance \(D\) at the terminal levels costs \(\Theta(D^2d^2/\delta^2)\) queries.

The distance that remains to be covered at those levels is itself large, and it is set by the shape of the search region. Suppose the walk is confined to the box \(Q_L=[-L,L]^d\), whose boundary the algorithm knows. Not knowing \(R\), it cannot confine an individual slab coordinate; moving \(z_j\) alone means translating \(\xv\) along the \(j\)-th row \(a_j\) of \(R\), and the box allows a travel of exactly \(2L/\|a_j\|_\infty\) in that direction. For a generic rotation \(\|a_j\|_\infty\approx\sqrt{2\log d/d}\), so a single slab coordinate has room $\Theta\!\left(L\sqrt{\frac{d}{\log d}}\right)$ to wander in (between the width \(2L\) of the box along a coordinate axis and its diameter \(2L\sqrt d\) along a diagonal). Bringing the last coordinate across a comparable distance with steps of size \(\Theta(\delta/d)\) therefore costs
\begin{equation}\label{eq:rw-query-main}
\EE[\tau]=\Theta\!\left(\frac{L^2}{\delta^2}\cdot\frac{d^3}{\log d}\right),
\end{equation}
and this single level already accounts for a constant fraction of the whole run. 
This derivation is made precise in Appendix~\ref{app:rw-complexity}. 
Numerical simulations up to $d=100$ confirms this scaling, as we shall see shortly in Figure~\ref{fig:geometry-aware} of Section~\ref{sec:numerics-RSW}.

This random-walk analysis is informative because it captures a natural class of classical algorithms that search locally without learning the hidden geometry. However, it is not the strongest classical baseline for the rotated square well. The random-ray tomography algorithm of Section~\ref{subsubsec:learn-rotation} explicitly learns the slab arrangement and improves the classical query complexity for solving the rotated square well family from \(\Theta(d^3/\log d)\) to \(\Theta(d^2\log d)\). Thus, for comparing against the quantum algorithm, random-ray tomography should be viewed as the sharper classical strategy.

\subsection{Benchmarking classical algorithms on the Rotated Square Well}
\label{sec:numerics-RSW}
We benchmark several classical optimization algorithms on the rotated square well problem. The goal is to test whether standard classical algorithms exhibit rapidly growing time-to-solution as the dimension $d$ increases. The landscape is deliberately unfavorable to generic local or gradient-based search: it is flat and piecewise constant over most of the domain, and the global well is hidden by an unknown rotation and shift. Thus the optimizer receives little useful local information until it lands near one of the narrow well regions. Consistent with this geometry, all of the classical baselines we test exhibit behavior consistent with superpolynomial scaling in $d$. These benchmarks are not lower bounds for all possible classical algorithms, but rather evidence that widely used generic classical optimization methods do not efficiently solve this hidden rotated square well problem. 

We optimize
\begin{equation}
    W_{R,\vect{s}}(\vect{x})=-\sum_{j=1}^d
    \mathbf{1}\{|(R\vect{x}+\vect{s})_j|<\delta\},\qquad\vect{x}\in[-L,L]^d.
\end{equation}
Here $R\in O(d)$ and $\vect{s}\in\RR^d$ are hidden from the optimizer, which only accesses $W_{R,\vect{s}}(\vect{x})$ through function evaluations. For each instance, we sample $R$ Haar-randomly from $O(d)$, and sample $\vect s$ uniformly at random. In all experiments, we use $L=1$ and $\delta=0.1$. The global optimum value is $-d$, while the next possible objective value is $-(d-1)$. We therefore declare a run successful if $ W_{R,\vect{s}}(\vect{x}_{\mathrm{out}})\leq-d+\frac{1}{2},$
which is equivalent to $W_{R,\mathbf{s}}(\vect{x}_{\mathrm{out}})=-d$. For an algorithm with tunable hyperparameter $\theta$, empirical success
probability $p(\theta)$, and average single-run effective evaluation count
$N_{\rm eff}(\theta)$, we define the evaluation-count time-to-solution
\begin{align}
    \mathrm{TTS}_{p_0}(\theta)
    =
    \left\lceil
    \frac{\log(1-p_0)}{\log(1-p(\theta))}
    \right\rceil
    N_{\rm eff}(\theta),
    \qquad
    p_0=0.90.
\end{align}
For each dimension $d$, we report the lower envelope
\begin{align}
    \mathrm{TTS}_{0.90}(d)
    =
    \min_{\theta}\mathrm{TTS}_{0.90}(\theta).
\end{align}
As in~\cite{leng2023QHD}, we suppress TTS estimates with fewer than five
observed successes. 

\begin{figure}[h]
    \centering
    \includegraphics[width=0.7\linewidth]{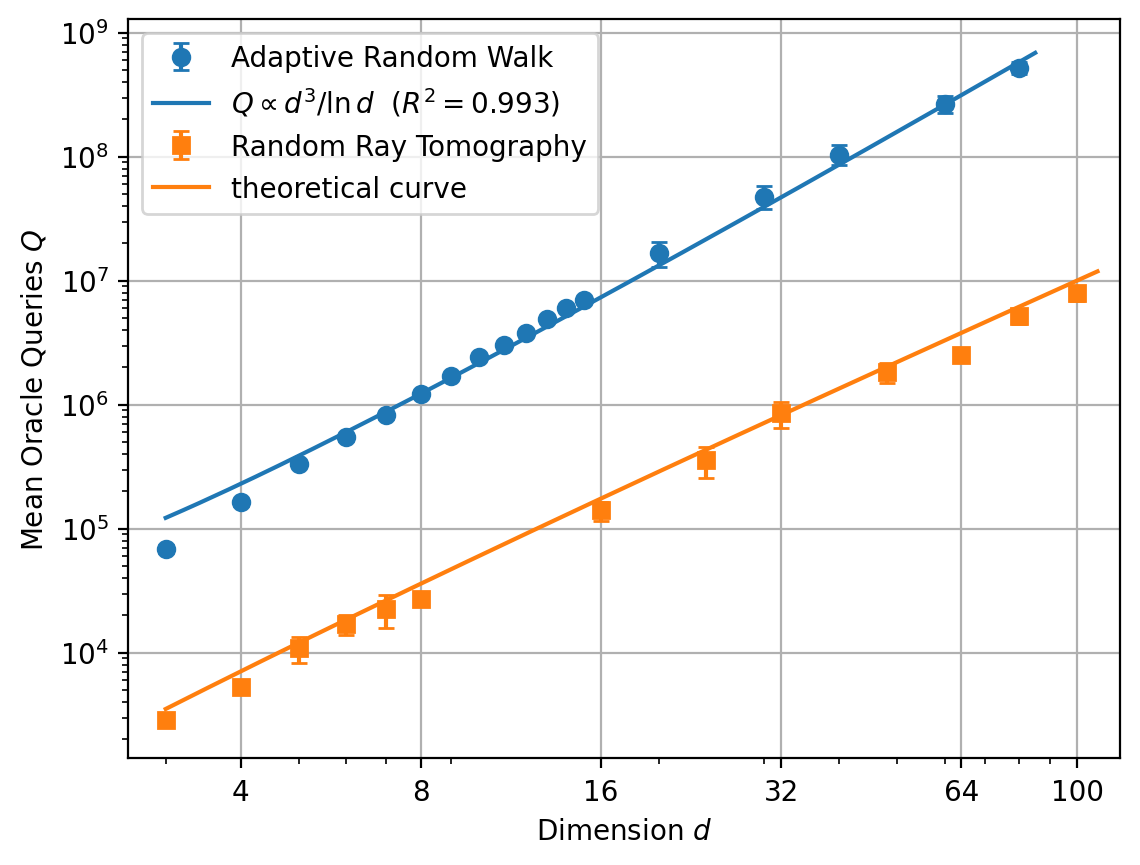}
    \caption{A plot of the mean number of oracle queries $Q$ versus dimension for the two structure-aware classical algorithms (log-log axes). The error bars are $95\%$ confidence intervals ($100$ runs per dimension for the tomography; $50$ runs for the random walk). The random ray tomography (orange) follows the parameter-free theoretical curve $Q \propto d^{2}\ln d$ with leading constant $2C_{\mathrm{pt}}C_{\mathrm{dir}} = 160$ (from $C_{\mathrm{pt}}=10$ points and $C_{\mathrm{dir}}=8$ directions per dimension), whereas the adaptive random walk (blue) follows a steeper $Q \propto d^{3}/\ln d$ ($R^{2}=0.993$, fit for $d\geq 3$).}
    \label{fig:geometry-aware}
\end{figure}

In what follows, we will compare four general-purpose classical algorithms that know nothing about the geometric structure of the problem, and two algorithms that are structure-aware. 
The structure-aware algorithms---the adaptive random walk and the random-ray tomography---have already been analyzed theoretically, and here we simply validated their scaling empirically. In our numerical implementation, the structure-aware algorithms are allowed to run until success is reached. 
The results are plotted in Figure~\ref{fig:geometry-aware}.
We observe that the mean query count of random-ray tomography closely tracks the parameter-free prediction $Q\propto d^2\ln d$, while the adaptive random walk follows the steeper $Q\propto d^3/\ln d$ ($R^2=0.993$). This confirms that explicitly learning the hidden slab geometry realizes the polynomial query complexities established in Section~\ref{subsubsec:learn-rotation}, unlike the generic baselines.

Below, we provide details on the general-purpose non-structure-aware algorithms, of which three are off-the-shelf SciPy optimizers, and one is quantum-inspired. The numerical simulation results are shown in Figure~\ref{fig:classical-tts-rotated-square-well}.

\paragraph{SciPy \texttt{dual\_annealing}.} First, we use this routine as simulated annealing, with $\theta=\texttt{maxiter}$ and
\begin{align}
    \theta\in\{2^6,2^7,\ldots,2^{14}\}.
\end{align}
We run two variants: one with local search disabled, $\texttt{no\_local\_search=True}$, and one with SciPy's default local-search polishing phase enabled, $\texttt{no\_local\_search=False}$. We refer to these as simulated annealing and simulated annealing with local search. 

\paragraph{SciPy \texttt{basinhopping}.} Second, we use SciPy's \texttt{basinhopping} routine, with $\theta=\texttt{niter}$ and
\begin{align}
    \theta\in\{2^3,2^4,\ldots,2^{11}\}.
\end{align}
Trial moves are projected back to $[-L,L]^d$, and the local minimization step is run with box constraints.

\begin{figure}[th]
    \centering
    \includegraphics[width=0.49\linewidth]{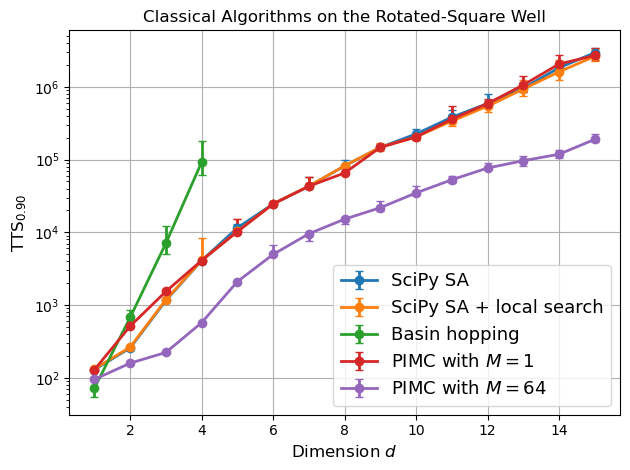}\includegraphics[width=0.5\linewidth]{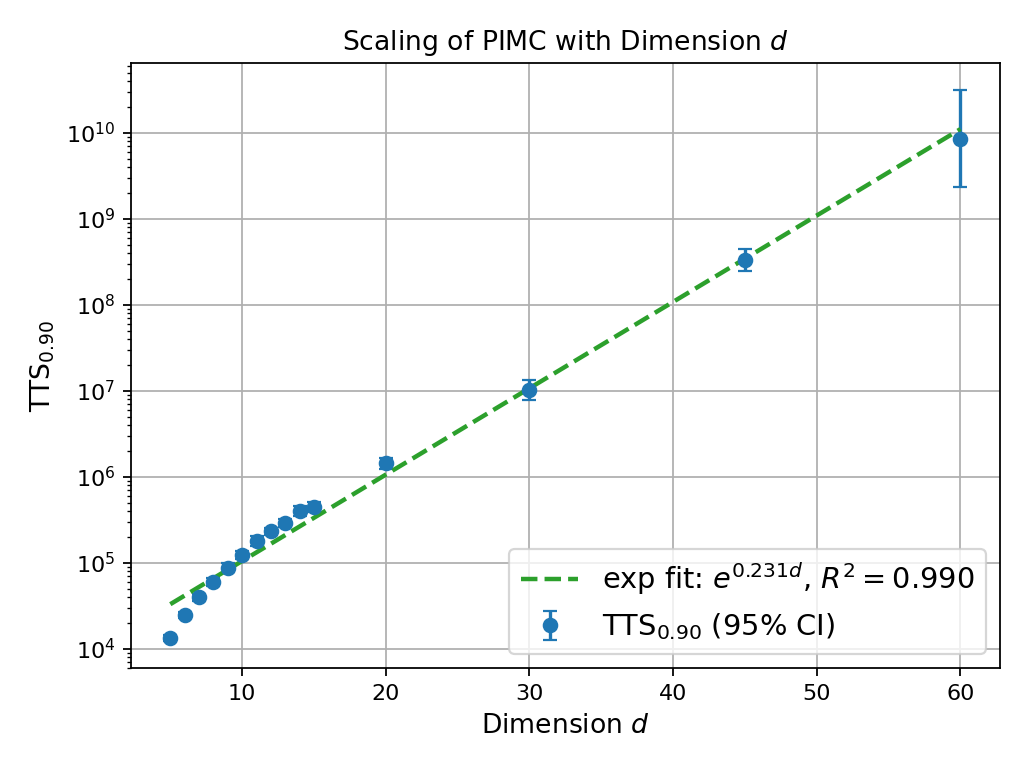}
    \caption{
    Effective-evaluation time-to-solution for classical baselines on the rotated square well problem.
    \textbf{Left:} We plot the lower-envelope $\mathrm{TTS}_{0.90}$ over the algorithmic tuning parameter $\theta$ for each dimension $d$.
    The baselines are simulated annealing, simulated annealing with local search, basin hopping, and path integral Monte Carlo with $M=1$ and $M=64$ beads.
    For the $M=64$ PIMC curve, bead evaluations are counted in the effective batched convention, so that $64$ bead evaluations are charged as one effective evaluation.
    All methods are run on randomly rotated and shifted instances with $L=1$ and $\delta=0.1$.
    Error bars are obtained by propagating a $95\%$ binomial confidence interval for the empirical success probability through the TTS estimate, and points with fewer than five successful runs are omitted.
    The rapid growth of $\mathrm{TTS}_{0.90}$ with $d$ is consistent with superpolynomial scaling for these classical baselines.
    \textbf{Right:} We scale up the best-performing classical baseline, $M=64$ PIMC, to $d=60$, where an exponential fit $\mathrm{TTS}_{0.90}\propto e^{0.231\,d}$ ($R^2=0.990$) confirms that the scaling is indeed superpolynomial. }
    \label{fig:classical-tts-rotated-square-well}
\end{figure}

\paragraph{Path-integral Monte Carlo.} 
Path Integral Monte Carlo (PIMC) has long served as a fundamental benchmark for continuous-variable optimization because of its close connection to quantum fluctuations and ground-state search. 
Here we implement a PIMC-inspired generalized simulated annealing baseline, motivated by the PIMC and simulated quantum annealing methods for continuous double-well and multi-well
landscapes~\cite{stella2006monte,inack2015simulated}. This is a custom
implementation rather than an off-the-shelf SciPy optimizer, but its $M=1$
limit is designed to reproduce the SciPy \texttt{dual\_annealing} update rule
used above, with local search disabled. The state is a ring polymer
\begin{align}
    X
    =
    \left(
        \vect{x}^{(1)},\ldots,\vect{x}^{(M)}
    \right),
    \qquad
    \vect{x}^{(m)}\in[-L,L]^d,
\end{align}
with $M=1$ or $M=64$ beads. For $M=1$, the ring polymer contains a single
walker, there is no internal ring-polymer dynamics, and the energy change is
simply the change in the objective function. Thus the $M=1$ implementation
reduces to a one-walker generalized simulated annealing algorithm with the same
visiting distribution, acceptance rule, temperature schedule, and restart
criterion as SciPy \texttt{dual\_annealing} with local search disabled, following the implementation~\cite{virtanen2020scipy} of generalized simulated
annealing~\cite{xiang2013gensa}.

For $M=64$, we view the beads as a ring-polymer configuration
$X=(\vect{x}^{(1)},\ldots,\vect{x}^{(M)})$ sampled with the annealed effective
action
\begin{align}
    S_s[X]
    =
    \frac{\beta_s}{M}
    \sum_{m=1}^{M}
    W\left(\vect{x}^{(m)}\right)
    +
    \frac{1}{2\Gamma_s^2}
    \sum_{m=1}^{M}
    \left\|
        \vect{x}^{(m+1)}-\vect{x}^{(m)}
    \right\|^2,
    \qquad
    \vect{x}^{(M+1)}\equiv \vect{x}^{(1)} .
\end{align}
Here $\beta_s=1/T_s$, with $T_s$ chosen to be the same temperature as
in the SciPy \texttt{dual\_annealing} schedule used for the $M=1$ update. The
rigid SciPy-like moves shift all beads together, while the internal moves redraw
short bead segments at fixed endpoints; both are accepted using the change in
$S_s[X]$. The internal fluctuation scale is annealed as
\begin{align}
    \Gamma_s
    =
    \Gamma_0
    \left(
        \frac{\Gamma_1}{\Gamma_0}
    \right)^{s/(\theta-1)},
    \qquad
    \Gamma_0=1,
    \qquad
    \Gamma_1=0.01.
\end{align}
For $M=1$, the spring term is absent and the method reduces to the same
single-walker SciPy-like generalized simulated annealing update. The tuning parameter is the number of sweeps $\theta$. For the $M=1$ one-walker implementation, we use the same range as SciPy \texttt{dual\_annealing},
\begin{align}
    \theta
    \in
    \left\{
        2^6,2^7,\ldots,2^{14}
    \right\}.
\end{align}
For the $M=64$ ring-polymer implementation, we use
\begin{align}
    \theta
    \in
    \left\{
        2^5,2^6,\ldots,2^{11}
    \right\}.
\end{align}
A run is declared successful if, at any point during the run, any bead attains
the global objective value $W=-d$. For every algorithm, dimension, and parameter value, we run independent random instances and random seeds. We estimate $p(\theta)$ and $t(\theta)$ empirically, compute $\mathrm{TTS}_{0.90}(\theta)$, and take the lower envelope over $\theta$.

The resulting time-to-solution is shown in Figure~\ref{fig:classical-tts-rotated-square-well}. For every baseline, the lower-envelope $\mathrm{TTS}_{0.90}(d)$ grows rapidly with the dimension $d$ (left), and scaling up the strongest baseline ($M=64$ PIMC) to $d=60$ yields a clean exponential fit $\mathrm{TTS}_{0.90}\propto e^{0.231d}$ with $R^2=0.990$ (right). Together these confirm that all the general-purpose classical optimizers we tested exhibit superpolynomial scaling on the hidden rotated square well problem.
In contrast, the two structure-aware algorithms scale only polynomially, as shown in Figure~\ref{fig:geometry-aware}.

\section{Numerical Simulations of CV-QAOA}\label{sec:Numerics}

In this section, we validate the theoretical results established in Theorems \ref{thm:quad} and \ref{thm:doublewell} with numerical simulations, and explore the performance of CV-QAOA on other functions not considered in the rigorous analysis. First, we carry out a second-quantized numerical simulation to implement $d$-step CV-QAOA $d$-dimensional quadratic function, focusing on $d=2$. Second, we carry out a first-quantized simulation of 1-step CV-QAOA to minimize the Styblinski-Tang (ST) function, spell out the dependence on $\s$, and confirm the predicted effect of initialization on CV-QAOA. Finally, we demonstrate the superior performance of 1-step CV-QAOA in the case of the deflected corrugated spring function relative to previously considered classical and quantum algorithms.

\subsection{$d$-dimensional quadratics}

We design the numerical simulation in second quantized formulation to simulate the CV-QAOA on the $d$-dimensional quadratics problem as follows.
\begin{algorithm}[h]
\caption{Numerical Simulation in second-quantized formalism}
\label{alg:second_quantized_sim}
\begin{algorithmic}[1]
\REQUIRE Dimension $d$, Fock-space cutoff $N$, cost Hamiltonian $C=h(\hat{x})$
\ENSURE Optimized parameters $(\beta, \gamma, \sigma)$ and state $\ket{\beta,\gamma,\sigma}$

\STATE Construct creation and annihilation operators 
$\{a_k, a_k^\dagger\}_{k=1}^d$ as $N \times N$ matrices.
\STATE Build quadrature operators $\{x_k, p_k\}_{k=1}^d$.
\STATE Form the cost and mixer Hamiltonians, and their corresponding unitaries.
\STATE Define squeeze operators
\[
S_k = \exp\left\{-\tfrac{1}{2}\log(\sqrt{2}\sigma)\,\big(a_k^2 - (a_k^\dagger)^2\big)\right\}.
\]
\STATE Prepare the squeezed vacuum state 
$\ket{\sigma} = \bigotimes_k S_k \ket{0}$.
\STATE Apply the unitaries to obtain the variational state 
$\ket{\beta,\gamma,\sigma}$.
\STATE Evaluate the cost function 
$\expval{C}_{\beta,\gamma,\sigma}$ 
and optimize over $(\beta,\gamma,\sigma)$.
\end{algorithmic}
\end{algorithm}
\begin{figure}[h!]
    \centering
    \includegraphics[width=\linewidth]{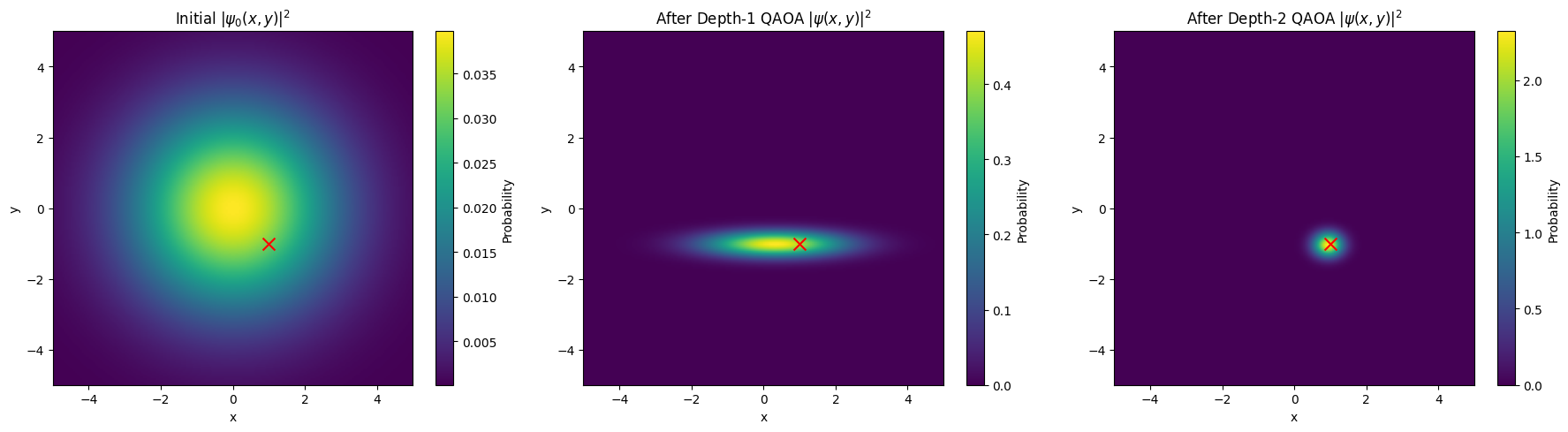}
    \caption{The above plot shows the \textit{final} probability distributions after 1-step and depth-two CV-QAOA respectively, and the initial distribution is displayed on the left-most panel. The red cross denotes the location of the minimum. We consider the function $F(x, y) = \frac{1}{2}(x-1)^2 + \frac{3}{2}(y+1)^2$ with $\sigma =2$. The key point to note is that the matrix $Q = \text{diag}(1, 3)$ does \textit{not} have equal eigenvalues in this case, hence 1-step QAOA is able to optimize only one of the two dimensions ($y$ but not $x$). Once we introduce an additional layer, the CV-QAOA has enough parameters to optimize both the $x$ and $y$ directions.}
    \label{fig:Quad1}
\end{figure}

In Fig.~\ref{fig:Quad1}, we compare the probability distributions obtained from CV-QAOA at depths one and two against the initial distribution. The cost function under consideration is $F(x,y) = \tfrac{1}{2}(x-1)^2 + \tfrac{3}{2}(y+1)^2$ with initial width $\sigma = 2$. The quadratic form is determined by the matrix $Q = \mathrm{diag}(1,3)$, which has unequal eigenvalues and therefore introduces anisotropy between the $x$ and $y$ directions. This asymmetry plays a crucial role in the optimization performance. At depth one, the variational ansatz has only enough degrees of freedom to align with a single direction of curvature, leading to partial optimization that improves the $y$-coordinate but leaves the $x$-coordinate far from optimal. By contrast, introducing an additional layer at depth two provides sufficient variational flexibility to capture the anisotropic curvature in both directions. As a result, the final probability distribution is sharply peaked around the true minimum, indicated by the red cross. This illustrates how increasing circuit depth systematically enhances the expressive power of CV-QAOA, enabling it to adapt to more general quadratic landscapes.

\subsection{$d=1$ Styblinski-Tang (quartic double well)}
\label{sec:STnumerics}

We use the first-quantized numerical simulation for CV-QAOA on general 1D functions as follows.

\begin{algorithm}[h]
\caption{Numerical Simulations in first-quantized formulation}
\label{alg:grid_sim}
\begin{algorithmic}[1]
\REQUIRE Spatial domain (e.g. $[-5,5]$), number of grid points $N=2^{17}$, initial wavefunction $\psi_0(x)$
\ENSURE Optimized parameters $(\beta,\gamma,\sigma)$ and state $\psi(x)=\braket{x|\beta,\gamma,\sigma}$

\STATE Discretize the spatial domain and normalize $\psi_0(x)$.
\STATE Apply the cost unitary in position space: $\psi_1(x) = e^{-i\gamma C(x)} \psi_0(x)$.
\STATE Perform a Fast Fourier Transform (FFT) to obtain momentum-space wavefunction $\tilde{\psi}_1(p)$.
\STATE Apply the mixer unitary in momentum space: $\tilde{\psi}(p) = e^{-i\beta B(p)} \tilde{\psi}_1(p)$.
\STATE Perform an inverse FFT to return to position space: $\psi(x) = \braket{x|\beta,\gamma,\sigma}$.
\STATE Evaluate the cost function: $\expval{C}_{\beta,\gamma,\sigma}$.
\STATE Optimize over parameters $(\beta,\gamma,\sigma)$.
\end{algorithmic}
\end{algorithm}

We have numerically implemented the CV-QAOA algorithm (see Figs.~\ref{fig:ST_first_quant}, \ref{fig:ST_prob}), and find excellent confirmation of the analytical expectation in \eqref{fOpt} that we converge to the true minimum (as opposed to the local minimum) as $\sigma,\alpha\rightarrow0$.

Fig.~\ref{fig:ST_first_quant} shows that, within a first-quantized formalism and using a sufficiently fine spatial discretization, 1-step QAOA can efficiently converge to the true minimum of the Styblinski–Tang (ST) function. In the limit $\sigma \to 0$, the algorithm has enough resolution to represent sharply localized states, which allows the variational ansatz to align with the global minimum even at minimal circuit depth.
\begin{figure}[h!]
    \centering
    \begin{subfigure}[t]{0.48\linewidth}
        \centering
        \includegraphics[width=\linewidth]{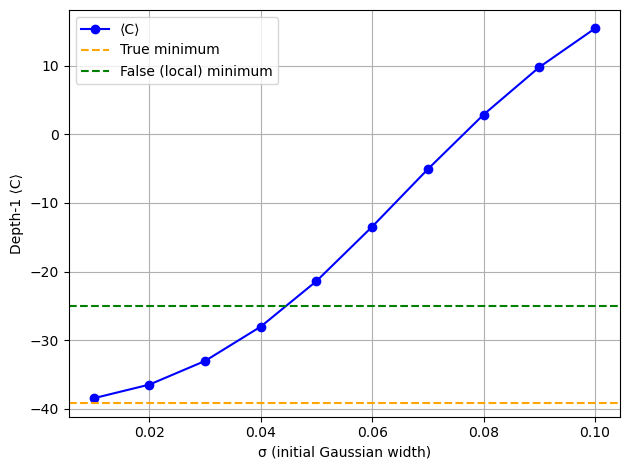}
        \caption{In a first quantized formalism with a sufficiently fine mesh ($N = 2^{17}$) on the domain $[-6, 6]$, one can efficiently implement 1-step QAOA to converge to the true minimum of the ST function as $\sigma\rightarrow 0$. To approach the optimal parameter slice mentioned in \eqref{gOpt}, we chose $\alpha = 0.01$ and set $(\beta, \gamma) = (\alpha\sigma, 1.161/\alpha\sigma)$ as we vary $\sigma\in[0.01, 0.1]$.}
        \label{fig:ST_first_quant}
    \end{subfigure}
    \hfill
    \begin{subfigure}[t]{0.48\linewidth}
        \centering
        \includegraphics[width=\linewidth]{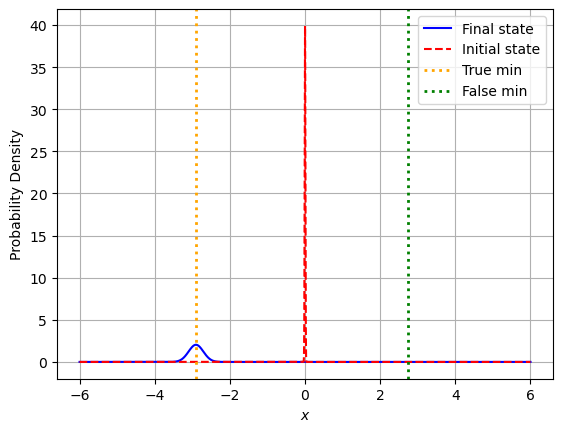}
        \caption{The initial (dashed red) and final (blue) probability distributions after 1-step QAOA for the 1D Styblinski-Tang (ST) function. Since $\sigma$ is small, we start out with an approximate delta function, and 1-step QAOA moves us to a Gaussian localized around the true minimum of the ST function.}
        \label{fig:ST_prob}
    \end{subfigure}
    \caption{Results for 1-step QAOA applied to the 1D Styblinski-Tang (ST) function. (Left) Convergence to the true minimum in the first-quantized formalism as $\sigma \to 0$. (Right) Initial and final probability distributions, showing localization around the true minimum after QAOA evolution.}
    \label{fig:ST_combined}
\end{figure}

Fig.~\ref{fig:ST_prob} presents the corresponding probability distributions for the one-dimensional ST function. The initial state (dashed red) is prepared with a small $\sigma$, resulting in an approximate delta function centered near the starting point. After the action of 1-step QAOA, the distribution (blue) becomes a Gaussian localized around the true minimum of the ST function, illustrating that even a single QAOA layer is sufficient to shift the state toward the correct solution. It also shows explicitly that the final distribution after 1-step CV-QAOA is peaked precisely around the expected location of the true minimum: $x^* = -2.90353$.

\subsection{Beyond the classically dequantizable regime} \label{sec:beyondClassSim}

In this subsection, we provide two systematic strategies to generate optimization problems that cannot be solved by CV-QAOA in the classically simulable regime. Furthermore, these two strategies produce problems that are not solved by previously considered classical and quantum algorithms with the same success rate as that of CV-QAOA. Given the superior performance of CV-QAOA, these problems offer evidence in favor of the quantum advantages of CV-QAOA over existing classical and quantum algorithms.

\subsubsection{Strategy 1: $\nabla C(\vect{0})\not\propto \arg\min C$ -- the deformed deflected corrugated spring}
As illustrated by Theorem \ref{thm:GradPointingToMin}, $1$-step CV-QAOA in the classically-simulable regime can minimize any $d$-dimensional real-analytic function that satisfies 
\begin{align}\label{GradPoint}
    \nabla C(\vect{0})\propto\arg\min_\xv C(\xv)
\end{align}
To find problems that are not solved in the classically-simulable regime, one may analyze problems $C$ that do not satisfy \eqref{GradPoint}. 

We consider one such concrete example: a deformed version of the deflected corrugated spring function. 
\begin{align}
    S(x, y) = \frac{1}{10}\left[(x-x_0)^2+(y-y_0)^2\right] - \cos\left(5\sqrt{(x-x_0')^2+(y-y_0')^2}\right)
\end{align}
When $\xv_0=\xv_0'$, $S(x, y)$ is the usual deflected corrugated spring function. When $\xv_0\neq\xv_0'$, we say it is \textit{deformed}. It is instructive to compute:
\begin{align}
    \nabla S(\xv) = \frac{1}{5}(\xv-\xv_0) +
    5\frac{\sin(5\|\xv-\xv_0'\|)}{\|\xv-\xv_0'\|}(\xv-\xv_0')\Rightarrow \begin{cases} \nabla S(\vect{0}) \propto\arg\min S,&\xv_0=\xv_0'\text{ (undeformed)}\\
     \nabla S(\vect{0}) \not\propto\arg\min S,&\xv_0\ne\xv_0'\text{ (deformed)}
    \end{cases}
\end{align}
On the undeformed problem, 1-step CV-QAOA with a point-like initial state can achieve $\Omega(1)$ success probability since \eqref{GradPoint} is satisfied. On the other hand, QHD performs poorly, even worse than generic classical algorithms like stochastic gradient descent \cite{leng2023QHD}. Furthermore, classical line search from the gradient at the origin can find the global minimum.

For the deformed problem, $S$ does not satisfy equation \eqref{GradPoint}, making it a viable candidate for classical hardness. In the CV-QAOA notation, the QHD schedule \cite{leng2023QHD} is given for any $t\in[0, \overline{T})$
\begin{align}
    \beta_{\text{QHD}}(t) &= \frac{2\d}{\d+t^3} &\g_{\text{QHD}}(t)  = 2\d t^3
\end{align}
with $\d = 0.001$.

\begin{table}[h!]
\centering
\begin{tabular}{|l|r|}
\hline
Algorithm & Success Rate \\
\hline
QAA  & 6.99\% \\
QHD  & 3.14\% \\
CV-QAOA & 39.93\% \\
\hline
\end{tabular}

\caption{A summary of success rates by algorithm on the deflected corrugated spring problem \cite{leng2023QHD} in outputting a point within distance 1 of the global minimum. The last row shows the superior performance of 2-step CV-QAOA. 
}\label{TableSummary}
\end{table}
\begin{figure}[h!]
    \centering
    \includegraphics[width=\linewidth]{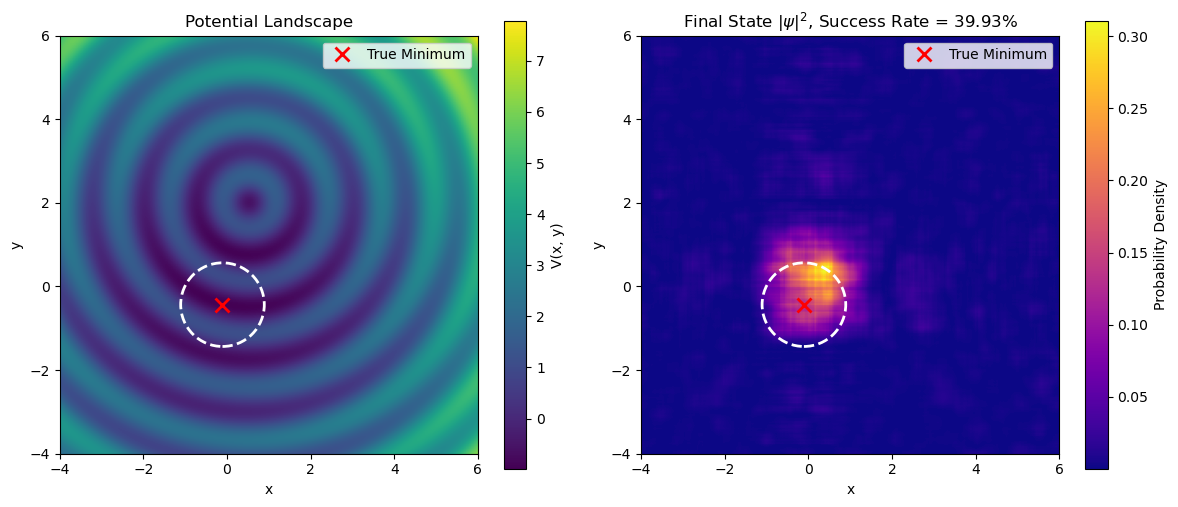}
    \caption{{\textbf{Left}}: the cost landscape of the deformed Deflected Corrugated spring function $S$ with $\pmb{x}_0 = (5,5)$ and $\pmb{x}_0' = (0.5, 2)$, which has many false minima and is highly non-convex. The red star marks global minimum, where the white dash line marks a radius-1 circle around minimum. {\textbf{Right}}: the final probability distribution of CV-QAOA state at its optimal parameters $(\b_1,\g_1, \b_2, \g_2) \approx (0.23,0.72,2.88,0.85)$ with a uniform initial state.}
    \label{fig:DeflCorrSpr}
\end{figure}
As illustrated in Fig.~\ref{fig:DeflCorrSpr}, CV-QAOA with a delocalized initial state achieves a high success probability for the deformed problem, far outperforming other quantum algorithms (see Table~\ref{TableSummary}). We implement a first-quantized simulation of the CV-QAOA applied to this function in the domain $(x,y)\in[-4,6]^2$ and initialize the CV-QAOA with a uniformly-distributed initial wavefunction in the domain (which was also the choice for QHD~\cite{leng2023QHD}). To discretize space, we use a resolution of $2^8$ grid-points per dimension.  Success is defined as $\|\xv-\xv_*\|_2<\ep$ with $\ep = 10\%\text{ of the box-side}$ and $\xv_* = \arg\min S$. To compute the success rate, we simply integrate the final CV-QAOA/QHD probability density over the region with $\|\xv-\xv_*\|_2<\ep$. In the right half of Figure~\ref{fig:DeflCorrSpr}, we plot the resultant probability distribution output by CV-QAOA using the best parameters we have found.

\subsubsection{Strategy 2: Exponentially small cost-gradient -- the Easom function}
As noted in the proof of Theorem \ref{thm:doublewell}, the product $\b\g\propto w'(0)^{-1}$. In particular, when the gradient $w'(0)$ is exponentially small in a large neighborhood of the initial support of $\ket{\s}$, $1$-step CV-QAOA requires the product $\b\g$ to be exponentially large. While simulating CV-QAOA with such exponentially large parameters is not feasible, it turns out that such problems can often be solved by CV-QAOA with a delocalized initial state and have higher success rates than previously considered algorithms.
We focus on the example of the Easom function in $d=1,2$ dimensions, given by
\begin{align}
    E(\vect{x}) &= 1-\prod_{j=1}^d\left(e^{-x_j^2}\cos x_j\right).
\end{align}
Observe that its gradients are exponentially small everywhere, $|\p_j E| = O(|x_j|e^{-x_j^2})$, except near the global minimum of $\xv=0$.

While this problem is not feasibly solved by $1$-step CV-QAOA in the de-quantized regime, low-depth CV-QAOA with a uniform initial state far outperforms the QHD algorithm, which was previously established to be superior to other classical and quantum algorithms on this problem \cite{leng2023QHD}.\footnote{Also refer to \url{https://jiaqileng.github.io/quantum-hamiltonian-descent/nonconvex2d.html}.} To compare the performance of CV-QAOA with QHD, we adopt the following conventions. We discretize space $[-L, L]^d$ into a lattice with $2^8$ grid-points per dimension. Success is defined as $\|\xv-\xv_*\|_2<\ep$ with $\ep = 10\%\text{ of } 2L$ and $\xv_* = \arg\min E$. To compute the success rate, we simply integrate the final CV-QAOA/QHD probability density over the region with $\|\xv-\xv_*\|_2<\ep$. In the CV-QAOA notation, the QHD schedule is given for any $t\in[0, \overline{T})$ \cite{leng2023QHD}
\begin{align}
    \beta_{\text{QHD}}(t) &= \frac{2\d}{\d+t^3} &\g_{\text{QHD}}(t)  = 2\d t^3
\end{align}
with $\d = 0.001$. To ensure a fair comparison between the two algorithms, we define the total runtime of the QHD algorithm in terms of the CV-QAOA parameters as the $\ell^1$-norm 
\begin{align}
    T_p = \sum_{k=1}^p\left[\abs{\b_k}+
    \abs{\g_k}\right].
\end{align}
In other words, for QHD, the parameter $t$ runs between $0$ and $\overline{T}$, where $\overline{T}$ is the solution to the equation:
\begin{align}
    T_p = \int_0^{\overline{T}} dt\left[\b_{\text{QHD}}(t)+\g_{\text{QHD}}(t)\right].
\end{align}
Stated concisely, we ensure that both algorithms have the same integrated time. To obtain near-optimal CV-QAOA parameters, we used a number $G$ of random initial guesses, which were then used as initial guesses for the gradient-free Powell optimizer that works efficiently for this problem. Then we select the set of parameters that maximize the success probability, using which we evaluate the runtime $T_p$ for each $p$.

\begin{figure}[ht]
    \centering
    \includegraphics[scale=0.5]{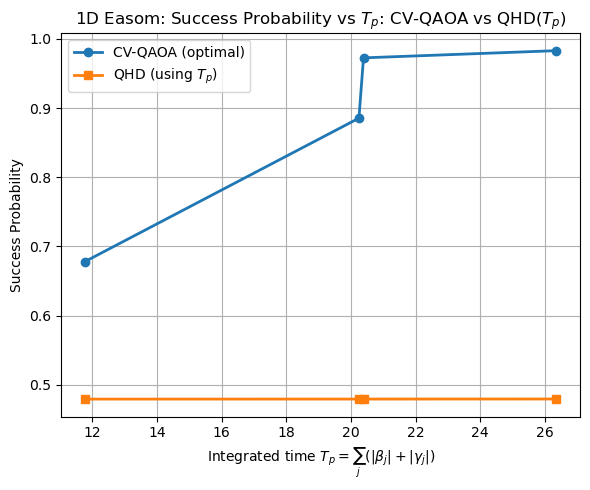}
    \includegraphics[scale=0.5]{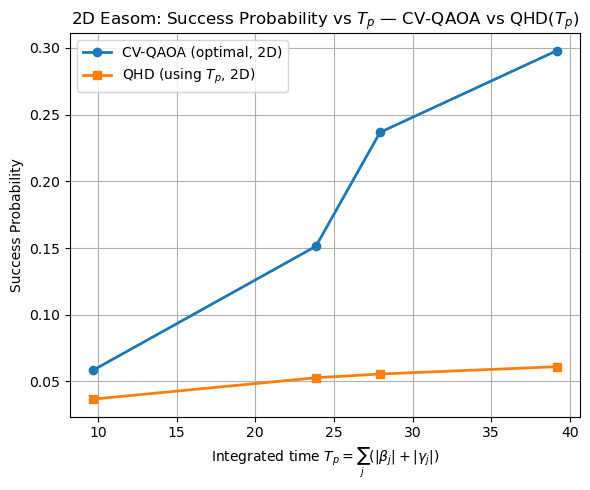}
    \caption{A comparison of the final accuracy at time $T_p$ achieved by QHD and CV-QAOA on the one- and two-dimensional Easom functions.}
    \label{fig:easom}
\end{figure}

Figure~\ref{fig:easom} compares the final accuracy achieved by CV-QAOA and QHD as a function of the matched integrated runtime $T_p$, in both $d=1$ and $d=2$. In both dimensions, low-depth CV-QAOA with a uniform initial state reaches a higher accuracy than QHD at the same integrated time, and this advantage persists as $T_p$ grows. The greater difficulty of the two-dimensional instance is evident in the absolute success probabilities, which peak near $0.3$ for $d=2$ compared to nearly $1$ for $d=1$, yet CV-QAOA retains a clear margin over QHD in both cases. The improvement is realized already at small depth $p$: a few layers suffice for CV-QAOA to achieve superior success probability on the Easom function, whereas QHD saturates at a lower accuracy over the same runtime budget. This behavior is consistent with the mechanism discussed above: since the Easom gradient is exponentially small away from the minimum, the point-like ($\sigma\to0$) regime is ineffective, and only a delocalized initial state allows CV-QAOA to sense the global structure of the landscape.

\section{Discussion and Outlook}\label{sec:Outlook}

Our results establish the Continuous-Variable Quantum Approximate Optimization Algorithm (CV-QAOA) as a general-purpose framework for tackling continuous optimization problems that have traditionally resisted efficient classical treatment.
By rigorously analyzing its performance, we demonstrated that CV-QAOA can match classical methods for convex quadratic problems.
Moreover, it achieves provable superpolynomial-in-$d$ query complexity separations for various families of nonconvex optimization problems.
For the rotated double well problem from \cite{leng2023nonconvex}, we show CV-QAOA solves in $O(1)$ query while information theoretically $\Omega(d/\log d)$ classical queries are required.
For the rotated square well problem we introduced here, we show RsAA solves in $O(\poly\log d)$ time and can be simulated by CV-QAOA with $d^{o(1)}$ quantum queries, while any classical algorithms that learn the hidden rotation provably require $\Omega(d^2/\log d)$ queries.
Our small-scale numerical simulations further show that low-depth CV-QAOA can outperform QHD on the deformed deflected corrugated spring and Easom functions, illustrating its promise on more general nonconvex landscapes.
The comparison in this paper highlights that CV-QAOA can achieve much fewer queries and significantly reduced gate complexity relative to both classical solvers and the Quantum Hamiltonian Descent (QHD), pointing to its potential as a powerful primitive for optimization in high dimensions.

Looking forward, several challenges and opportunities emerge. From a theoretical perspective, extending the analysis of CV-QAOA beyond depth $p=1$ to broader classes of nonconvex functions, including those with multiple interacting minima or highly anisotropic landscapes, remains an open problem. Another goal is to identify families with polynomial quantum query complexity and an exponential classical query lower bound in $d$. Establishing a substantial advantage in total runtime also requires accounting for state preparation, parameter optimization, and query implementation. 
Finally, the advantage of CV-QAOA over stimulated annealing and path integral Monte Carlo on the rotated square well problem motivates a systematic search for broader classes of optimization landscapes in which CV-QAOA can outperform classical methods, potentially revealing more general regimes of quantum advantage in continuous optimization.

Towards practical implementation of CV-QAOA, incorporating adaptive parameter selection, error mitigation, and hybrid quantum-classical routines could further enhance the robustness and practical performance of CV-QAOA. Realizing CV-QAOA on CV-hardware requires high-fidelity Gaussian and non-Gaussian operations. Relevant developments include proposals for multimode photonic gates and programmable trapped-ion bosonic simulation~\cite{hanamura2024implementing,katz2023programmable}, together with experimental demonstrations of nonlinear quadrature measurements using optical feedforward~\cite{sakaguchi2023nonlinear}. 
In Apppendix~\ref{sec:hardware}, we provide some resource estimates for implementing CV-QAOA on both qubit and CV hardware.
Scaling to the gate counts required for high-dimensional problems will require further advances in hardware performance and scalability. Nevertheless, small-scale experimental demonstrations may become feasible as these platforms improve, providing an opportunity to empirically test the predicted advantages.

Furthermore, the exponentially accurate adiabatic theorem employed in our analysis of the rotated square well problem (Appendix~\ref{app:QAA}) is not specific to optimization: the same machinery may be applied to ground-state preparation in quantum many-body systems with bosonic degrees of freedom. This direction connects naturally to continuous-variable and hybrid quantum simulations of quantum field theories \cite{jha2024continuous,ale2025quantum,ale2025simulating} and to adiabatic state preparation for lattice gauge theories \cite{chakraborty2022classically,danna2025adiabatic}. In ongoing work, this develop exponentially accurate adiabatic state preparation for bosonic many-body systems and lattice gauge theories, targeting quantum advantages in regimes where classical methods face obstacles such as Monte Carlo sign problems \cite{troyer2005computational} and the computational cost of tensor-network contraction \cite{schuch2007computational}. CV-QAOA provides a circuit framework for implementing such adiabatic state-preparation protocols, extending its scope beyond the optimization problems considered here.

Overall, our work positions CV-QAOA as a versatile and analytically tractable tool for quantum continuous optimization. By combining provable guarantees with practical implementability, it lays the groundwork for a new class of quantum optimization algorithms that may eventually impact applications ranging from power grid management to drug discovery and machine learning. Future research at the intersection of theory, algorithm design, and hardware development will be critical in translating these conceptual advances into real-world quantum advantage.

\section*{Acknowledgments}
\addcontentsline{toc}{section}{Acknowledgments}
We thank Jiaqi Leng for helpful discussions.
SB is supported by the Mani L. Bhaumik Institute for Theoretical Physics.
LZ is supported by the Google Research Scholar Award and Department of Navy Award No. N000142612399 issued by the Office of Naval Research.

\subsection*{Use of generative AI} Claude Fable 5 and GPT-5.6 Sol assisted in developing the numerical simulation code used to benchmark classical algorithms on the rotated square well problem. Both were also used to derive and check the finite-size and smoothing error estimates, and to develop the geometric and probabilistic arguments underlying the classical query-complexity analysis of the random-ray tomography algorithm. GPT-6 Astra assisted with drafting and revising manuscript text and checking intermediate algebraic calculations. All calculations and mathematical arguments were independently checked by the authors to ensure their correctness.

\addcontentsline{toc}{section}{References}
\bibliographystyle{unsrt}
\bibliography{refs}

\newpage
\begin{appendices}
\addtocontents{toc}{\protect\setcounter{tocdepth}{1}}

\section{CV-QAOA for quadratic functions}\label{app:Quadratics}
In this appendix, we present the detailed analysis underlying Theorems~\ref{thm:quad} and~\ref{thm:1DrealAnalytic}, together with the gradient-driven quadratic variant stated in Theorem~\ref{thm:grad}. Our starting point is the CV-QAOA objective
\begin{align}
    F(\beta,\gamma,\sigma)\coloneqq \expval{C}_{\beta,\gamma,\sigma}
    =
    \bra{\beta,\gamma,\sigma}C\ket{\beta,\gamma,\sigma},
\end{align}
whose minimization determines the performance of the algorithm. 

Using the general conjugation formulas for the basic operators \((\hat x,\hat p)\) under a single CV-QAOA layer, we compute expectation values of polynomial and real-analytic cost functions explicitly. We then analyze one-step CV-QAOA for quadratic functions in detail. This serves two purposes. First, it gives a complete closed-form expression for the quadratic objective after one layer, which makes clear why a single ordinary CV-QAOA step is sufficient only in the degenerate case \(Q\propto \Id\), but not for a generic nondegenerate Hessian. Second, it isolates the two distinct limiting regimes of the initial Gaussian width \(\sigma\), thereby clarifying the role of the initial state already at the quadratic level. Building on this one-step analysis, we prove Theorem~\ref{thm:quad}, showing that \(d\)-step CV-QAOA is sufficient for arbitrary strictly convex quadratics, and we derive explicit error bounds in terms of the variational parameters. We also show that by adding a gradient-driven correction \(\hat D(\etav)\), one can effectively reduce the quadratic problem to an isotropic one, yielding the one-step gradient-driven construction of Theorem~\ref{thm:grad}. Finally, we prove Theorem~\ref{thm:1DrealAnalytic}, which identifies the mechanism of one-step CV-QAOA in the point-like regime for general one-dimensional real-analytic functions: to leading order, the circuit acts as an effective displacement controlled by \(C'(0)\), thereby reducing the optimization to a one-parameter search.

\subsection{The effect of one-step CV-QAOA for $d$-dimensional quadratics}
Consider a general quadratic cost function 
\begin{align}\label{GenQuadCost}
    C(\xv)=\frac{1}{2}\sum_{i, j}Q_{ij} x_i x_j + \sum_kb_k x_k = \frac12 \xv^\top Q \xv + \bv^\top \xv
\end{align}
where we take $Q=Q^\top\in\amsmathbb{R}^{d\times d}$ to be a symmetric matrix without loss of generality. Furthermore, we assume strict convexity, i.e. that all eigenvalues of $Q$ are strictly positive. Since $\det Q >0$, we can rewrite $C(\xv)$ as
\begin{align}\label{CompSq}
    C(\xv) = \frac{1}{2}(\xv - \xv_0)^\top Q(\xv - \xv_0) + c_0
    \quad \text{ with }\quad
    \xv_0 = -Q^{-1}\bv
\end{align}
where $c_0 = -\frac12\xv_0^\top Q \xv_0$ is a constant offset that is unimportant for the optimization problem.
With this rewriting, $C(\xv)\geq c_0$, with the minimum located at $\xv=\xv_0$. It will be convenient to define $\zv\equiv \xv-\xv_0$. Without loss of generality, one may set $c_0=0$.

However, when $Q\not\propto \Id$ is non-degenerate, 1-step CV-QAOA is \textit{not} sufficient to guarantee convergence to the minimum. 

Using Lemma \ref{lem:px} and $[z_k, p_\ell]=[x_k,p_\ell]=i\delta_{k,\ell}$,
\begin{align}
    e^{i\gamma C}e^{i\beta B}z_ke^{-i\beta B}e^{-i\gamma C} &= \sum_m \cA_{km} z_m+ \beta p_k
\end{align}
where $\cA_{km} = \delta_{km} - \beta\gamma Q_{km}$ is a symmetric matrix. This implies
\begin{align}
    g(\beta,\gamma,\sigma) 
    &= \frac{1}{2} \sum_{k,\ell} Q_{k\ell}\bra{\sigma} e^{i\gamma C}e^{i\beta B}z_kz_\ell e^{-i\beta B}e^{-i\gamma C}\ket{\sigma}\nonumber\\
    &= \frac{1}{2}\bra{\sigma}\cA_{mk}Q_{k\ell}\cA_{\ell n}z_m z_n +\beta \cA_{mk}Q_{k\ell}\{p_\ell, z_m\}+\beta^2p_k Q_{k\ell}p_\ell \ket{\sigma}\label{fCalc}
\end{align}
where we used Einstein notation (repeated indices that are not free are summed over) in the last line. 

Next, we evaluate the above expectation term-by-term. Given any real, symmetric, invertible matrix $A\in\amsmathbb{R}^{n\times n}$, one has
\begin{align}
    \int e^{-\frac{1}{2\sigma^2} \xv^\top A\xv}d^n x &= \frac{(2\pi\sigma^2)^{\frac{n}{2}}}{\sqrt{\det A}}&
    \int e^{-\frac{1}{2\sigma^2} \xv^\top A\xv} x_p x_q\;d^n x&= \frac{(2\pi\sigma^2)^\frac{n}{2}}{\sqrt{\det A}}A^{-1}_{p,q}\sigma^2
\end{align}
Using the second of these equations, we obtain
\begin{align}\label{SeqI}
    \sum_{k,\ell}S_{k\ell}\bra{\sigma}z_kz_\ell\ket{\sigma} &=\sum_{k,\ell}S_{k\ell}\bra{\sigma}x_kx_\ell\ket{\sigma} +\xv_0^\top S\xv_0\nonumber= \text{Tr}(S)\sigma^2+\xv_0^\top S\xv_0
\end{align}
where $S_{mn} = \cA_{mk}Q_{k\ell}\cA_{\ell n}$ is symmetric. The linear terms vanish since the initial state is centered at the origin, so $\braket{\sigma|x\xv|\sigma}=0$. Similarly, we have 
\begin{align}
    \sum_{k,\ell}\bra{\sigma}Q_{k\ell}p_kp_\ell\ket{\sigma} &= - {\frac{1}{(2\pi\sigma^2)^\frac{n}{2}}}\sum_{k\ell}Q_{k\ell}\int d^dx\; e^{-\frac{1}{4\sigma^2} \xv^\top \xv}\frac{\partial^2}{\partial x_k\partial x_\ell}e^{-\frac{1}{4\sigma^2} \xv^\top \xv}=\frac{\Tr(Q)}{4\sigma^2} .
\end{align}
The cross-terms must vanish since they are purely imaginary, i.e.
\begin{align}
    \bra{\sigma}\{x_k,p_\ell\}\ket{\sigma} &= -{\frac{i}{(2\pi\sigma^2)^\frac{n}{2}}}\int d^dx\; e^{-\frac{1}{4\sigma^2} \xv^\top \xv}\left\{x_k,\frac{\partial}{\partial x_\ell}\right\}e^{-\frac{1}{4\sigma^2} \xv^\top \xv} = 0.
\end{align}
Then the expected value of the cost is (see Appendix~\ref{app:Quadratics}):
\begin{align}\label{fC}
    \braket{\hat C}_{\b,\g,\s}= \text{Tr}[Q(1-\beta\gamma Q)^2]\frac{\sigma^2}{2}+\xv_0^\top Q(1 - \beta\gamma Q)^2\xv_0+\frac{\beta^2}{2}\frac{\text{Tr}(Q)}{4\sigma^2}
\end{align}

\paragraph{Special case: degenerate Hessian $Q\propto \Id$.} In this case, we let $Q=q\Id$. The optimality gap simplifies to
\begin{align}
    g(\beta,\gamma,\sigma) &= q(1-\beta\gamma q)^2\left(d\sigma^2+\|{\xv_0}\|^2\right) + \frac{qd}{8} \frac{\beta^2}{\sigma^2}
\end{align}
We note the first term can be set to zero by choosing $\gamma=1/(\beta q)$.
Then the optimality gap is $(qd/8)(\beta/\sigma)^2$, where $(\beta,\sigma)$ are free parameters.
Remarkably, there are two distinct ways to make this gap vanish and converge to the global minimum:
\begin{enumerate}
    \item If we start out in a state that is highly delocalized by letting $\sigma\rightarrow\infty$ (which is the strategy taken by QHD), then setting $\gamma =1/(\beta q)$ drives us to the minimum with optimality gap $O(qd(\beta/\sigma)^2)$.
    \item Alternatively, one may consider the regime where we take $\sigma\to 0$. We can still converge to the minimum by setting $\gamma=1/(\beta q)$ and taking $\beta\ll \sigma$.
\end{enumerate}
Both of these scenarios drive the cost-expectation to zero.
\subsection{Proof of Theorem \ref{thm:quad}}
In this section, we provide a complete proof of Theorem \ref{thm:quad}. We start with the heuristic argument underlying the proof, then recount some basic facts about the Schatten $p$-norm and H\"older's inequality, which will be crucial for proving Theorem \ref{thm:quad}, and then introduce some useful notation. Following that, we give clean derivations of each error term in the statement of Theorem \ref{thm:quad}.

\paragraph{Intuition behind Theorem \ref{thm:quad}.}
Since \(Q\) is real, symmetric, and positive definite, we may diagonalize the quadratic form and write
\begin{align}
    C(\xv)=\frac12\sum_{s=1}^d q_s z_s^2,
    \qquad q_s>0,
    \qquad z_s:=x_s-x_{0s},
\end{align}
with \(Q=\mathrm{diag}(q_1,\dots,q_d)\). Applying Lemma \ref{lem:px}, one finds that a single QAOA layer acts on \((\zv,\pmb p)\) as
\begin{align}
    e^{i\gamma C}e^{i\beta B}z_ke^{-i\beta B}e^{-i\gamma C}
    &= \cA_{km} z_m+ \beta p_k, \\
    e^{i\gamma C}e^{i\beta B}p_ke^{-i\beta B}e^{-i\gamma C}
    &= p_k - \gamma Q_{k\ell} z_\ell,
\end{align}
where $\cA_k=\amsmathbb{I}-\beta_k\gamma_k Q$. The important point is that the transformed operators are still linear in \(\zv\) and \(\pmb p\). Iterating this through the circuit therefore gives
\begin{align}
    U_n^\dagger \zv U_n = R^{(n)}\zv + S^{(n)}\pmb p,
\end{align}
for some coefficient matrices \(R^{(n)}\) and \(S^{(n)}\). A straightforward computation gives the recursion
\begin{align}\label{recursionRS}
    R^{(n)} &= R^{(n-1)}\cA_n - \gamma_n S^{(n-1)}Q, &
    S^{(n)} &= \beta_n R^{(n-1)} + S^{(n-1)},
\end{align}
with initial data
\begin{align}
    (R^{(1)},S^{(1)})=(\cA_1,\beta_1\amsmathbb{I}).
\end{align}
Thus the entire effect of the depth-\(d\) circuit is encoded in the matrices \(R^{(d)}\) and \(S^{(d)}\). Evaluating the expectation value in the Gaussian initial state gives
\begin{align}\label{eqExp}
    \braket{\hat{C}}_{\betav,\vect{\gamma},\sigma}
    =
    \bra{\sigma}\zv^\top (R^{(d)})^\dagger Q R^{(d)} \zv\ket{\sigma}
    +
    \frac{1}{4\sigma^2}\Tr\bigl(S^{(d)\dagger}Q S^{(d)}\bigr).
\end{align}
The recursion \eqref{recursionRS} shows how this cancellation can be achieved sequentially. Consider first taking \(\beta_1\to 0\) while holding \(\lambda_1=\beta_1\gamma_1\) fixed. Then
\begin{align}
    R^{(2)} &= \cA_1\cA_2-\gamma_2\beta_1 Q
    \xrightarrow[\beta_1\gamma_1=\lambda_1]{\beta_1\to0}
    \cA_1\cA_2, \\
    S^{(2)} &= \beta_2\cA_1+\beta_1\Id
    \xrightarrow[\beta_1\gamma_1=\lambda_1]{\beta_1\to0}
    \beta_2\cA_1.
\end{align}
Proceeding similarly at the next step,
\begin{align}
    R^{(3)}
    &\xrightarrow[\beta_1\gamma_1=\lambda_1]{\beta_1\to0}
    \cA_1\cA_2\cA_3-\gamma_3\beta_2\cA_1Q
    \xrightarrow[\beta_2\gamma_2=\lambda_2]{\beta_2\to0}
    \cA_1\cA_2\cA_3, \\
    S^{(3)}
    &\xrightarrow[\beta_1\gamma_1=\lambda_1]{\beta_1\to0}
    \beta_3\cA_1\cA_2+\beta_2\cA_1
    \xrightarrow[\beta_2\gamma_2=\lambda_2]{\beta_2\to0}
    \beta_3\cA_1\cA_2.
\end{align}
The pattern is now transparent: by taking the limits sequentially for \(k=1,\dots,d-1\), one removes the unwanted mixing terms and obtains
\begin{align}\label{limitR}
    \lim\limits_{\substack{\beta_{d-1}\rightarrow 0 \\ \beta_{d-1}\gamma_{d-1}=\lambda_{d-1}}}
    \cdots
    \lim\limits_{\substack{\beta_1\rightarrow 0 \\ \beta_1\gamma_1=\lambda_1}}
    R^{(d)}
    =
    \prod_{j=1}^d \cA_j
    =
    \begin{pmatrix}
        \prod_j(1-\lambda_j q_1) & 0 & \cdots & 0 \\
        \cdots & \prod_j(1-\lambda_j q_2) & \cdots & 0 \\
        \cdots & \cdots & \ddots & \cdots \\
        0 & 0 & \cdots & \prod_j(1-\lambda_j q_d)
    \end{pmatrix}.
\end{align}
Because each mode appears independently in this diagonal form, we can choose
\begin{align}
    \lambda_k=\frac{1}{q_k},
\end{align}
which sets \(R^{(d)}=0\). This eliminates the position-space contribution in \eqref{eqExp}. The remaining momentum-space term is then suppressed by taking \(\sigma\to\infty\). This is the basic mechanism behind Theorem~\ref{thm:quad}: each layer is used to cancel one normal mode of the quadratic form, and the broad initial Gaussian removes the residual momentum cost.

\paragraph{Useful facts and notation.} Given a linear transformation $T:\RR^d\rightarrow\RR^d$, the Schatten $p$-norm may be defined in terms of the singular values $\{t_j\}_{j=1}^d$. 
\begin{align}
    \|T\|_p = \Big(\sum_{j=1}^dt_j^p\Big)^{\frac{1}{p}},
\end{align}
The above definition extends formally to $p=\infty$ with the convention that $\|*\|_\infty$ is the operator norm. The case $p=1$ coincides with the so-called trace/nuclear norm, and the case $p=2$ is the Frobenius norm. To derive a rigorous error bound, we use the H\"older inequality (see \cite{simon2005trace, watrous2018theory} for details) for the Schatten $p$-norm. For all $1\le p, q\le \infty$ and linear operators $A, B$ on $\RR^d$, we have
\begin{align}\label{factA}
    \abs{\Tr(AB)}&\le \|A\|_p\|B\|_q, &\frac{1}{p}+\frac{1}{q} = 1.
\end{align}
As with the usual $p$-norm, the Schatten norm satisfies the hierarchy $\|*\|_\infty\leq \dots\leq \|*\|_2\leq \|*\|_1$, and is sub-multiplicative, i.e. for any linear operators $A, B\in \RR^{d\times d}$
\begin{align}\label{SubMult}
    \|AB\|_p\leq \|A\|_p\|B\|_p.
\end{align}

It will be convenient to phrase the argument in terms of the limiting (in the sequential, ordered sense of \eqref{limitR}) values of $R^{(n)}$ and $S^{(n )}$ so we can derive a physically meaningful bound as a function of the CV-QAOA parameters. From the recursion relations \eqref{recursionRS}, the limiting values are:
\begin{align}\label{limSeq}
    \tilde{R}^{(n)} &= \prod_{j=1}^n\cA_j, &\tilde{S}^{(n)} &=\sum_{k=1}^{n}\b_k\tilde{R}^{(k-1)}, &\tilde{R}^{(0)} = \Id.
\end{align}
It is convenient to introduce the deviation from the limiting value
\begin{align}
    \delta^{(n)}(\b_{k<n}, \g_{k\le n}) &= R^{(n)} - \tilde{R}^{(n)} & \ep^{(n)}(\b_{k\le n}, \g_{k<n}) &= S^{(n)} - \tilde{S}^{(n)}.
\end{align}

\paragraph{Derivation of the first error term of Theorem \ref{thm:quad}.} 
The first step is to rephrase \eqref{eqExp} as a trace:
\begin{align}
    \bra{\sigma}\vect{z}^\top(R^{(d)})^\dagger Q R^{(d)}\vect{z}\ket{\sigma} = \Tr(\Sigma(R^{(d)})^\dagger Q R^{(d)}) \text{ where }\Sigma_{nm} =  \bra{\sigma}{z}_n{z}_m\ket{\sigma}.
\end{align}
Since $\Tr(\Sigma R^{(d)\dagger }QR^{(d)}) =\Tr(\Sigma \tilde{R}^{(d)\dagger }Q\tilde{R}^{(d)})+2\Tr(\Sigma Q\tilde{R}^{(d)}\delta^{(d)}) + \Tr(\Sigma Q\delta^{(d)}\delta^{(d)})$, the triangle inequality implies
\begin{align}\label{MainIneq}
    \abs{\Tr(\Sigma R^{(d)\dagger }QR^{(d)})}\le \abs{\Tr(\Sigma \tilde{R}^{(d)\dagger }Q\tilde{R}^{(d)})}+2\abs{\Tr(\Sigma Q\tilde{R}^{(d)}\delta^{(d)})} + \abs{\Tr(\Sigma Q\delta^{(d)}\delta^{(d)})}.
\end{align}
where we have used the fact that $[Q, R^{(n)}] =[Q, \tilde{R}^{(n)}]=[Q, \d^{(n)}]= 0$ and $R^{(n)\dagger} = R^{(n)}$. Applying the special case of H\"older's inequality for the Schatten norm \eqref{factA} with $p=\infty$ and $q=1$ and sub-multiplicativity \eqref{SubMult}, we obtain:

\begin{align}
    \abs{\Tr(\Sigma R^{(d)\dagger }QR^{(d)})}
    &\le \Tr(Q)\|\Sigma\|_\infty\left(\|\tilde{R}^{(d)}\|_1^2+2\|\tilde{R}^{(d)}\|_1\|\d^{(d)}\|_1+\|\d^{(d)}\|_1^2\right)
\end{align}
where we have used the fact that $\|Q\|_1 = \Tr(Q)$ since $Q\succ 0$. Equation \eqref{limitR} already establishes the fact that $\|\tilde{R}^{(d)}\|_1$ may be made arbitrarily small: {  since $\tilde{R}^{(d)}$ is diagonal, its trace norm is}
\begin{align}
    { 
    \|\tilde{R}^{(d)}\|_1
    =\sum_{k=1}^d\abs{1-\lambda_k q_k}\prod_{j\neq k}\abs{1-\lambda_j q_k}
    \le \sqrt{\zeta}\,\max_{\ell=1,\dots,d}\abs{1-\lambda_\ell q_\ell},
    \quad\text{with}\quad
    \sqrt{\zeta}=\sum_{k=1}^d\prod_{j\neq k}\abs{1-\lambda_j q_k}.}
\end{align}
To complete the argument, it suffices to show that the function $\delta^{(n)}(\b_{k<n}, \g_{k\le n}) = R^{(n)} - \tilde{R}^{(n)}$ may be  made arbitrarily small (in matrix norms). To derive such a bound, it is instructive to inspect a few low-order expressions for $R^{(n)}$:
\begin{align}
    \delta^{(1)} = R^{(1)} - \tilde{R}^{(1)} &= 0\no\\
    \delta^{(2)} = R^{(2)} - \tilde{R}^{(2)} &= O(\g_2\b_1; \Id, Q)\no\\
    \delta^{(3)} = R^{(3)} - \tilde{R}^{(3)} &= O\Big(\sum_{i=2}^3\Big[\g_i\Big[\sum_{j=1}^{i-1}\b_j\Big]\Big]; \Id, Q\Big)\no\\
    \delta^{(4)} = R^{(4)} - \tilde{R}^{(4)} &= O\Big(\sum_{i=2}^4\Big[\g_i\Big[\sum_{j=1}^{i-1}\b_j\Big]\Big]; \Id, Q\Big) +O\Big(\g_4\sum_{k=3}^4\Big[\b_k\Big[\sum_{i=2}^{k-1}\g_i\Big[\sum_{j=1}^{i-1}\b_j\Big]\Big]\Big]; \Id, Q, Q^2\Big)\no
\end{align}
where the notation $O(p(\vect{\b}, \g); M_1, \dots, M_n)$ means that the coefficients of the terms in the polynomial $p$ is one of the matrices $M_1, \dots, M_n\in \complex^{d\times d}$. Note, furthermore, that the polynomial in our case always vanishes when all arguments are zero. Reasoning inductively, it follows for any $n\ge 2$ that $\delta^{(n)}$ tends to zero \textit{at least} as fast as:
\begin{align}\label{deltaBound2}
    \d^{(n)} = R^{(n)}-\tilde{R}^{(n)} = O\Bigg(\sum\limits_{\substack{k\leq n \\ k\text{ even}} }\sum_{1\leq i_1<\dots<i_k}\g_{i_k}\b_{i_{k-1}}\dots\g_{i_2}\b_{i_1}; \Id, Q, \dots, Q^{n}\Bigg)
\end{align}
The above expression indeed vanishes in the sequential limit of \eqref{limitR} because each monomial in the sum has its lowest index element as a $\b_k$ for some $k$. Since we take $\b_k<\min_{m>k}\{\b_m, \g_m\}$, as in \eqref{limitR}, each monomial is guaranteed to be small in this limit. {  To convert \eqref{deltaBound2} into an explicit bound, recall the pairwise-product sum $P=\sum_{1\le i<j\le d}|\beta_i|\,|\gamma_j|$ and the standing hypothesis $\hat q\,P\le\tfrac12$ of Theorem~\ref{thm:quad}, where $\hat q=\max\{1,(\Tr Q)^d\}$. In each order-$2m$ monomial of \eqref{deltaBound2}, the indices pair up as $(\beta_{i_1},\gamma_{i_2}),\dots,(\beta_{i_{2m-1}},\gamma_{i_{2m}})$ with the $\beta$-index below the $\gamma$-index in every pair, so the sum of absolute values of all order-$2m$ monomials is at most $P^m$. Since every matrix coefficient satisfies $\|Q^j\|_1\le(\Tr Q)^j\le\hat q$ for $0\le j\le d$, the triangle inequality and a geometric series (using $P\le\tfrac12$) give}
\begin{align}\label{deltaBound}
    { 
    \|\d^{(n)}\|_1 \le \hat q\sum_{m\ge 1}P^m \le 2\hat q\,P.}
\end{align}
{  Since $\ket{\sigma}$ is centered at the origin, $\Sigma=\sigma^2\Id+\xv_0\xv_0^\top$ exactly, so $\|\Sigma\|_\infty=\sigma^2+\|\xv_0\|^2$.} Putting everything together:
\begin{align}
    { 
    \abs{\Tr(\Sigma R^{(d)\dagger }QR^{(d)})} \le \Tr(Q)\left(\s^2+\|\xv_0\|^2\right)\left(\|\tilde{R}^{(d)}\|_1+\|\d^{(d)}\|_1\right)^2.}
\end{align}
Applying \eqref{deltaBound} with $n=d$, we obtain
\begin{align}
    { 
    \abs{\Tr(\Sigma R^{(d)\dagger }QR^{(d)})} \le \Tr(Q)\left(\s^2+\|\xv_0\|^2\right)\Big(\sqrt{\zeta}\max_{\ell = 1, \dots, d}|1- \lambda_\ell q_\ell|+2\hat{q}\,P\Big)^2,}
\end{align}
{  which is the first error term of Theorem~\ref{thm:quad}, with the factor of $2$ absorbed into $K$.} Hence, the error in approximating the first term of \eqref{eqExp} indeed tends to zero in the anticipated sequential limit of \eqref{limitR}.

\paragraph{Derivation of the second error term of Theorem \ref{thm:quad}.}
Applying the H\"older inequality for the Schatten norm \eqref{factA} with $p=\infty$ and $q=1$ to the second term of \eqref{eqExp}:
\begin{align}\label{second}
    \abs{\Tr(S^{(d)\dagger}  Q S^{(d)})}=\abs{\Tr(Q S^{(d)}S^{(d)\dagger})}\leq \|{Q}\|_\infty \|S^{(d)}S^{(d)^\dagger}\|_1 \leq \Tr(Q)\|{S^{(d)}}\|_1^2 
\end{align}
where we have used $\|Q\|_\infty \le \|Q\|_1=\Tr(Q)$ and sub-multiplicativity $\|S^{(d)}S^{(d)\dagger}\|_1\le\|S^{(d)}\|_1^2$. Hence, the problem boils down to deriving a bound on $S^{(n)}$. As before, it will be instructive to examine a few low-order terms 
\begin{align}
    \ep^{(1)} &= S^{(1)} - \tilde{S}^{(1)} = 0\no\\
    \ep^{(2)} &= S^{(2)} - \tilde{S}^{(2)}=0\no\\
    \ep^{(3)} &= S^{(3)} - \tilde{S}^{(3)} = O(\b_3\g_2\b_1;Q)\no\\
    \ep^{(4)} &= S^{(4)} - \tilde{S}^{(4)} = O\Big(\sum_{k=3}^4\Big[\b_k\Big[\sum_{i=2}^{k-1} \Big[\g_i\sum_{j=1}^{i-1}\b_j\Big]\Big]\Big];Q\Big)
\end{align}
The right-hand side of the above equations always contains one extra power of the $\b_k$'s than the $\g_k$'s, which implies that $\ep^{(n)}$ has a higher vanishing order than $\tilde{S}^{(d)}$. {  Explicitly, grouping monomials as in the derivation of \eqref{deltaBound} (each order-$(2m{+}1)$ monomial is one unpaired $|\beta_i|\le\hat\beta$ times $m$ pairs, where $\hat\beta=\max_{1\le j\le d}|\beta_j|$), the hypothesis $\hat q\,P\le\tfrac12$ gives}
\begin{align}
    { 
    \|\ep^{(n)}\|_1 \le \hat{q}\,d\hat\beta\sum_{m\ge1}P^m \le 2\hat{q}\,d\hat\beta\,P \le d\hat\beta.}
\end{align}
The recursion for $S^{(n)}$ implies
\begin{align}
    S^{(n)} = \sum_{k=1}^{n}\b_k\tilde{R}^{(k-1)} + \ep^{(n)}
\end{align}
Taking the norm of both sides for $n=d$ and applying the triangle inequality:
\begin{align}
    \|S^{(d)}\|_1
    \leq \sum_{j=1}^{d}\abs{\b_j}\|\tilde{R}^{(j-1)}\|_1 + \|\ep^{(d)}\|_1
    \leq \a\sum_{j=1}^{d}\abs{\b_j} + \|\ep^{(d)}\|_1,
\end{align}
where $\a = \max_{0\le j\le d-1}\|\tilde{R}^{(j)}\|_1$ and the second inequality uses $\|\tilde{R}^{(j-1)}\|_1 \le \a$ by definition. By sub-multiplicativity \eqref{SubMult}, $\a$ is bounded above by $\max\{\prod_{k=1}^{j-1}\|\cA_k\|_1 \mid j=1,\dots,d\}$. Hence:
\begin{align}
    { 
    \|S^{(d)}\|_1 \le \a d\hat{\b} + d\hat\b \le 2\a d\hat{\b},}
\end{align}
{  where the last step uses $\a\ge\|\tilde{R}^{(0)}\|_1=\|\Id\|_1=d\ge 1$.} Plugging this into \eqref{second}, we arrive at a bound on the second term of \eqref{eqExp}:
\begin{align}\label{eq:quad-second-error}
    \frac{1}{4\s^2}\Tr\left(S^{(d)\dagger} Q\, S^{(d)}\right)
    { 
    \;\le\; \frac{\Tr(Q)}{4\s^2}\|S^{(d)}\|_1^2
    \;\le\; \frac{\Tr(Q)\,\a^2 d^2\hat{\b}^2}{\s^2}.}
\end{align}
This concludes the proof of Theorem~\ref{thm:quad}{ , with $K=4\Tr(Q)$ covering both error terms}.

\subsection{A gradient-driven variant for $d$-dimensional quadratics}\label{sec:grad-driven}
Inspired by the fact that gradient information can be used to accelerate classical and quantum optimization methods~\cite{leng2025gradient}, we further introduce a gradient-driven CV-QAOA ansatz with an additional driver $\hat{D}(\etav)$ specified by $d^2$ variational parameters $\etav=(\eta_{i,j})_{i,j=1}^d$:
\begin{equation} \label{eq:grad-QAOA}
    \ket{\paramv,\etav,\sigma} = \prod_{j=1}^p \Big(e^{-i\beta_j \hat{B}}e^{-i\gamma_j [\hat{C}+\hat{D}(\etav)]} \Big)\ket{\sigma},
    \quad \text{where} \quad
    \hat{D}(\etav) = \sum_{i,j=1}^d \eta_{i,j} \frac{\partial  C(\hat\xv)}{\partial x_i} \frac{\partial  C(\hat\xv)}{\partial x_j}.
\end{equation}
We show that including a gradient driver reduces the circuit depth required for convergence of ordinary CV-QAOA. In fact, we reduce number of layers depth from $2d$ to just $2$ for any quadratic regardless of whether $Q$ has degenerate eigenvalues or not. Define
\begin{align}
    U(\beta, \gamma, \etav) = e^{-i\beta B}e^{-i\gamma C} e^{-i\gamma D(\etav)}.
\end{align}

\begin{thm}[1-step gradient-driven CV-QAOA minimizes any $d$-dimensional quadratic.]
\label{thm:grad}
Given any $d$-dimensional strictly convex quadratic, i.e. $C(\xv) =\frac12 \xv^\top Q \xv + \bv^\top \xv$ with $Q\succ 0$, there exists a choice of parameters $(\beta^*, \gamma^*, \etav^*,\sigma^*)$ such that 1-step gradient-driven CV-QAOA \eqref{eq:grad-QAOA} minimizes the cost function up to arbitrarily small optimality gap, i.e.,
$$\left|\bra{\beta^*, \gamma^*, \etav^*,\sigma^*}\hat{C}\ket{\beta^*, \gamma^*, \etav^*,\sigma^*} -\min_{\xv} C(\xv)\right|= \frac{\Tr(Q)}{8} \frac{(\beta^*)^2}{(\sigma^*)^2}.$$ 
\end{thm}
\begin{proof}
Recall that $C(\xv)= \frac{1}{2}\vect{z}^\top Q\vect{z} + c_0$ with $\vect{z} = \xv-\xv_0$, $\xv_0 = -Q^{-1}\bv$ and $c = -\frac12\xv_0^\top Q \xv_0$, then we have (using Einstein summation notation)
\begin{align}
    \partial_j C &= Q_{jk} z_k, &
    D(\etav) &= \eta_{ij} (Q_{i\ell} z_\ell ) (Q_{jk} z_k) = Q_{\ell i} \eta_{ij} Q_{jk} z_k z_{\ell}.
\end{align}
Since $[C, D] = 0$, we may add the arguments of the exponential to find:
\begin{align}
    C+D = \frac{1}{2}\vect{z}^\top (Q + 2Q\vect{\eta} Q)\vect{z} 
\end{align}
We now choose $\vect{\eta}^*$ so that $Q +2Q\vect{\eta} Q = \xi \Id$ for any fixed non-zero $c\in\amsmathbb{R}$, i.e. $\vect{\eta}^* = \half\left[Q^{-1}-\xi(Q^{-1})^2\right]$. With this, we effectively have
\begin{align}\label{tildeC}
    U(\beta, \gamma,\pmb{\eta}^*) =e^{-i\beta B}e^{-i\gamma \tilde{C}}
    \quad \text{ where } \quad
    \tilde{C} = \frac{\xi}{2}\sum_k z_k^2  
\end{align}
This amounts to setting $\cA = (1-\beta\gamma c)\amsmathbb{I}_d$ in \eqref{fCalc}.
\begin{align}
    f_C(\beta,\gamma,\pmb{\eta}^*,\sigma) = (1-\beta\gamma \xi)^2\left[\frac{\sigma^2}{2}\text{Tr}\;Q+\xv_0^\top Q\xv_0\right]+\frac{\beta^2}{2}\frac{\text{Tr}(Q)}{4\sigma^2}
\end{align}
Hence, choosing $\beta\gamma = 1/\xi$ and taking $\sigma\rightarrow\infty$ minimizes the cost. Alternatively, we could also take $\beta\rightarrow 0$ while holding $\beta\gamma = 1/\xi$ and $\sigma$ fixed. Hence, 1-step gradient-driven CV-QAOA is enough to minimize the cost even with \textit{non-degenerate} eigenvalues for $Q$.

\end{proof}

We  validate Theorem~\ref{thm:grad} by numerically simulating the gradient-driven CV-QAOA on a 2D convex quadratic function. In such a simulation, we simply modify the argument of the cost unitary so that it now uses $\tilde{C}$ rather than $C$, as in equation~\eqref{tildeC}. Consistent with theoretical expectations, we indeed find convergence to the true minimum with only $p=1$ CV-QAOA.
\begin{figure}[h!]
    \centering
    \includegraphics[width=\linewidth]{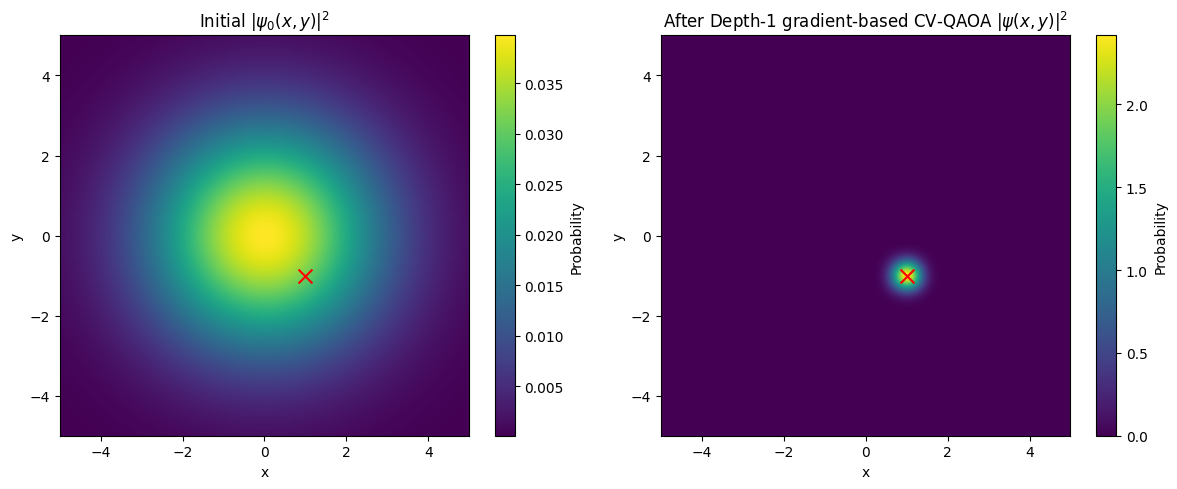}
    \caption{The above figure shows that the improved gradient-driven CV-QAOA is now able to converge to the minimum of the function $F(x, y) = \frac{1}{2}(x-1)^2+\frac{3}{2}(y+1)^2$ considered in Fig.~\ref{fig:Quad1} in 1-step. As before, red cross denotes the location of the true minimum, where the final probability distribution is sharply peaked and the inital distribution has $\sigma=2$. We take $\xi=2$.}
    \label{fig:Quad2}
\end{figure}

\section{CV-QAOA on real-analytic functions, and dequantizable regime}\label{app:PfGradMin}

In this section, we prove the two main results concerning 1-step CV-QAOA on real-analytic cost functions. First, we establish Theorems~\ref{thm:GradPointingToMin} and~\ref{thm:1DrealAnalytic}, showing that whenever the gradient of $C$ at the origin is proportional to the location of the global minimizer, a point-like Gaussian initial state and a single QAOA step suffice to approximate the global minimum arbitrarily well. Second, we prove Proposition~\ref{prop:dequant}, showing that this same point-like regime is classically simulable: the output state is well-approximated by a Gaussian, and the expected cost is reproduced by a quadratic truncation of $C$.

\subsection{Proof of Theorem~\ref{thm:GradPointingToMin} (and \ref{thm:1DrealAnalytic}): 1-step QAOA for real-analytic functions}
In this subsection, we provide a proof of Theorem \ref{thm:GradPointingToMin}, which explains when CV-QAOA works for arbitrary $d$-dimensional real analytic functions. Theorem~\ref{thm:1DrealAnalytic} follows as a direct consequence of Theorem~\ref{thm:GradPointingToMin} by restricting to the case with $d=1$.

\begin{proof}
We first prove the $d$-dimensional result. Set $a=-\lambda_*\nabla C(\vect 0)=\xv_*$, $r=\sigma+|\beta|/\sigma$, and write $\hat x_j=\sigma Y_j$ and $\hat p_j=\sigma^{-1}P_j$ in the unit-width Gaussian state $\phi$. Lemma~\ref{lem:px} gives $F(\beta,\lambda_*/\beta,\sigma)=\langle\phi,C(a+Q)\phi\rangle$, where
\begin{equation}
{Q_j=\sigma Y_j-\lambda_*\bigl(\partial_j C(\sigma Y)-\partial_j C(\vect 0)\bigr)+(\beta/\sigma)P_j.}
\end{equation}
{The $Q_j$ commute because $a_j+Q_j$ are unitary conjugates of the commuting position operators. By the fundamental theorem of calculus,}
\begin{equation}
{\partial_j C(\sigma y)-\partial_j C(\vect 0)
=\sigma\sum_k y_k\int_0^1\partial_{jk}C(t\sigma y)\,dt.}
\end{equation}
For $0<\sigma\le1$, the polynomial-growth assumption bounds every derivative of the coefficient on the right by $\sigma$ times a polynomial in $|y|$, uniformly in $\sigma$. Consequently, each application of a $Q_j$ to a Gaussian times a polynomially bounded smooth function contributes at most a factor $r$. Repeating this for any fixed product gives
\begin{equation}
\|Q_{j_1}\cdots Q_{j_n}\phi\|\le K_{n,C,d,\lambda_*}r^n,
\qquad n=1,2,\ldots .
\end{equation}
Since $a$ is a global minimizer, $\nabla C(a)=0$. Taylor's theorem near $a$, together with the global polynomial bound on $C$, gives some integer $m\ge1$ and a constant $K$ such that $0\le C(a+q)-C(a)\le K(|q|^2+|q|^{2m})$ for all $q\in\RR^d$. The joint spectral theorem for the commuting $Q_j$ and the preceding moment bounds therefore imply, for $r\le1$,
\begin{equation}
0\le F(\beta,\lambda_*/\beta,\sigma)-C(a)
\le K\langle |Q|^2+|Q|^{2m}\rangle_\phi
=O_{C,d,\lambda_*}(r^2)
=O_{C,d,\lambda_*}\!\left(\sigma^2+\beta^2/\sigma^2\right).
\end{equation}
For Theorem~\ref{thm:1DrealAnalytic}, use the same argument with $d=1$ and $a=-\lambda C'(0)$, which need not minimize $C$. The scalar Taylor bound now reads $|C(a+q)-C(a)-C'(a)q|\le K(q^2+|q|^{2m})$. The real Gaussian has $\langle P\rangle=\langle Y\rangle=0$, and Taylor expansion of $C'(\sigma Y)$ at zero, with Gaussian tails controlled by polynomial growth, gives $\langle Q\rangle=O_{C,\lambda}(\sigma^2)$. The remaining terms are $O_{C,\lambda}(r^2)$ by the same moment estimate. This proves both statements.
\end{proof}

In summary, we have shown that when the gradient of the cost at the origin is proportional to the location of the global minimum, then 1-step CV-QAOA with a point-like initial state can minimize the objective function arbitrarily well.

\subsection{Dequantization of 1-step CV-QAOA with a point-like initial state}\label{app:PfPropDequant}
Here we prove \cref{prop:dequant}, which shows that the output of 1-step CV-QAOA initialized with a point-like Gaussian state (small $\sigma$) can be efficiently simulated classically. Specifically, we show that both the expected cost value and the output wavefunction are well-approximated by those obtained under a quadratic truncation $\tilde{C}$ of the cost function. In the limit $\s\rightarrow 0$, the final state converges to a Gaussian state, which admits an efficient classical description.
Therefore, whenever the conditions of Theorems~\ref{thm:1DrealAnalytic} and \ref{thm:GradPointingToMin} apply, and the 1-step QAOA finds the global minima, there is also a way to efficiently simulate the QAOA output distribution on a classical computer.
Nevertheless, this is a regime where a quantum-classical query complexity separation can be obtained (see \cref{prop:query}).
\begin{proof}
Let $|\param\rangle\!\rangle$ denote the state obtained by replacing $C$ with $\tilde{C}$ in the CV-QAOA ansatz. We first consider the expected cost
\begin{align}
    \tilde{F}(\b,\g,\s) = \langle\!\langle\param| C(\hat{\xv}) |\param\rangle\!\rangle.
\end{align}
Applying Lemma~\ref{lem:px} (with $\lambda = \b\g$) to the truncated evolution gives
\begin{align}
    \tilde{F}(\b,\g,\s)=\bra{\s}C\!\left(\hat{\xv} - \lambda \nabla \tilde{C}(\hat{\xv})+\b\hat{\pv}\right)\ket{\s}.
\end{align}
Following the same steps as in the proof of Theorem~\ref{thm:GradPointingToMin}, one finds
\begin{align}
     \tilde{F}(\b,\g,\s) = C\!\left(-\lambda \nabla \tilde{C}(\vect{0})\right)+ O\!\left(\frac{d\b^2}{\s^2}+d\s^2\right).
\end{align}
Since $\nabla\tilde{C}(\vect{0}) = \nabla {C}(\vect{0})$, we have
\begin{align}
    |F(\b,\g,\s)-\tilde{F}(\b,\g,\s) | = O\!\left(\frac{d\b^2}{\s^2}+d\s^2\right),
\end{align}
which proves the first claim. Turning to the wavefunction approximation,
\begin{align}
    \bigl\|\,|\param\rangle - |\param\rangle\!\rangle\,\bigr\|
    = \left\|e^{-i\b B}e^{-i\gamma C}(1-e^{i\gamma (C-\tilde{C})})\ket{\s}\right\|
    = \left\|(1-e^{i\gamma (C-\tilde{C})})\ket{\s}\right\|,
\end{align}
where we have used the invariance of the $\ell_2$-norm under unitary operations. Applying the triangle inequality gives
\begin{align}
    \left\|(1-e^{i\gamma (C-\tilde{C})})\ket{\s}\right\| \leq \sum_{n=1}^\infty\frac{|\gamma|^n}{n!}\|(C-\tilde{C})^n\ket{\s}\|.
\end{align}
By Hermiticity of $C-\tilde{C}$, for any $n\geq 1$:
\begin{align}
    \|(C-\tilde{C})^n\ket{\s}\|=\sqrt{\bra{\s}(C-\tilde{C})^{2n}\ket{\s}}.
\end{align}
Since
\begin{align}
    C(\xv) - \tilde{C}(\xv) = \sum_{\ell_1+ \dots+\ell_d> 2}c(\ell_1, \dots, \ell_d)x_1^{\ell_1}\dots x_d^{\ell_d},
\end{align}
it follows that $\|(C-\tilde{C})^n\ket{\s}\| = O(d^{3n}\s^{3n})$,
where we pick up a factor of $\s$ for each $x_j$ in the integrand and a multiplicity of at most $d$ per occurrence, giving a scaling of $(d\s)^{6n}$ whose square-root yields the above. Hence, there exists $\mu>0$ such that
\begin{align}
    \|(1-e^{i\g (C-\tilde{C})})\ket{\s}\|\leq\mu  \sum_{n=1}^\infty\frac{(|\g|d^3\s^3)^{n}}{n!} = \mu\left(e^{|\g|d^3\s^3}-1\right) = O\left(|\g| (d\s)^3\right).
\end{align}
In the limit $\s\rightarrow 0$, the final state at the optimal parameters may therefore be approximated to arbitrarily small error by a Gaussian state, which is classically simulable. This concludes the proof.
\end{proof}

\section{Rotated square well: error bounds, runtime, query complexity}
\label{app:QAA}
In this section, we develop the technical estimates underlying our analysis of the quantum adiabatic algorithm for the rotated square-well problem with a harmonic mixer. First, we bound the change in the two lowest energies caused by Gaussian smoothing and show that a suitable smoothing width preserves the spectral gap. Second, using the exponentially accurate adiabatic theorem of \cite{HagedornJoye2002}, together with uniform spectral estimates and an endpoint-flat analytic schedule, we derive the runtime bound \(T=O(\log^8 d)\). Finally, using the factorization into one-dimensional evolutions and higher-order product formulas, we obtain the query bound \(Q(d)=d^{o(1)}\) for quantum simulation via CV-QAOA.
\subsection{Smoothing the square well}
\label{app:rsw-smoothing}

In this subsection, we smooth the square well and show that the resulting Hamiltonian retains the spectral properties needed for the adiabatic algorithm. We bound the change in its two lowest energies, which also bounds the change in its spectral gap. The argument has two steps: smoothing changes the potential by a small amount in integral, and states with bounded kinetic energy cannot concentrate enough probability near the well boundaries to turn this into a large energy change. The resulting estimate allows a smoothing width of order $\nu^2$, where $\nu$ is the endpoint cutoff used in the runtime analysis.

We round the well boundaries by averaging over Gaussian shifts. Writing $g_\eta(y)=e^{-y^2/\eta^2}/(\sqrt\pi\,\eta)$, define
\begin{equation}
w(z)=-\mathbf1_{\{|z|<\delta\}},
\qquad
w_\eta(z)=(g_\eta*w)(z)
=-\frac{1}{\sqrt\pi\,\eta}
\int_{-\delta}^{\delta}e^{-(z-y)^2/\eta^2}\,dy.
\end{equation}
Here $\delta>0$ is the half-width of the well and $\eta>0$ is the smoothing width; we set $w_0=w$. The full cost and Hamiltonian are
\begin{align}
W^\eta(\xv)&\equiv W_{R,\sv}^\eta(\xv)
=\sum_{j=1}^d w_\eta((R\xv+\sv)_j),\nonumber\\
H^\eta(s)&=(1-s)A+sW^\eta,
\qquad
A=-\frac12\Delta+\frac{|\xv|^2}{8\sigma^4}.
\label{eq:rsw-harmonic-path}
\end{align}
The rotation $R\in O(d)$ and shift $\sv$, with $\|\sv\|_\infty\le B$, are unknown. The initial state is the known Gaussian ground state $\ket{\sigma}$ of $A$, with $\braket{\xv|\sigma}=(2\pi\sigma^2)^{-d/4}e^{-|\xv|^2/(4\sigma^2)}$. Success means reaching $V_a=\{\xv:\|R\xv+\sv\|_\infty\le a\}$, where $a>\delta$. Throughout, $d\ge1$ and $0<\varepsilon<1/3$, while $\sigma,\delta,a,B$ are fixed independently of $d$ and $\varepsilon$. Unless stated otherwise, constants depend only on these fixed parameters. We write $W=W^0$ and $H=H^0$ for the sharp well.

The schedule $s_M(u)$ increases from $0$ to $s_f=1-\nu^2$, with $0<\nu<1$. In physical time, write $s(t)=s_M(t/T)$, and keep $\eta$ fixed during the evolution. Since $W^\eta$ is bounded and $1-s\ge\nu^2>0$, the Hamiltonians share the domain
\[
D(A)=H^2(\RR^d)\cap\{f:|\xv|^2f\in L^2(\RR^d)\}
\]
and generate unitary evolutions. Denote the state evolved from $\ket{\sigma}$ by $\Psi^\eta(t)$ and its final value by $\Psi^\eta=\Psi^\eta(T)$.

As in Section~\ref{sec:RSW}, the coordinates $\zv=R\xv+\sv$ separate the problem into $d$ independent one-dimensional evolutions. For a scalar shift $b$, let
\begin{equation}
h_b^\eta(s)=(1-s)A_b+sw_\eta,
\qquad
A_b=-\frac12\partial_z^2+\frac{(z-b)^2}{8\sigma^4}.
\label{eq:rsw-coordinate-reduction}
\end{equation}
The $j$th coordinate has $b=s_j$, so we write $h_j^\eta=h_{s_j}^\eta$ and $A_j=A_{s_j}$. Explicitly, the coordinate change
\[
(\mathcal{R}f)(\zv)=f(R^\top(\zv-\sv))
\]
gives
\[
\mathcal{R}H^\eta(s)\mathcal{R}^{-1}
=\sum_{j=1}^d h_j^\eta(s).
\]
Each factor starts in the ground state of $A_b$,
\[
\phi_b(z)=(2\pi\sigma^2)^{-1/4}e^{-(z-b)^2/(4\sigma^2)}.
\]
If $\psi_b^\eta(t)$ is its evolved state, then
\[
\mathcal{R}\Psi^\eta(t)
=\bigotimes_{j=1}^d\psi_j^\eta(t),
\]
where $\psi_j^\eta(t)=\psi_{s_j}^\eta(t)$ and $\phi_j=\phi_{s_j}$. We suppress time arguments for final states, writing $\psi_b^\eta=\psi_b^\eta(T)$ and $\psi_j^\eta=\psi_j^\eta(T)$; the sharp final states are $\psi_j=\psi_j^0$. All one-dimensional estimates below are uniform in $|b|\le B$. The coordinate change is used only in the analysis.

We first bound the total change in the potential. Shifting an interval by $y$ changes its indicator on a set of length $2\min\{|y|,2\delta\}$. Thus, for $v_\eta=w_\eta-w$, averaging over the Gaussian shifts gives
\begin{align}
\|v_\eta\|_1
&\le\int_\RR g_\eta(y)\|w(\,\cdot-y)-w\|_1\,dy\le2\int_\RR g_\eta(y)|y|\,dy
=\frac{2\eta}{\sqrt\pi}.
\label{eq:L1-error-clean}
\end{align}
For the later product-formula analysis, we also note that $\|w_\eta^{(m)}\|_\infty\le C_m\eta^{-m}$ for $m\ge1$. This follows by differentiating the Gaussian in the convolution: $\|w*g_\eta^{(m)}\|_\infty\le\|w\|_\infty\|g_\eta^{(m)}\|_1=C_m\eta^{-m}$, with $C_m$ independent of $\eta$.

To compare energies, it is convenient to remove the overall factor $1-s$ and make the potential nonnegative. Define $D_\eta=1+w_\eta$, $q=s/(1-s)$, and $K_b^\eta(q)=A_b+qD_\eta$, so that $h_b^\eta(s)=(1-s)(K_b^\eta(q)-qI)$. Here $0\le D_\eta\le1$ and $D_0=\mathbf1_{\{|z|\ge\delta\}}$. The scalar shift $-qI$ does not affect the gap. Let $E_j^\eta(q,b)$ denote the eigenvalues of $K_b^\eta(q)$, numbered from zero.

\begin{lemma}[Effect of smoothing on the lowest energies]
\label{lem:rsw-smoothing-spectrum}
There are constants $\kappa,C_0>0$ such that
\begin{equation}
|E_j^\eta(q,b)-E_j^0(q,b)|\le C_0q\eta,
\qquad j=0,1,
\label{eq:rsw-low-energy-smoothing}
\end{equation}
whenever $q\ge0$, $\eta\ge0$, and $q\eta\le\kappa$, uniformly in $|b|\le B$. The first two energies are uniformly bounded in this range.
\end{lemma}

\begin{proof}
First, bound the two lowest energies of the sharp well using trial states supported inside $(-\delta,\delta)$, where $D_0=0$. Take the first two eigenfunctions of the infinite square well on this interval and extend them by zero. Every normalized state in their span has kinetic energy at most $\pi^2/(2\delta^2)$. On the same interval, the oscillator potential is at most $(\delta+B)^2/(8\sigma^4)$. The variational principle therefore gives
\begin{equation}
E_j^0(q,b)\le E_*,
\qquad j=0,1,
\qquad
E_*:=\frac{\pi^2}{2\delta^2}+\frac{(\delta+B)^2}{8\sigma^4},
\label{eq:rsw-energy-ceiling}
\end{equation}
independently of $q$ and $b$.

Next, bound the change in the energy of any normalized finite-energy state $f$. The one-dimensional inequality $\|f\|_\infty^2\le\|f\|_2\|f'\|_2=\|f'\|_2$ follows by integrating the derivative of $|f|^2$ from both ends of the real line and applying Cauchy--Schwarz. Together with \eqref{eq:L1-error-clean}, it gives
\begin{equation}
\left|\braket{f|K_b^\eta(q)-K_b^0(q)|f}\right|
=q\left|\int_\RR v_\eta(z)|f(z)|^2\,dz\right|
\le\frac{2q\eta}{\sqrt\pi}\|f'\|_2.
\end{equation}
Both $K_b^0(q)$ and $K_b^\eta(q)$ have nonnegative potentials, so an energy bound also bounds the kinetic energy: $\|f'\|_2^2\le2\braket{f|K_b^\eta(q)|f}$, and likewise for the sharp well.

We now apply the variational principle in both directions. On the span of the first $j+1$ sharp eigenfunctions, every normalized $f$ satisfies $\|f'\|_2\le\sqrt{2E_*}$. Using this span as trial states for the smoothed Hamiltonian gives
\[
E_j^\eta(q,b)\le E_j^0(q,b)
+\frac{2\sqrt{2E_*}}{\sqrt\pi}q\eta,
\qquad j=0,1.
\]
Choose $\kappa$ small enough that the added term is at most $1$. Then $E_j^\eta(q,b)\le\bar E:=E_*+1$, so every normalized state in the span of the first $j+1$ smoothed eigenfunctions satisfies $\|f'\|_2\le\sqrt{2\bar E}$. Using this span as trial states for the sharp Hamiltonian gives the reverse bound,
\[
E_j^0(q,b)\le E_j^\eta(q,b)
+\frac{2\sqrt{2\bar E}}{\sqrt\pi}q\eta.
\]
The two inequalities prove the claim with $C_0=2\sqrt{2\bar E}/\sqrt\pi$.
\end{proof}

The role of this estimate is to preserve the gap after smoothing. Once the sharp gap is bounded below by $\gamma_*>0$, as shown in Lemma~\ref{lem:gap-all-s}, the lemma above gives $\operatorname{gap}K_b^\eta(q)\ge\gamma_*-2C_0q\eta$. Along the schedule, $q\le\nu^{-2}-1$, so choosing $\eta\le c_\eta\nu^2$ with $c_\eta$ sufficiently small keeps this gap at least $\gamma_*/2$. Multiplying by $1-s\ge\nu^2$ then gives $\operatorname{gap}h_b^\eta(s)\ge\gamma_*\nu^2/2$. The uniform energy bound $\bar E$ will also control the probability outside the target interval in the next subsection.

\subsection{Quantum adiabatic runtime with an analytic schedule}
\label{app:rsw-runtime}

We show that both the sharp and sufficiently smoothed evolutions succeed
in time $O(\log^8(d/\varepsilon))$. The argument bounds the gap throughout
the evolution, the ground-state probability outside the target interval
at its endpoint, and the error incurred by following that ground state.
We use the polynomial schedule
\begin{equation}
s_M(u)
=s_f\frac{\int_0^u v^M(1-v)^M\,dv}
          {\int_0^1 v^M(1-v)^M\,dv},
\qquad
s_f=1-\nu^2,
\qquad
s(t)=s_M(t/T),
\label{eq:beta-schedule}
\end{equation}
where $M\ge1$ is chosen below. This schedule is increasing on $[0,1]$,
has degree $2M+1$, and satisfies
\[
s_M(0)=0,\qquad s_M(1)=s_f,\qquad
s_M^{(j)}(0)=s_M^{(j)}(1)=0
\quad(1\le j\le M).
\]

\begin{thm}[Infinite-volume runtime]
\label{thm:rsw-harmonic-runtime}
There is a constant $c_\eta>0$ such that the evolution
\eqref{eq:rsw-harmonic-path}, initialized in $\ket{\sigma}$,
admits a schedule \eqref{eq:beta-schedule} for which
\begin{equation}
\int_{V_a}|\Psi^\eta(\xv)|^2\,d\xv
\ge1-\varepsilon/2,
\qquad
T=O\left(\log^8\frac{d}{\varepsilon}\right).
\label{eq:rsw-infinite-volume-success}
\end{equation}
The guarantee holds uniformly for $0\le\eta\le c_\eta\nu^2$,
$R\in O(d)$, and $\|\sv\|_\infty\le B$, with $\eta$ held fixed during
the evolution. The schedule and runtime depend only on $d$,
$\varepsilon$, and the fixed model parameters.
\end{thm}

\subsubsection{Gap and localization}

\begin{lemma}[Uniform gap]
\label{lem:gap-all-s}
There is a constant $\gamma_*>0$ such that
\begin{equation}
\operatorname{gap}K_b^0(q)\ge\gamma_*
\qquad(q\ge0,\ |b|\le B).
\label{eq:rsw-rescaled-gap}
\end{equation}
For sufficiently small $c_\eta>0$, this implies
\begin{equation}
\operatorname{gap}h_b^\eta(s)
\ge(1-s)\gamma_*/2
\ge\nu^2\gamma_*/2
\label{eq:gap-all-s}
\end{equation}
throughout $0\le s\le1-\nu^2$, whenever
$0\le\eta\le c_\eta\nu^2$.
\end{lemma}

\begin{proof}
The spectrum of $K_b^0(q)$ is discrete because its potential grows
quadratically at infinity. Its eigenvalues are nondegenerate: for two
square-integrable solutions $f,g$ at the same energy,
$(fg'-f'g)'=0$, and decay at infinity forces this constant to vanish,
so the solutions are proportional. Let $E_j^D(b)$ be the eigenvalues
of $A_b$ on $I=(-\delta,\delta)$ with wavefunctions vanishing at both
endpoints. Extending these wavefunctions by zero gives admissible trial
states for the sharp well, so
\[
E_j^0(q,b)\le E_j^D(b)\le E_*,
\qquad j=0,1.
\]
To compare the energies as $q$ grows, take normalized $f$ in the span
of the first $j+1$ eigenfunctions of $K_b^0(q)$. Its energy bounds both
the probability outside the well and the kinetic energy:
\[
\|f\|_{L^2(I^c)}^2\le E_*/q,
\qquad
\|f'\|_2^2\le2E_*.
\]
The boundary values consequently satisfy
\[
|f(-\delta)|^2+|f(\delta)|^2
\le2\|f\|_{L^2(I^c)}\|f'\|_{L^2(I^c)}
\le Cq^{-1/2}.
\]
Subtract the straight line joining these two values from $f|_I$.
The resulting function $g$ vanishes at both endpoints and obeys
$\|g-f|_I\|_{H^1(I)}\le Cq^{-1/4}$.
Since the oscillator potential is bounded on $I$, this correction
changes the energy by at most $Cq^{-1/4}$ and gives
$\|g\|_2\ge1-Cq^{-1/4}$. For sufficiently large $q$, it also preserves
the number of linearly independent trial states. The variational
principle therefore gives
\[
0\le E_j^D(b)-E_j^0(q,b)\le Cq^{-1/4},
\qquad j=0,1,
\]
uniformly in $|b|\le B$.

The first two energies are continuous in $q$ and $b$. For the shift
dependence, use
\[
A_b-A_{b'}
=\frac{(b'-b)z}{4\sigma^4}+\frac{b^2-b'^2}{8\sigma^4};
\]
the energy bound controls the expectation of $|z|$ on the relevant
trial states, so the variational comparison is continuous in $b$.
The gap $E_1^D(b)-E_0^D(b)$ is positive and continuous on $[-B,B]$,
and hence has a positive minimum. The uniform convergence above
therefore bounds the sharp gap away from zero at sufficiently large
$q$. On the remaining bounded range of $q$ and $b$, continuity and
nondegeneracy give another positive minimum. This proves
\eqref{eq:rsw-rescaled-gap}. Finally, choose
\[
0<c_\eta\le
\min\left\{\kappa,\frac{\gamma_*}{4C_0},1,\delta\right\}.
\]
For $q\le\nu^{-2}-1$ and $\eta\le c_\eta\nu^2$,
Lemma~\ref{lem:rsw-smoothing-spectrum} gives
\[
\operatorname{gap}K_b^\eta(q)
\ge\gamma_*-2C_0q\eta\ge\gamma_*/2.
\]
Multiplication by $1-s$ gives \eqref{eq:gap-all-s}.
\end{proof}

Fix this choice of $c_\eta$ for the rest of the analysis. In particular,
$E_j^\eta(q,b)\le\bar E$ for $j=0,1$ throughout the schedule.

\begin{lemma}[Endpoint localization]
\label{lem:deep-well}
There are constants $\nu_0\in(0,1/2]$ and $C_{\mathrm{loc}}>0$
such that the normalized ground state $\varphi_{b,f}^\eta$
of $h_b^\eta(s_f)$ satisfies
\begin{equation}
\int_{|z|>a}|\varphi_{b,f}^\eta(z)|^2\,dz
\le C_{\mathrm{loc}}\nu
e^{-\sqrt2(a-\delta)/\nu}
\label{eq:tail-bound-HJ}
\end{equation}
for $0<\nu\le\nu_0$, $0\le\eta\le c_\eta\nu^2$, and $|b|\le B$.
\end{lemma}

\begin{proof}
Outside the well, $D_\eta(z)\ge1/2$: when $|z|\ge\delta$,
at most half the Gaussian mass lies inside $(-\delta,\delta)$.
This also holds for $\eta=0$. Let $\varphi$ be the positive normalized
ground state of $K_b^\eta(q)$ at $q=\nu^{-2}-1$, with energy
$E\le\bar E$. Its kinetic energy gives
\[
\|\varphi\|_\infty^2
\le\|\varphi\|_2\|\varphi'\|_2
\le\sqrt{2\bar E}.
\]
For $|z|\ge\delta$, the eigenvalue equation implies
\[
\varphi''(z)
=2\left(\frac{(z-b)^2}{8\sigma^4}+qD_\eta(z)-E\right)\varphi(z)
\ge(q-2\bar E)\varphi(z).
\]
Choose
\[
\nu_0=\min\left\{\frac12,\,[2(1+2\bar E)]^{-1/2}\right\}.
\]
Then $q-2\bar E\ge(2\nu^2)^{-1}$. Comparison with the decaying
exponential gives
\[
\varphi(z)\le
\varphi(\delta)e^{-(z-\delta)/(\sqrt2\nu)}
\qquad(z\ge\delta),
\]
and the corresponding bound on the left. Indeed, a positive interior
maximum of the difference would contradict the differential inequality,
while the difference vanishes at the boundary and tends to zero at
infinity. Integrating the squared tails yields
\[
\int_{|z|>a}|\varphi(z)|^2\,dz
\le2\sqrt{\bar E}\,\nu
e^{-\sqrt2(a-\delta)/\nu}.
\]
Thus one may take $C_{\mathrm{loc}}=2\sqrt{\bar E}$.
\end{proof}

\subsubsection{The adiabatic theorem and endpoint conditions}

Writing $u=t/T$, the one-dimensional evolution satisfies
\begin{equation}
iT^{-1}\partial_u\psi_b^\eta(Tu)
=h_b^\eta(s_M(u))\psi_b^\eta(Tu),
\qquad
\psi_b^\eta(0)=\phi_b.
\label{eq:rsw-scaled-evolution}
\end{equation}
We use the exponential adiabatic theorem of Hagedorn and
Joye~\cite{HagedornJoye2002}. Its superadiabatic approximation is the
instantaneous eigenvector with corrections in powers of $1/T$:
\begin{equation}
\Xi_N(u)
=e^{-iT\int_0^uE(v)\,dv}
\left(
\varphi(u)+\sum_{n=1}^NT^{-n}\psi_n(u)
+T^{-N-1}\psi_{N+1}^\perp(u)
\right).
\label{eq:rsw-HJ-expansion}
\end{equation}
Here $\varphi(u)$ is normalized for real $u$, and its phase is chosen
so that $\braket{\varphi|\varphi'}=0$, the parallel-transport
gauge. The recursively defined corrections have the form
$\psi_n=f_n\varphi+\psi_n^\perp$, with
$\psi_n^\perp\perp\varphi$ for real $u$.
An analytic family of type A has a fixed operator domain on which
$\widehat H(u)f$ is analytic in $u$ for every vector $f$.

\begin{thm}[Exponentially accurate adiabatic theorem
{\cite{HagedornJoye2002}}]
\label{thm:HJ-original}
Let $\widehat H(u)$ be self-adjoint on a common dense domain for
$u\in[0,1]$ and admit an analytic extension of type A to a complex
neighborhood of this interval. Suppose $E(u)$ is a nondegenerate
eigenvalue separated from the remaining spectrum by a positive
uniform gap. There are constants $G>0$, $C(g)>0$, and $\Gamma(g)>0$
such that, for $0<g<G$, $T\ge1$, and $N=\lfloor gT\rfloor$, the solution
\[
iT^{-1}\partial_u\psi(u)=\widehat H(u)\psi(u),
\qquad
\psi(0)=\Xi_N(0),
\]
satisfies
\begin{equation}
\sup_{u\in[0,1]}
\|\psi(u)-\Xi_N(u)\|
\le C(g)e^{-\Gamma(g)T}.
\label{eq:HJ-original-bound}
\end{equation}
\end{thm}

This is the finite-interval version discussed in Remark 3 following
Theorem 1.1 of \cite{HagedornJoye2002}. The initial vector in this
theorem need not equal the Gaussian we can prepare. Vanishing schedule
derivatives at the endpoints remove the corrections responsible for
this mismatch.

\begin{lemma}[Endpoint flatness]
\label{lem:endpoint-flat-schedule}
Let $\widehat H(u)=h_b^\eta(s_M(u))$ satisfy the hypotheses of
Theorem~\ref{thm:HJ-original}, with $\eta$ fixed. If
\[
s_M^{(j)}(0)=s_M^{(j)}(1)=0,
\qquad 1\le j\le N+1,
\]
then
\begin{equation}
\Xi_N(0)=\varphi(0),
\qquad
\Xi_N(1)
=e^{-iT\int_0^1E(v)\,dv}\alpha_N\varphi(1)
\label{eq:rsw-endpoint-matching}
\end{equation}
for a scalar $\alpha_N$. If the exact evolution starts from
$\varphi(0)$ and $\|\psi(1)-\Xi_N(1)\|\le r<1$, then
\[
\inf_{\vartheta\in\RR}
\|\psi(1)-e^{i\vartheta}\varphi(1)\|
\le2r.
\]
\end{lemma}

\begin{proof}
Since $\widehat H(u)=A_b+s_M(u)(w_\eta-A_b)$, its derivatives through
order $N+1$ vanish at either endpoint $u_*\in\{0,1\}$.
Differentiating $(\widehat H-E)\varphi=0$ and inducting on the derivative
order gives
\[
E^{(m)}(u_*)=0,
\qquad
\varphi^{(m)}(u_*)=0,
\qquad 1\le m\le N+1.
\]
Indeed, once the lower derivatives vanish, the order-$m$ equation is
\[
(\widehat H(u_*)-E(u_*))\varphi^{(m)}(u_*)
=E^{(m)}(u_*)\varphi(u_*).
\]
Its inner product with $\varphi(u_*)$ gives $E^{(m)}(u_*)=0$.
Nondegeneracy then makes $\varphi^{(m)}(u_*)$ proportional to
$\varphi(u_*)$, while differentiating
$\braket{\varphi|\varphi'}=0$ another $m-1$ times makes that
coefficient zero. Let $P$ project onto $\varphi$, put $Q=I-P$,
and let $S$ be the reduced resolvent: the inverse of $\widehat H-E$ on the
subspace orthogonal to $\varphi$, extended by zero on $\varphi$.
The Hagedorn--Joye recurrence is
\begin{align}
\psi_1^\perp&=iS\varphi',
\nonumber\\
f_n(u)&=-\int_0^u
\braket{\varphi(v)|(\psi_n^\perp)'(v)}\,dv,
\nonumber\\
\psi_{n+1}^\perp
&=iS\bigl(f_n\varphi'+Q(\psi_n^\perp)'\bigr),
\qquad n\ge1.
\label{eq:rsw-HJ-recurrence}
\end{align}
Induction in this recurrence gives
\[
(\psi_n^\perp)^{(k)}(u_*)=0
\qquad(n\ge1,\ k\ge0,\ n+k\le N+1).
\]
For $n=1$, this follows from the vanishing derivatives of $\varphi$.
At the next order, every differentiated term contains either a
positive-order derivative of $\varphi$ or a derivative of
$\psi_n^\perp$ already covered by the induction hypothesis.
Thus all the orthogonal corrections in
\eqref{eq:rsw-HJ-expansion}, including $\psi_{N+1}^\perp$, vanish
at both endpoints. Since $f_n(0)=0$, all corrections vanish initially.
At the final endpoint the coefficients multiplying the ground state
may remain, giving
\[
\alpha_N=1+\sum_{n=1}^NT^{-n}f_n(1).
\]
Normalization of the exact state implies
$\bigl||\alpha_N|-1\bigr|\le r$. Absorbing the phase of the nonzero
scalar $\alpha_N$ into $\varphi(1)$ and applying the triangle
inequality proves the bound $2r$.
\end{proof}

The schedule \eqref{eq:beta-schedule} satisfies this lemma whenever
$M\ge N+1$.

\subsubsection{Bounds for the adiabatic error}

The exponential estimate depends on high derivatives of the
instantaneous ground state. We bound these derivatives by continuing
the Hamiltonian a short distance into the complex parameter plane.
The following estimates are uniform in the allowed smoothing widths;
they use the size of $D_\eta$, not its spatial derivatives.

\begin{lemma}[Complex spectral neighborhood]
\label{lem:isolation-strip}
Let $P_b^\eta(s)$ and $S_b^\eta(s)$ be the ground-state projector and
reduced resolvent of $h_b^\eta(s)$. For every real
$s_0\in[0,s_f]$, these operators and the ground-state eigenvalue extend
analytically to
\[
|s-s_0|<c_*(1-s_0),
\]
where $c_*>0$ is fixed. On this neighborhood,
\begin{align}
\|P_b^\eta(s)\|+\|I-P_b^\eta(s)\|&\le C,
&\|S_b^\eta(s)\|&\le C/(1-s_0),&
\|P_b^\eta(s)-P_b^\eta(s_0)\|
&\le C|s-s_0|/(1-s_0).
\label{eq:rsw-complex-spectral-bounds}
\end{align}
The continued eigenvalue remains nondegenerate, and all constants
are uniform for $|b|\le B$ and $0\le\eta\le c_\eta\nu^2$.
\end{lemma}

\begin{proof}
Adding $sI$ changes neither the eigenvectors nor the gap, and gives
the positive operator
\[
G_b^\eta(s)=h_b^\eta(s)+sI=(1-s)A_b+sD_\eta
\]
for real $s$. Write $\alpha=1-s_0$, $G_0=G_b^\eta(s_0)$,
and let $F_0$ be its ground energy. The preceding estimates give
\[
G_0\ge\alpha/(4\sigma^2),
\qquad
F_0\le\alpha\bar E,
\qquad
\operatorname{gap}G_0\ge\alpha\gamma_*/2.
\]
Fix $0<r<\gamma_*/8$. On the annulus
$r\alpha/2\le|\zeta-F_0|\le r\alpha$,
\[
\|(I-\zeta G_0^{-1})^{-1}\|
=\sup_{e\in\operatorname{spec}(G_0)}\frac{e}{|e-\zeta|}
\le C.
\]
Although $A_b$ is unbounded, measuring its size relative to $G_0$
gives the bounded operator
\[
B_0=G_0^{-1/2}(D_\eta-A_b)G_0^{-1/2},
\qquad
\|B_0\|\le\frac{1+4\sigma^2}{\alpha}.
\]
Here the expression is defined through its bounded quadratic form;
the estimate follows from $A_b\le G_0/\alpha$ and
$0\le D_\eta\le I$. The corresponding inverse is
\[
(G_b^\eta(s)-\zeta)^{-1}
=G_0^{-1/2}
\bigl[I-\zeta G_0^{-1}+(s-s_0)B_0\bigr]^{-1}
G_0^{-1/2}.
\]
For sufficiently small $c_*$, $|s-s_0|<c_*\alpha$ implies
\[
\|(s-s_0)(I-\zeta G_0^{-1})^{-1}B_0\|\le\frac12.
\]
The inverse in brackets is therefore bounded by
$2\|(I-\zeta G_0^{-1})^{-1}\|$, giving
$\|(G_b^\eta(s)-\zeta)^{-1}\|\le C/\alpha$.
Taking $c_*<1/2$ also ensures $\operatorname{Re}(1-s)>0$,
so these operators retain the common domain $D(A_b)$.

The inner and outer circles continue to enclose exactly one
eigenvalue, denoted $F(s)$: their spectral projections depend
continuously on $s$ and have rank one at $s_0$.
In particular, $|F(s)-F_0|<r\alpha/2$.
Writing $\mathcal C=\{\zeta:|\zeta-F_0|=r\alpha\}$ with positive
orientation, the formulas
\[
P_b^\eta(s)
=-\frac{1}{2\pi i}\oint_{\mathcal C}
(G_b^\eta(s)-\zeta)^{-1}\,d\zeta,
\qquad
S_b^\eta(s)
=\frac{1}{2\pi i}\oint_{\mathcal C}
\frac{(G_b^\eta(s)-\zeta)^{-1}}{\zeta-F(s)}\,d\zeta
\]
give the projector and reduced-resolvent bounds.
The inverse factorization also gives
\[
\|(G_b^\eta(s)-\zeta)^{-1}-(G_0-\zeta)^{-1}\|
\le C|s-s_0|/\alpha^2.
\]
Integrating this difference proves the last bound in
\eqref{eq:rsw-complex-spectral-bounds}. Decrease $c_*$ once more
so that the projector difference is less than $1/4$.
\end{proof}

For the schedule parameter, define
\begin{equation}
\rho=c_a\frac{\nu^2}{\sqrt{M+1}},
\qquad
\Omega_r=\{u\in\mathbb C:
\operatorname{dist}(u,[0,1])<r\},
\label{eq:strip-width}
\end{equation}
where $c_a>0$ is a sufficiently small fixed constant.

\begin{cor}[Analytic bounds along the schedule]
\label{cor:rsw-scheduled-spectral-data}
The family $h_b^\eta(s_M(u))$ is analytic of type A on
$\Omega_{4\rho}$, with domain $D(A_b)$ and a nondegenerate eigenvalue
continued from its real ground energy. On $\Omega_\rho$, its
eigenvector $\varphi$, projectors $P,Q=I-P$, and reduced resolvent
$S$ satisfy
\begin{equation}
\|S\|\le C\nu^{-2},
\qquad
\|P\|+\|Q\|+\|\varphi\|\le C,
\qquad
\|\varphi'\|\le C/\rho.
\label{eq:C1-bound}
\end{equation}
For real $u$, $\varphi(u)$ is the positive normalized ground state
and satisfies $\braket{\varphi|\varphi'}=0$.
The bounds are uniform in $b$ and the allowed $\eta$.
\end{cor}

\begin{proof}
Let $I_M=\int_0^1[v(1-v)]^M\,dv$. Integrating over an interval
of width $c/\sqrt{M+1}$ around $1/2$ gives
$I_M\ge c4^{-M}/\sqrt{M+1}$. Hence
\[
s_M'(u)=s_f\frac{[u(1-u)]^M}{I_M}
\]
satisfies $|s_M'(u)|\le C\sqrt{M+1}$ within a complex distance
$c/\sqrt{M+1}$ of $[0,1]$, where
$|u(1-u)|\le1/4+C/(M+1)$.
For $u\in\Omega_{4\rho}$ and a nearest real point $u_0\in[0,1]$,
choosing $c_a$ sufficiently small therefore gives
\[
|s_M(u)-s_M(u_0)|
\le c_*\nu^2
\le c_*[1-s_M(u_0)].
\]
Lemma~\ref{lem:isolation-strip} supplies the analytic eigenvalue
and the bounds on $P,Q,S$. It also ensures that $1-s_M(u)$
does not vanish, so the operator domain remains $D(A_b)$.

To normalize the eigenvector without losing analyticity, use
$\braket{\overline f|g}=\int_\RR f(z)g(z)\,dz$, where
$\overline f$ denotes pointwise complex conjugation. The conjugation
cancels that in the bra, so this pairing is analytic in both vectors.
The Hamiltonian and its projectors are symmetric for this pairing.
At a real center, let $\varphi_0$ be the positive normalized ground
state and set locally
\[
v(u)=P(u)\varphi_0,
\qquad
\varphi(u)=\frac{v(u)}{\braket{\overline{v(u)}|v(u)}^{1/2}}.
\]
Since $\braket{\overline v|v}=\braket{\varphi_0|P(u)|\varphi_0}$ is within $1/4$ of $1$,
the square root positive at the real center gives a bounded analytic
vector. These choices agree on real overlaps, and therefore give
one analytic vector throughout the neighborhood, with
\[
\braket{\overline\varphi|\varphi}=1,
\qquad
\braket{\overline\varphi|\varphi'}=0,
\qquad
P=\ket{\varphi}\bra{\overline\varphi}.
\]
Cauchy's derivative estimate gives $\|\varphi'\|\le C/\rho$
on $\Omega_\rho$. Finally, if $F$ is the continued ground energy
of $G_b^\eta$, then
\[
(1-s_M)A_b\varphi=(F-s_MD_\eta)\varphi,
\qquad
(1-s_M)A_bS=Q+(F-s_MD_\eta)S.
\]
Thus $A_b\varphi$ and $A_bS$ also depend analytically on $u$.
This justifies differentiating the eigenvector and the adiabatic
corrections within the domain of the unbounded Hamiltonian.
\end{proof}

The preceding bounds determine both the exponential rate and its
prefactor. Write
\begin{equation}
\mathcal K_{\nu,M}
=C_{\mathrm{ad}}(M+1)\nu^{-6},
\qquad
g=\frac{1}{e\mathcal K_{\nu,M}},
\label{eq:rsw-HJ-scale}
\end{equation}
where $C_{\mathrm{ad}}\ge1$ is a sufficiently large constant depending
only on the fixed model parameters.

\begin{lemma}[Quantitative adiabatic estimate]
\label{lem:gamma-nu-M}
The constants in Theorem~\ref{thm:HJ-original} can be chosen so that
\begin{equation}
g<G,
\qquad
\Gamma(g)\ge\frac{1}{e\mathcal K_{\nu,M}},
\qquad
C(g)\le\mathcal K_{\nu,M}.
\label{eq:gamma-nu-M}
\end{equation}
These bounds are uniform in $b$ and the allowed smoothing widths.
\end{lemma}

\begin{proof}
Use constants $C_1,C_2,C_3,C_4$ to bound respectively
$\|S\|$, both $\|P\|$ and $\|Q\|$, $\|\varphi\|$, and
$\|\varphi'\|$ on $\Omega_\rho$, with $C_3\ge1$.
The constants in \cite[Sections~4--5]{HagedornJoye2002} are
\begin{align}
C_5&=C_1(2C_3C_4+C_2),
&
C_6&=\frac{2C_1C_4}{\rho},
&
C_7&=\frac{2eC_5}{\rho},
\label{eq:HJ-constant-chain}\\
G&=C_7^{-1},
&
\Gamma(g)&=g|\log(C_7g)|,
&
C(g)&=C_6e^{|\log(C_7g)|}.
\nonumber
\end{align}
The factor $2$ in $C_5$ bounds the integration-path length when
$\rho\le1$. To apply their derivative estimates on the finite
neighborhood $\Omega_\rho$, put
\[
r(u)=\rho-\operatorname{dist}(u,[0,1]),
\qquad
\beta_0=1,\qquad \beta_k=k^k\quad(k\ge1).
\]
Cauchy's estimate on a circle of radius $r(u)/(k+1)$ gives
\[
\|f(u)\|\le A\beta_k r(u)^{-k}
\quad\Longrightarrow\quad
\|f'(u)\|\le A\beta_{k+1}r(u)^{-k-1},
\]
with a limiting argument for $k=0$. Along the integration segment
from $0$ to $u$,
\[
\operatorname{dist}(\tau u,[0,1])
\le\tau\operatorname{dist}(u,[0,1]),
\qquad 0\le\tau\le1,
\]
so $r(\tau u)\ge r(u)$, and the segment has length at most $2$.
Their recursive estimates therefore apply on $\Omega_\rho$
with \eqref{eq:HJ-constant-chain}; the scalar products are
analytically continued using the normalization above.

By \eqref{eq:C1-bound} and \eqref{eq:strip-width},
\[
C_1=O(\nu^{-2}),
\qquad
C_2,C_3=O(1),
\qquad
C_4=O(\rho^{-1}),
\]
and consequently
\[
C_5=O(\nu^{-2}\rho^{-1}),
\qquad
C_6,C_7=O(\nu^{-2}\rho^{-2})
=O((M+1)\nu^{-6}).
\]
Choose $C_{\mathrm{ad}}$ so that $C_7\le\mathcal K_{\nu,M}$.
Then $C_7g\le1/e$, giving $g<G$ and $\Gamma(g)\ge g$.
Moreover,
\[
\frac{C_6}{C_7}
=\frac{C_4}{e(2C_3C_4+C_2)}
\le\frac{1}{2e},
\qquad
C(g)=\frac{C_6}{C_7g}
\le\frac{1}{2eg}
\le\mathcal K_{\nu,M}.
\]
\end{proof}

Combining this estimate with endpoint matching gives
\begin{equation}
\inf_{\vartheta\in\RR}
\|\psi_b^\eta-e^{i\vartheta}\varphi_{b,f}^\eta\|_2
\le2\mathcal K_{\nu,M}
e^{-T/(e\mathcal K_{\nu,M})}
\label{eq:1d-state-bound-HJ}
\end{equation}
when $N=\lfloor T/(e\mathcal K_{\nu,M})\rfloor\ge1$,
$M\ge N+1$, and
$\mathcal K_{\nu,M}e^{-T/(e\mathcal K_{\nu,M})}<1$.

\begin{proof}[Proof of Theorem~\ref{thm:rsw-harmonic-runtime}]
Choose the endpoint cutoff by setting
\begin{equation}
\ell=\log\frac{16(1+C_{\mathrm{loc}})d}{\varepsilon},
\qquad
\nu=\min\left\{\nu_0,\frac{\sqrt2(a-\delta)}{\ell}\right\}.
\label{eq:nu-choice}
\end{equation}
Lemma~\ref{lem:deep-well} then bounds the probability outside
$[-a,a]$ in each final ground state by
\[
C_{\mathrm{loc}}\nu e^{-\sqrt2(a-\delta)/\nu}
\le C_{\mathrm{loc}}\nu e^{-\ell}
\le\frac{\varepsilon}{16d}.
\]
For the adiabatic error, choose
\begin{align}
N&=\left\lceil
2\log\frac{16C_{\mathrm{ad}}d}{\varepsilon\nu^6}
\right\rceil,
&
M&=N+1,
\label{eq:rsw-orders}\\
T&=e\mathcal K_{\nu,M}N.
\label{eq:T-and-M}
\end{align}
These choices give $gT=N$ exactly and provide the required
endpoint flatness. The superadiabatic error obeys
\[
\mathcal R_N
:=\sup_{u\in[0,1]}
\|\psi_b^\eta(Tu)-\Xi_N(u)\|_2
\le\mathcal K_{\nu,M}e^{-N}
\le\frac{\varepsilon}{16d}.
\]
For the last inequality, write
$x=\log(16C_{\mathrm{ad}}d/(\varepsilon\nu^6))\ge2$.
Then $N=\lceil2x\rceil$, $M+1\le2x+3$, and
$(2x+3)e^{-x}\le1$, giving
$\mathcal K_{\nu,M}e^{-N}\le
C_{\mathrm{ad}}\nu^{-6}e^{-x}=\varepsilon/(16d)$.
Lemma~\ref{lem:endpoint-flat-schedule} therefore bounds the final
one-dimensional state error, up to phase, by $2\mathcal R_N$.

For normalized states $f,g$ and an orthogonal projector $\Pi$,
\begin{equation}
|\braket{f|\Pi|f}-\braket{g|\Pi|g}|
\le2\|f-g\|_2.
\label{eq:rsw-probability-transfer}
\end{equation}
Thus each evolved coordinate has failure probability at most
$\varepsilon/(16d)+4\mathcal R_N\le5\varepsilon/(16d)$.
The product structure gives
\begin{equation}
\int_{V_a}|\Psi^\eta(\xv)|^2\,d\xv
=\prod_{j=1}^d
\int_{-a}^a|\psi_j^\eta(z)|^2\,dz.
\label{eq:factorized-success}
\end{equation}
The probability that any coordinate fails is at most the sum of
its failure probabilities, hence at most $5\varepsilon/16<
\varepsilon/2$. Finally, $\nu^{-1}=O(\ell)$ and $N,M=O(\ell)$,
so $\mathcal K_{\nu,M}=O(\ell^7)$ and
$T=e\mathcal K_{\nu,M}N=O(\ell^8)$.
All choices are uniform in $\eta$, $R$, and the allowed shifts.
\end{proof}

\subsection{RsAA/QHD query complexity via product formulas}
\label{app:rsw-product-formulas}

In this subsection, we derive a phase-query bound of $Q(d)=d^{o(1)}$ for simulating the smoothed adiabatic evolution with fixed target error. We use higher-order product formulas that alternate the known harmonic mixer $A$ with the cost $W^\eta$ from \eqref{eq:rsw-harmonic-path}. Each $e^{-i\gamma W^\eta}$, for real $\gamma$, counts as one phase query; mixer evolution requires no cost queries. We work on $L^2(\RR^d)$ with the fixed model parameters of Theorem~\ref{thm:rsw-harmonic-runtime}.

Fix $\varepsilon\in(0,1/3)$. Choose $\nu,N,M,T$ as in \eqref{eq:nu-choice}, \eqref{eq:rsw-orders},
and \eqref{eq:T-and-M}, and set $\eta=c_\eta\nu^2$, held fixed during
the evolution. Theorem~\ref{thm:rsw-harmonic-runtime} gives success
probability at least $1-\varepsilon/2$ for the exact smoothed evolution.
Its parameter choices imply
\begin{equation}
\eta^{-1}=O(\ell^2),
\qquad M=O(\ell),
\qquad T=O(\ell^8),
\qquad \ell=\log\frac{Cd}{\varepsilon},
\label{eq:pf-smoothed-parameters}
\end{equation}
where $C=16(1+C_{\mathrm{loc}})$ as in \eqref{eq:nu-choice}.
Write $s(t)=s_M(t/T)$. In the hidden coordinates of
\eqref{eq:rsw-coordinate-reduction}, let $U_b^\eta(t_2,t_1)$ be the
propagator generated by $(1-s(t))A_b+s(t)w_\eta$, and define
\begin{equation}
\psi_b^\eta(t)=U_b^\eta(t,0)\phi_b,
\qquad
\phi_b(z)=\frac{1}{(2\pi\sigma^2)^{1/4}}
e^{-(z-b)^2/(4\sigma^2)}.
\end{equation}

Fix an even order $k\ge2$, meaning agreement with the exact Taylor
expansion through order $k$. The midpoint Strang step is
\begin{equation}
S_{2,b}(t+h,t)
=e^{-\frac{ih}{2}(1-s(t+h/2))A_b}
 e^{-ihs(t+h/2)w_\eta}
 e^{-\frac{ih}{2}(1-s(t+h/2))A_b}.
\label{eq:pf-strang-step}
\end{equation}
For even $k\ge4$, define the higher-order steps recursively by
\begin{gather}
p_k=(4-4^{1/(k-1)})^{-1},
\qquad (a_1,\ldots,a_5)=(p_k,p_k,1-4p_k,p_k,p_k),
\nonumber\\
t_0=t,
\qquad t_j=t+h\sum_{i=1}^j a_i,
\qquad 1\le j\le5,
\nonumber\\
S_{k,b}(t+h,t)
=S_{k-2,b}(t_5,t_4)\cdots S_{k-2,b}(t_1,t_0).
\label{eq:pf-suzuki-recursion}
\end{gather}
The rightmost factor acts first. Symmetry under reversing the endpoints and
$4p_k^{k-1}+(1-4p_k)^{k-1}=0$ give agreement through order $k$
in this time-dependent Suzuki construction~\cite[Sections~III and V]{WiebeBerryHoyerSanders2010}.
Since $1/3<p_k<1/2$, the endpoint fractions
$0,p_k,2p_k,1-2p_k,1-p_k,1$ lie in $[0,1]$. The middle substep
runs backwards, but all evaluation times remain within the original
interval. The query model allows the required inverse phases. Each recursion multiplies the
number of cost exponentials by five, so a step contains at most
$m_k=5^{k/2-1}$ of them. These steps use the same coefficients in all
coordinates: the coordinate-change unitary $\mathcal{R}$ satisfies
\begin{align}
\mathcal{R}e^{-i\beta A}\mathcal{R}^{-1}
&=\bigotimes_{j=1}^d e^{-i\beta A_j},
\nonumber\\
\mathcal{R}e^{-i\gamma W^\eta}\mathcal{R}^{-1}
&=\bigotimes_{j=1}^d e^{-i\gamma w_\eta}.
\label{eq:pf-global-oracle-identity}
\end{align}
Thus $r$ steps require at most $m_kr$ global phase queries. The hidden
coordinates are used only to analyze this sequence.

Since the mixer is unbounded, we control repeated applications of position and momentum using the weighted norms
\begin{equation}
X_{\eta,b}=\eta(z-b),
\qquad P_\eta=-i\eta\partial_z,
\qquad
\|f\|_{n,\eta,b}=\sum_{D\in\mathcal D_n}\|Df\|_2,
\label{eq:pf-weighted-norm}
\end{equation}
where $\mathcal D_n$ consists of all ordered products of at most $n$
factors $X_{\eta,b},P_\eta$, including the identity.
We write $\mathcal S(\RR)$ for the smooth functions whose derivatives decay faster than every inverse power.

\begin{lemma}[Control of the evolving state]
\label{lem:pf-state-control}
For $0<\eta\le1$, the state $\psi_b^\eta(t)$ evolved from $\phi_b$ remains in $\mathcal S(\RR)$. For every integer $n\ge0$,
\begin{equation}
\sup_{b\in\RR,\;0\le t\le T}
\|\psi_b^\eta(t)\|_{n,\eta,b}
\le C_n(1+T)^n,
\label{eq:pf-moment-bound}
\end{equation}
where $C_n$ is independent of $b,\eta,M,T$.
\end{lemma}

\begin{proof}
Let $\mathcal X_{n,\eta,b}$ be the completion of $\mathcal S(\RR)$ in the norm \eqref{eq:pf-weighted-norm}. The Gaussian smoothing gives
$\eta^j\|w_\eta^{(j)}\|_\infty\le\sqrt{2^j j!}$ for $j\ge0$.
Together with
\begin{equation}
A_b=\frac{P_\eta^2}{2\eta^2}+\frac{X_{\eta,b}^2}{8\sigma^4\eta^2},
\qquad [X_{\eta,b},w_\eta]=0,
\qquad [P_\eta,w_\eta]=-i\eta w_\eta',
\end{equation}
the product rule therefore gives
\begin{equation}
\|w_\eta f\|_{n,\eta,b}\le C_n\|f\|_{n,\eta,b},
\qquad
\|A_b f\|_{n,\eta,b}\le C_n\eta^{-2}\|f\|_{n+2,\eta,b}.
\label{eq:pf-domain-bounds}
\end{equation}
Harmonic evolution mixes position and momentum through bounded linear
combinations:
\begin{align}
e^{i\alpha A_b}X_{\eta,b}e^{-i\alpha A_b}
&=X_{\eta,b}\cos(\alpha/(2\sigma^2))
  +2\sigma^2P_\eta\sin(\alpha/(2\sigma^2)),
\nonumber\\
e^{i\alpha A_b}P_\eta e^{-i\alpha A_b}
&=P_\eta\cos(\alpha/(2\sigma^2))
  -\frac{X_{\eta,b}}{2\sigma^2}\sin(\alpha/(2\sigma^2)).
\end{align}
It follows that $\|e^{-i\alpha A_b}f\|_{n,\eta,b}
\le C_n\|f\|_{n,\eta,b}$ uniformly in real $\alpha$.
The product rule similarly bounds $e^{-i\gamma w_\eta}$ in these norms
for bounded $|\gamma|$. Both evolutions preserve $\mathcal S(\RR)$,
with bounds uniform in $b\in\RR$ and $0<\eta\le1$.

For the exact evolution, factor out the harmonic evolution by writing
\begin{equation}
\chi_b(t)=e^{i\alpha(t)A_b}\psi_b^\eta(t),
\qquad \alpha(t)=\int_0^t(1-s(v))\,dv,
\qquad i\partial_t\chi_b(t)=B_b(t)\chi_b(t),
\end{equation}
where $B_b(t)=s(t)e^{i\alpha(t)A_b}w_\eta e^{-i\alpha(t)A_b}$.
On every $\mathcal X_{n,\eta,b}$, $B_b(t)$ is uniformly bounded and
$B_b(t)f$ is continuous in the weighted norm for each $f$ in that space.
The integral equation
\[
\chi_b(t)=\phi_b-i\int_0^tB_b(v)\chi_b(v)\,dv
\]
has a unique continuous solution in each space. Uniqueness in $L^2$
identifies these solutions, so the evolved state remains in
$\mathcal S(\RR)$. For the quantitative bound, commute an ordered
product past $w_\eta$.
Every nonzero term removes at least one momentum factor and introduces
only bounded scaled derivatives of $w_\eta$. Harmonic conjugation preserves
the number of factors, so
\begin{equation}
\|[D,B_b(t)]f\|_2\le C_n\|f\|_{n-1,\eta,b},
\qquad D\in\mathcal D_n,\quad n\ge1.
\label{eq:pf-commutator-bound}
\end{equation}
The equation
$i\partial_t(D\chi_b)=B_b(t)D\chi_b+[D,B_b(t)]\chi_b$
and unitarity of the evolution generated by $B_b(t)$ then imply
\begin{equation}
F_n(t)\le F_n(0)+C_n\int_0^tF_{n-1}(v)\,dv,
\qquad F_n(t)=\|\chi_b(t)\|_{n,\eta,b}.
\end{equation}
Here $F_0(t)=1$, and the Gaussian initial norms are uniformly bounded for $0<\eta\le1$. Induction gives polynomial growth in $t$; transforming back to $\psi_b^\eta$ proves \eqref{eq:pf-moment-bound}.
\end{proof}

\begin{lemma}[Error of one product-formula step]
\label{lem:pf-one-step-error}
For the schedule and parameter choices above, put $\Lambda=1+\eta^{-2}$. For every even $k\ge2$, there are constants $C_k,c_k>0$ such that
\begin{align}
\bigl\|(S_{k,b}(t+h,t)-U_b^\eta(t+h,t))\psi_b^\eta(t)\bigr\|_2
&\le\mathcal E_k h^{k+1},
&
\mathcal E_k&=C_k\Lambda^{k+1}(1+T)^{2k+2},
\label{eq:pf-local-error}
\end{align}
whenever $0\le t\le t+h\le T$ and $h\Lambda\le c_k$. The constants are independent of $b,d,\eta,M,T$.
\end{lemma}

\begin{proof}
First, bound the schedule derivatives. With $n_M=2M+1$,
the polynomial in \eqref{eq:beta-schedule} can be written as
\begin{equation}
\frac{s_M(u)}{s_f}
=\sum_{j=M+1}^{n_M}\binom{n_M}{j}u^j(1-u)^{n_M-j}.
\label{eq:pf-schedule-polynomial}
\end{equation}
Indeed, differentiating the sum gives
$n_M!\,u^M(1-u)^M/(M!)^2$, the derivative of the normalized integral,
and both expressions vanish at $u=0$. Its coefficients in the full
sum from $j=0$ to $n_M$ are zero or one. Each derivative replaces
them by adjacent differences and introduces a factor no larger than
$n_M$. After $j$ derivatives their magnitudes are at most $2^j$,
while the remaining basis polynomials are nonnegative and sum to one
on $[0,1]$. Derivatives above degree $n_M$ vanish. Thus
\begin{equation}
\|\partial_t^j s\|_\infty
\le\left(\frac{4M+2}{T}\right)^j\le1,
\qquad j\ge1.
\label{eq:pf-schedule-derivatives}
\end{equation}
The last inequality follows from
$T=eC_{\mathrm{ad}}(M+1)\nu^{-6}N$,
$C_{\mathrm{ad}},N\ge1$, and $\nu\le1/2$, which give
$(4M+2)/T\le4\nu^6/(eC_{\mathrm{ad}}N)\le1/(16e)$.

Next, apply Taylor's theorem to the error on the evolved state.
Every phase in $S_{k,b}(t+h,t)$ has an
angle of the form $\lambda h\,c(t+\xi h)$, with $c=s$ or $1-s$
and fixed coefficients $\lambda,\xi$ determined by the recursion.
Equation~\eqref{eq:pf-schedule-derivatives} bounds the angle and
its first $k+1$ derivatives by constants depending only on $k$.
Differentiating a phase $j$ times introduces at most $j$ factors of
$A_b$ or $w_\eta$. Equations~\eqref{eq:pf-domain-bounds} and the
phase bounds consequently give, for $1\le j\le k+1$,
\begin{align}
\|\partial_h^j S_{k,b}(t+h,t)f\|_2
&\le C_k\Lambda^j\|f\|_{2j,\eta,b},
\nonumber\\
\|\partial_h^j\psi_b^\eta(t+h)\|_2
&\le C_k\Lambda^j\|\psi_b^\eta(t+h)\|_{2j,\eta,b}.
\label{eq:pf-step-derivatives}
\end{align}
The second inequality follows by differentiating the Schr\"odinger
equation: each Hamiltonian factor requires at most two additional
position or momentum factors and one factor of $\Lambda$.
The regularity established in Lemma~\ref{lem:pf-state-control}
justifies these derivatives in $L^2$. The local error vector
\[
\Delta_{t,b}(v)
=\bigl(S_{k,b}(t+v,t)-U_b^\eta(t+v,t)\bigr)\psi_b^\eta(t)
\]
is $k+1$ times continuously differentiable on $[0,h]$, and
\eqref{eq:pf-moment-bound} gives
\[
\sup_{0\le v\le h}\|\Delta_{t,b}^{(k+1)}(v)\|_2
\le C_k\Lambda^{k+1}(1+T)^{2k+2}.
\]
The Taylor coefficients are finite expressions in the Hamiltonian
and its derivatives acting on $\mathcal S(\RR)$, so the Suzuki
identities apply and give $\Delta_{t,b}^{(j)}(0)=0$ for
$0\le j\le k$. The integral remainder
\[
\Delta_{t,b}(h)
=\frac1{k!}\int_0^h(h-v)^k\Delta_{t,b}^{(k+1)}(v)\,dv
\]
proves \eqref{eq:pf-local-error}, after absorbing the remainder constant into $C_k$.
\end{proof}

\begin{thm}[Subpolynomial query complexity]
\label{thm:rsw-subpolynomial-queries}
For every fixed $\varepsilon\in(0,1/3)$, a product formula initialized in
the known Gaussian ground state $\ket{\sigma}$ reaches $V_a$
with probability at least $1-\varepsilon$ using
\begin{equation}
Q(d)=d^{o(1)}
\end{equation}
phase queries to $W^\eta$, for a suitable smoothing width and
schedule. Its coefficients depend only on $d$, $\varepsilon$, and the
fixed model parameters, uniformly in the unknown rotation and allowed
shifts.
\end{thm}

\begin{proof}
Lemma~\ref{lem:pf-one-step-error} bounds the error of one step by $\mathcal E_k h^{k+1}$ whenever $h\Lambda\le c_k$. Let $\mathcal S_{k,r,b}$ compose $r$ steps of length $h=T/r$. Inserting one exact step at a time and using unitarity to prevent error amplification gives
\begin{equation}
\sup_b\|\mathcal S_{k,r,b}\phi_b-\psi_b^\eta\|_2
\le\delta_r,
\qquad
\delta_r=\mathcal E_k\frac{T^{k+1}}{r^k},
\qquad r\ge\Lambda T/c_k.
\label{eq:pf-factor-error}
\end{equation}
For the coordinate errors, write $u_j=\psi_j^\eta$ and
$v_j=\mathcal S_{k,r,s_j}\phi_j$. Expanding the tensor-product
difference gives $d$ terms, each with one factor $v_j-u_j$
and otherwise normalized factors. Thus
\begin{align}
\left\|\bigotimes_{j=1}^d v_j
-\bigotimes_{j=1}^d u_j\right\|_2
&\le\sum_{j=1}^d\|v_j-u_j\|_2
\le d\delta_r.
\label{eq:pf-product-state-error}
\end{align}
Choosing $r\ge(4d\,\mathcal E_kT^{k+1}/\varepsilon)^{1/k}$
makes the full-state error at most $\varepsilon/4$.
Equation~\eqref{eq:rsw-probability-transfer} bounds the resulting
change in success probability by $\varepsilon/2$. Since the exact
smoothed evolution succeeds with probability at least
$1-\varepsilon/2$, the product formula succeeds with probability
at least $1-\varepsilon$. Including the step-size condition, its
query count is bounded by
\begin{equation}
Q_k(d)\le m_k\left\lceil
\max\left\{\frac{\Lambda T}{c_k},
\left(\frac{4d\,\mathcal E_kT^{k+1}}{\varepsilon}\right)^{1/k}
\right\}\right\rceil.
\label{eq:pf-query-count}
\end{equation}
Equations~\eqref{eq:pf-smoothed-parameters} and~\eqref{eq:pf-local-error} give
$\Lambda=O(\ell^4)$ and $\mathcal E_k=O_k(\ell^{20k+20})$. The first term in the maximum
is $O_k(\ell^{12})$, and the second is
$O_k(d^{1/k}\ell^{28+28/k})$. Consequently,
\begin{equation}
Q_k(d)\le A_kd^{1/k}\ell^{28+28/k},
\label{eq:pf-fixed-order-bound}
\end{equation}
where $A_k\ge3$ includes the stage count, step-size restriction,
fixed error tolerance, and constants in the smoothing and runtime
estimates, but is independent of $d$.

Increasing $k$ reduces the power of $d$ but increases $A_k$.
We can bound $A_k$ by computable integers: fix integer bounds for
the runtime constants, fixed model parameters, their required
reciprocals, and $\varepsilon^{-1}$. The Gaussian derivative bounds
and bounds $C_\sigma^j\sqrt{j!}$ for ordered position and momentum
products on $\phi_b$ then bound the remaining constants through
finite product-rule, commutator, and Suzuki-recursion calculations.
Enlarge $A_k$ to cover all $d\ge1$, and increase $k$ only when $d$
can absorb its cost.
For sufficiently large $d$, set
\begin{equation}
k(d)=\max\left\{k\ge2:\ k\text{ is even},\quad
2^{k^2}A_k^{2k}\le d\right\},
\label{eq:pf-order-choice}
\end{equation}
using order two when the set is empty. The set is finite, and every
fixed even order eventually qualifies, so $k(d)\to\infty$. The
threshold also gives $A_{k(d)}\le d^{1/(2k(d))}$. Since
$28+28/k\le42$ for $k\ge2$, \eqref{eq:pf-fixed-order-bound} yields
\begin{equation}
Q_{k(d)}(d)\le d^{3/(2k(d))}\ell^{42},
\qquad
\frac{\log Q_{k(d)}(d)}{\log d}
\le\frac{3}{2k(d)}+42\frac{\log\ell}{\log d}
\longrightarrow0.
\end{equation}
This proves $Q(d)=d^{o(1)}$.
\end{proof}

The preceding theorem guarantees localization in $V_a$. Averaging
independently measured output points strengthens this to exact
minimization of the sharp objective while preserving the subpolynomial
query bound.

\begin{prop}[Exact minimization by averaging]
\label{prop:rsw-averaging}
Under the assumptions and phase-oracle model of
Theorem~\ref{thm:rsw-subpolynomial-queries}, for every fixed
$\varepsilon\in(0,1/3)$, suitably chosen parameters and
$m=O(\log(d/\varepsilon))$ independent runs of the preparation
procedure yield measured points $\xv^{(1)},\ldots,\xv^{(m)}$
whose arithmetic mean satisfies
\[
\Pr\!\left[W(\overline{\xv})=-d\right]\ge1-\varepsilon,
\qquad
\overline{\xv}=\frac1m\sum_{r=1}^m\xv^{(r)}.
\]
The total number of phase queries to $W^\eta$ is $d^{o(1)}$.
All parameters depend only on $d$, $\varepsilon$, and the fixed
model parameters, uniformly in the unknown rotation and allowed shifts.
\end{prop}

\begin{proof}
Let $Z$ have density $p(z)=|\varphi_{b,f}^\eta(z)|^2$.
For $|b|\le B$, $0<\nu\le\nu_0$, and
$0\le\eta\le c_\eta\nu^2$, the uniform energy bound in
Lemma~\ref{lem:rsw-smoothing-spectrum} gives
\[
\|p\|_\infty\le L:=\sqrt{2\bar E}.
\]
Moreover, outside $[-\delta,\delta]$,
$D_\eta\ge1/2$, so the eigenvalue equation gives
$(\varphi_{b,f}^\eta)''\ge
\varphi_{b,f}^\eta/(2\nu^2)$.
The exponential comparison in the proof of
Lemma~\ref{lem:deep-well} therefore yields, for every $t\ge0$,
\[
\Pr(|Z|>\delta+t)
\le C_{\mathrm{loc}}\nu e^{-\sqrt2t/\nu}.
\]
These estimates apply directly to the shifted, smoothed endpoint
ground state, with constants independent of $b$, $\eta$, and $\nu$.

Set
\[
\kappa_*=\min\{\delta/8,(32L)^{-1}\},
\qquad
\nu_*=\min\left\{
\nu_0,\frac{1}{4C_{\mathrm{loc}}},
\sqrt{\frac{\sqrt2\kappa_*}{C_{\mathrm{loc}}}}
\right\}.
\]
For $\nu\le\nu_*$, the density and tail bounds imply
\[
\Pr(|Z|\le\delta-4\kappa_*)
\ge1-8\kappa_*L-C_{\mathrm{loc}}\nu\ge\frac12.
\]
Consequently,
\[
|\amsmathbb EZ|
\le\amsmathbb E|Z|
\le\delta-2\kappa_*+
\int_0^\infty\Pr(|Z|>\delta+t)\,dt
\le\delta-\kappa_*.
\]

Define the normalized endpoint ground state in the original
coordinates by
\[
\Phi^\eta(\xv)
=\prod_{j=1}^d
\varphi_{s_j,f}^\eta\bigl((R\xv+\sv)_j\bigr).
\]
The tail bound also gives a uniform exponential moment.
Thus Bernstein's inequality, followed by a union bound over
coordinates, gives a constant $c>0$, depending only on the fixed
model parameters, such that for $m$ independent samples from
$|\Phi^\eta|^2$,
\[
\Pr\!\left(\|R\overline{\xv}+\sv\|_\infty\ge\delta\right)
\le2d e^{-cm}.
\]
Here
$R\overline{\xv}+\sv
=m^{-1}\sum_{r=1}^m(R\xv^{(r)}+\sv)$,
so averaging requires no knowledge of $R$ or $\sv$.
Choose
\[
m=\left\lceil c^{-1}\log\frac{4d}{\varepsilon}\right\rceil,
\qquad
\varepsilon_{\mathrm{run}}=\frac{\varepsilon}{2m}.
\]
The ideal failure probability is then at most $\varepsilon/2$.

For the actual preparations, use the preceding parameter choices
with error parameter $\varepsilon_{\mathrm{run}}$, replacing
$\nu_0$ by $\nu_*$ in \eqref{eq:nu-choice}, and set
$\eta=c_\eta\nu^2$.
Reducing this fixed cutoff changes only constants in the earlier
bounds.
Let $\Psi_{\mathrm{ad}}^\eta$ and $\Psi_{\mathrm{PF}}^\eta$ denote,
respectively, the final states of the exact smoothed evolution
and its product-formula simulation.
Equation~\eqref{eq:1d-state-bound-HJ} and the parameter choices
in Theorem~\ref{thm:rsw-harmonic-runtime} give per-coordinate
state error at most $\varepsilon_{\mathrm{run}}/(8d)$,
up to phase.
Choosing these phases and telescoping the tensor products,
then using \eqref{eq:pf-product-state-error}, gives
\[
\inf_{\theta\in\RR}
\|\Psi_{\mathrm{ad}}^\eta-e^{i\theta}\Phi^\eta\|_2
\le\frac{\varepsilon_{\mathrm{run}}}{8},
\qquad
\|\Psi_{\mathrm{PF}}^\eta-\Psi_{\mathrm{ad}}^\eta\|_2
\le\frac{\varepsilon_{\mathrm{run}}}{4}.
\]
Both comparisons are in $L^2(\RR^d)$ at the same smoothing
width; the coordinate change is unitary.
The triangle inequality therefore bounds the state error by
$3\varepsilon_{\mathrm{run}}/8$, up to phase.
Position measurement gives total-variation distance at most
$\varepsilon_{\mathrm{run}}$ from the ideal output law.
Across $m$ independent runs, the joint total-variation distance
is at most $m\varepsilon_{\mathrm{run}}=\varepsilon/2$.
Adding the ideal failure probability proves the claimed
sharp-well success guarantee.

Finally, retaining the error dependence in
\eqref{eq:pf-query-count}, the total query count at each fixed
even order $k\ge2$ satisfies
\[
\begin{aligned}
Q_{\mathrm{tot},k}
&=O_k\!\left(
m\left(\frac{d}{\varepsilon_{\mathrm{run}}}\right)^{1/k}
\left[\log\frac{2d}{\varepsilon_{\mathrm{run}}}\right]^{28+28/k}
\right)=O_{k,\varepsilon}\!\left(
d^{1/k}[\log(2d)]^{29+29/k}
\right).
\end{aligned}
\]
Applying the increasing-order argument from the proof of
Theorem~\ref{thm:rsw-subpolynomial-queries}, with the computable
constants appropriate to these bounds, yields
$Q_{\mathrm{tot}}(d)=d^{o(1)}$ for each fixed $\varepsilon$.
\end{proof}

\section{Random-ray tomography for the rotated square well}
\label{app:random-ray-tomography}

In this appendix we give the details behind the rotation-learning analysis used in Section~\ref{subsubsec:learn-rotation}. We consider the aggregate-count oracle
\begin{equation}
W_{R,s}(x)
=
{-}\sum_{j=1}^d
\mathbf 1\{|a_j\cdot x+s_j|<\delta\},
\qquad
x\in\RR^d,
\label{eq:app-rsw-oracle}
\end{equation}
where \(a_1^\top,\ldots,a_d^\top\) are the rows of \(R\in O(d)\), and we assume \(\|s\|_\infty\le B\). The target region is
\begin{equation}
V_d
=
\bigcap_{j=1}^d
\{x:|a_j\cdot x+s_j|<\delta\}.
\end{equation}
Equivalently, \(V_d\) is the rotated cube bounded by the \(2d\) hyperplanes
\begin{equation}
H_{j,\tau}
=
\{x:a_j\cdot x+s_j=\tau\delta\},
\qquad
j=1,\ldots,d,
\qquad
\tau\in\{+1,-1\}.
\label{eq:app-boundary-hyperplanes}
\end{equation}
If the hyperplanes \(H_{j,+}\) and \(H_{j,-}\) are known, then their midpoint is the center hyperplane
\begin{equation}
a_j\cdot x+s_j=0.
\end{equation}
Thus reconstructing all \(2d\) boundary hyperplanes allows one to solve \(Rx+s=0\) and recover the center
\begin{equation}
x_c=-R^\top s\in V_d.
\end{equation}

\subsection{Fixing the radius of the random line scan}

We now justify the geometry behind the random-ray tomography algorithm. Fix an affine line
\[
\ell(t)=r+tv,
\qquad
\|v\|=1,
\qquad
\|r\|\le \rho.
\]
Along this line, the oracle becomes the one-dimensional step function
\begin{equation}
W_{R,s}(r+tv)
=
{-}\sum_{j=1}^d
\mathbf 1\{|a_j\cdot r+t(a_j\cdot v)+s_j|<\delta\}.
\label{eq:app-line-step-function}
\end{equation}
The line crosses \(H_{j,\tau}\) when
\[
a_j\cdot(r+tv)+s_j=\tau\delta.
\]
If \(a_j\cdot v\neq 0\), which holds with probability one for random \(v\), the crossing time is
\begin{equation}
t_{j,\tau}
=
\frac{\tau\delta-s_j-a_j\cdot r}{a_j\cdot v}.
\label{eq:app-crossing-time}
\end{equation}
At this time, the \(j\)-th indicator in \eqref{eq:app-line-step-function} changes value. Thus jumps of the one-dimensional oracle trace reveal intersections with the boundary hyperplanes, up to the resolution of the line scan. We now fix the half-length \(T\) of the scanned interval \(t\in[-T,T]\). In \eqref{eq:app-crossing-time} the numerator is
\(O(\delta+B+\rho)\) while \(|a_j\cdot v|=\Theta(d^{-1/2})\) for a random direction, so
crossings occur at \(|t_{j,\tau}|=\Theta\bigl((\delta+B+\rho)\sqrt d\bigr)\).

\begin{lemma}[Natural scan radius]
\label{lem:app-natural-scan-radius}
Let \(v\sim\mathrm{Unif}(\Sph^{d-1})\), let \(a_j\) be any fixed unit vector, and set
\(
A:=\delta+B+\rho.
\)
If \(
T=C_T A\sqrt d
\)
with \(C_T\ge 2\), then for every \(j\) and \(\tau\in\{\pm1\}\),
\[
\PP\bigl(|t_{j,\tau}|\le T\bigr)
\ge
\frac{3}{16}.
\]
\end{lemma}

\begin{proof}
The numerator in \eqref{eq:app-crossing-time} is uniformly bounded:
\[
|\tau\delta-s_j-a_j\cdot r|
\le
\delta+B+\rho
=
A.
\]
Therefore a sufficient condition for \(|t_{j,\tau}|\le T\) is
\begin{equation}
|a_j\cdot v|\ge \frac{A}{T}.
\label{eq:app-hit-sufficient}
\end{equation}
Since \(v\) is uniform on the unit sphere and \(a_j\) is fixed, \(a_j\cdot v\) has the same distribution as the first coordinate of a uniformly random point on \(\Sph^{d-1}\). Let
\[
X=(a_j\cdot v)^2.
\]
Then
\[
\EE X=\frac1d,
\qquad
\EE X^2=\frac{3}{d(d+2)}.
\]
By the Paley--Zygmund inequality, with \(\theta=1/4\),
\begin{align}
\PP\!\left(X\ge \frac{1}{4d}\right)
&=
\PP\!\left(X\ge \frac14 \EE X\right)  \ge
\left(1-\frac14\right)^2
\frac{(\EE X)^2}{\EE X^2} =
\frac{9}{16}\frac{d+2}{3d}
\ge
\frac{3}{16}.
\end{align}
Equivalently,
\[
\PP\!\left(|a_j\cdot v|\ge \frac{1}{2\sqrt d}\right)
\ge
\frac{3}{16}.
\]
If \(T=C_T A\sqrt d\) with \(C_T\ge 2\), then
\[
\frac{A}{T}
=
\frac{1}{C_T\sqrt d}
\le
\frac{1}{2\sqrt d}.
\]
Combining this with \eqref{eq:app-hit-sufficient} proves the claim.
\end{proof}

\subsection{The random-ray tomography algorithm}

We now state the algorithm in the notation used in the proof. Choose constants \(C_M,C_T,c_\Delta,\rho>0\), with \(C_T\ge2\), and define
\begin{equation}
M=\lceil C_M d\rceil,
\qquad
T=C_T(\delta+B+\rho)\sqrt d,
\qquad
\Delta t=\frac{c_\Delta}{\sqrt d}.
\label{eq:app-ray-parameters}
\end{equation}
For each \(m=1,\ldots,M\), draw \(v_m\sim\mathrm{Unif}(\Sph^{d-1})\) and draw \(r_m\) from a distribution supported in \(\|r_m\|\le\rho\). Define the line
\[
\ell_m(t)=r_m+t v_m.
\]
Query the oracle on the grid
\[
t_q=-T+q\Delta t,
\qquad
q=0,1,\ldots,\left\lceil \frac{2T}{\Delta t}\right\rceil.
\]
Whenever adjacent grid values differ, use the midpoint
\[
\widehat p_{m,q}
=
\ell_m\!\left(\frac{t_q+t_{q+1}}2\right)
\]
as an approximate boundary point. We now prove that \(\Theta(d)\) random lines suffice.
\begin{lemma}[\(\Theta(d)\) lines hit every boundary hyperplane]
\label{lem:app-theta-d-lines}
Let \(C_T\ge2\) and \(M=\lceil C_M d\rceil\). Let \(\epsilon\) denote the probability that the
crossing of a given sheet shares its grid cell with another sheet, and assume
\(\epsilon\le 3/32\), so that each line records a point of \(H_{j,\tau}\) with probability at
least \(p_0=3/16-\epsilon\ge 3/32\). If \(C_M\ge 4/p_0\), then with probability at least
\(1-2d\,e^{-d/2}\) every boundary hyperplane \(H_{j,\tau}\) contributes a recorded boundary
point on at least \(d\), and at most \(M\), of the sampled lines.
\end{lemma}
\begin{proof}
Fix \(j\in\{1,\ldots,d\}\) and \(\tau\in\{\pm1\}\). Let
\[
N_{j,\tau}
=
\#\{m:\ \ell_m \text{ crosses }H_{j,\tau}\text{ inside }[-T,T]\text{, and no other sheet is
crossed in the same grid cell}\}.
\]
The second requirement is what makes the crossing visible: \(W_{R,s}\) reports only the number
of slabs, so if a cell contains two crossings the recorded jump is \(0\) or \(\pm2\) and no
point is attributed to \(H_{j,\tau}\). By Lemma~\ref{lem:app-natural-scan-radius} a line
crosses \(H_{j,\tau}\) inside the scanned interval with probability at least \(3/16\), and by
hypothesis the crossing shares its cell with another sheet with probability at most
\(\epsilon\le 3/32\), so each line contributes to \(N_{j,\tau}\) with probability at least
\(p_0=3/16-\epsilon\ge 3/32\). Since a line meets a hyperplane at most once, also
\(N_{j,\tau}\le M\).
Therefore \(N_{j,\tau}\) is stochastically lower bounded by a binomial random variable with
\(M\) trials and success probability \(p_0\). In particular,
\[
\EE N_{j,\tau}\ge p_0 M\ge 4d,
\]
using \(C_M\ge 4/p_0\). Then \(d\le \frac12 \EE N_{j,\tau}\), so the Chernoff bound gives
\begin{align}
\PP(N_{j,\tau}<d)
&\le
\PP\!\left(N_{j,\tau}<\frac12\EE N_{j,\tau}\right)  \le
\exp\!\left(-\frac{\EE N_{j,\tau}}8\right)
\le
e^{-d/2}.
\end{align}
Taking a union bound over all \(2d\) boundary hyperplanes,
\[
\PP\!\left(\exists(j,\tau):N_{j,\tau}<d\right)
\le
2d\,e^{-d/2}.
\]
This proves the claim. The hypothesis \(\epsilon\le3/32\) holds once \(c_\Delta\) is small:
the \(2d\) crossing times of a line are spread over a window of width \(2T=\Theta(\sqrt d)\),
so the expected number of other crossings in a cell of width \(\Delta t=c_\Delta/\sqrt d\) is
\(O(c_\Delta)\), uniformly in \(d\).
\end{proof}

From Lemma~\ref{lem:app-theta-d-lines}, the point cloud contains enough generic points to reconstruct the \(2d\) boundary hyperplanes. Pairing approximately parallel hyperplanes separated by approximately \(2\delta\) then recovers the midpoint hyperplanes \(a_j\cdot x+s_j=0\). Solving the resulting linear system gives an estimate \(\widehat x\) of the center \(x_c=-R^\top s\). For sufficiently small reconstruction error, or after a final verification query, this yields a point with \({W_{R,s}(\widehat x)=-d}\).
\subsection{Coverage threshold for the reconstruction stage}
\label{subsec:app-random-ray-lower}
The scan of \eqref{eq:app-ray-parameters} outputs a set \(\widehat{\mathcal P}=\{\widehat p_{m,q}\}\) of recorded boundary points, each an approximate location where a scanned line crossed one of the \(2d\) hyperplanes \(H_{j,\tau}\). Since only cells whose two endpoint values differ by exactly one are recorded, every point of \(\widehat{\mathcal P}\) is attributable to a single sheet, and \(|\widehat{\mathcal P}|=\sum_{j,\tau}N_{j,\tau}\) in the notation of Lemma~\ref{lem:app-theta-d-lines}. The reconstruction stage turns this point set into estimates of the hyperplanes. We specify it in full and prove the one lemma the query count relies on: how many recorded points must be probed before all \(2d\) hyperplanes are seen.
Fix constants \(C_{\mathrm{pt}},C_{\mathrm{dir}}>0\) and a probe radius \(\eta>0\), and set
\begin{equation}
N_{\mathrm{pt}}=\lceil C_{\mathrm{pt}}\,d\log d\rceil,
\qquad
N_{\mathrm{dir}}=\lceil C_{\mathrm{dir}}\,d\rceil ,
\label{eq:app-recon-parameters}
\end{equation}
the number of recorded points we probe and the number of probe directions used at each. Draw \(N_{\mathrm{pt}}\) points i.i.d.\ uniformly, with replacement, from \(\widehat{\mathcal P}\). At each drawn point \(p\), estimate the local hyperplane normal by querying the oracle at the \(2N_{\mathrm{dir}}\) symmetric probes
\begin{equation}
p\pm\eta z_i,
\qquad
z_i\sim\mathrm{Unif}(\Sph^{d-1}),
\qquad
i=1,\ldots,N_{\mathrm{dir}},
\label{eq:app-probe-scheme}
\end{equation}
and forming
\[
\widehat n(p)\;\propto\;\sum_{i=1}^{N_{\mathrm{dir}}}\bigl(W_{R,s}(p+\eta z_i)-W_{R,s}(p-\eta z_i)\bigr)\,z_i .
\]
Clustering the estimated normals into the \(d\) slab directions, averaging within each cluster and re-projecting the resulting \(d\) directions onto \(\mathbf O(d)\), separating the two parallel sheets of each slab, and solving the linear system \(Rx+s=0\) yields the center estimate \(\widehat x\); a final query evaluates \(W_{R,s}(\widehat x)\).
The reconstruction recovers all \(d\) slab directions only if every slab contributes at least one probed point. Write \(\pi_{j,\tau}\) for the probability that a single drawn point lies on \(H_{j,\tau}\). The next lemma gives the number of draws needed for this coverage.
\begin{lemma}[Coverage threshold]
\label{lem:app-coupon-coverage}
Suppose the drawn points are i.i.d., and that conditionally on the scan each lies on hyperplane \(H_{j,\tau}\) with probability \(\pi_{j,\tau}\in[c_1/(2d),\,c_2/(2d)]\) for fixed constants \(0<c_1\le c_2\). Let \(\mathrm{Cov}\) be the event that every one of the \(2d\) hyperplanes is hit by at least one drawn point.
\begin{enumerate}
\item If \(N_{\mathrm{pt}}\ge \tfrac{4}{c_1}\,d\log d\), then \(\PP(\mathrm{Cov})\ge 1-2/d\).
\item If \(N_{\mathrm{pt}}\le \tfrac{1-\varkappa}{c_2}\,d\log d\) for a fixed \(\varkappa\in(0,1)\), then \(\PP(\mathrm{Cov})\to0\) as \(d\to\infty\).
\end{enumerate}
\end{lemma}
\begin{proof}
Write \(I_{j,\tau}=\mathbf 1\{H_{j,\tau}\text{ is never hit}\}\) and \(Z=\sum_{j,\tau}I_{j,\tau}\),
so \(\mathrm{Cov}=\{Z=0\}\) and \(\EE Z=\sum_{j,\tau}(1-\pi_{j,\tau})^{N_{\mathrm{pt}}}\).
For (1), using \(1-\pi\le e^{-\pi}\) and \(\pi_{j,\tau}\ge c_1/(2d)\),
\[
\EE Z\le 2d\,e^{-c_1 N_{\mathrm{pt}}/(2d)}\le 2d\,e^{-2\log d}=2/d,
\]
so \(\PP(\mathrm{Cov})=\PP(Z=0)\ge 1-\EE Z\ge 1-2/d\).
For (2), using \(\log(1-\pi)\ge-2\pi\) on \([0,\tfrac12]\) and \(\pi_{j,\tau}\le c_2/(2d)\),
\[
\EE Z\ge 2d\,e^{-c_2 N_{\mathrm{pt}}/d}\ge 2d\,e^{-(1-\varkappa)\log d}=2\,d^{\varkappa}\to\infty .
\]
The indicators are negatively correlated. A drawn point lies on at most one sheet, so for
\((j,\tau)\ne(j',\tau')\) a single draw misses both with probability
\(1-\pi_{j,\tau}-\pi_{j',\tau'}\), whence
\[
\EE\bigl[I_{j,\tau}I_{j',\tau'}\bigr]=(1-\pi_{j,\tau}-\pi_{j',\tau'})^{N_{\mathrm{pt}}},
\]
and \(1-\pi_{j,\tau}-\pi_{j',\tau'}\le(1-\pi_{j,\tau})(1-\pi_{j',\tau'})\) gives
\[
\mathrm{Cov}\bigl(I_{j,\tau},I_{j',\tau'}\bigr)
=(1-\pi_{j,\tau}-\pi_{j',\tau'})^{N_{\mathrm{pt}}}
-(1-\pi_{j,\tau})^{N_{\mathrm{pt}}}(1-\pi_{j',\tau'})^{N_{\mathrm{pt}}}\le0 .
\]
Hence \(\mathrm{Var}\,Z\le\sum_{j,\tau}\mathrm{Var}\,I_{j,\tau}\le\EE Z\), and by Chebyshev's
inequality \(\PP(Z=0)\le\mathrm{Var}\,Z/(\EE Z)^2\le 1/\EE Z\to0\).
\end{proof}
The hypothesis holds on the event of Lemma~\ref{lem:app-theta-d-lines}: there \(N_{j,\tau}\in[d,M]\) for every sheet, so \(|\widehat{\mathcal P}|\in[2d^2,2dM]\) and
\[
\pi_{j,\tau}=\frac{N_{j,\tau}}{|\widehat{\mathcal P}|}\in\Bigl[\frac{c_1}{2d},\frac{c_2}{2d}\Bigr],
\qquad
c_1=\frac1{C_M},\quad c_2=C_M .
\]
With \(C_{\mathrm{pt}}\ge 4C_M\), Lemma~\ref{lem:app-coupon-coverage}(1) gives coverage with probability \(1-o(1)\). Hence \(N_{\mathrm{pt}}=\Theta(d\log d)\) probed points suffice, and this is the reconstruction budget behind the query count \(Q_{\mathrm{ray}}=\Theta(d^2\log d)\) derived in the main text.
\begin{rem}[Accuracy of the recovered directions]
\label{rem:app-normal-accuracy}
Coverage is necessary but not sufficient: a single probed point gives a normal estimate whose angular error is
\(\theta_1\approx\sqrt{\pi d/(2N_{\mathrm{dir}})}=\sqrt{\pi/(2C_{\mathrm{dir}})}\),
a constant, since the along-normal signal \(N_{\mathrm{dir}}\EE|a_j\cdot z|=\Theta(N_{\mathrm{dir}}d^{-1/2})\) and the transverse noise \(\Theta(N_{\mathrm{dir}}^{1/2})\) carry the same power of \(d\). The same budget, however, delivers \(m=\Theta(\log d)\) probed points per sheet, and averaging within a cluster reduces the error to \(\theta=\theta_1/\sqrt m\). Since the errors are isotropic in the orthogonal complement of \(a_j\), the residual at the recovered center satisfies
\[
\max_j\bigl|a_j\cdot\widehat x+s_j\bigr|\;\approx\;\theta\,B\sqrt{2\log d}\;=\;B\sqrt{\frac{2\pi}{C_{\mathrm{dir}}C_{\mathrm{pt}}}},
\]
independent of \(d\). The reconstruction therefore succeeds, with the stated budget and no additional queries, provided the constants satisfy \(C_{\mathrm{dir}}C_{\mathrm{pt}}\gtrsim 2\pi B^2/\delta^2\), which is a condition on the fixed problem parameters \(B\) and \(\delta\) alone.
\end{rem}

\section{Random walk algorithms for the rotated square well problem}
\label{sec:RandomWalk}

In this appendix, we analyze a natural family of classical random-walk algorithms for the rotated square well problem and prove the cubic lower bound stated in the main text. The guiding idea is simple: since the landscape is flat on large regions and only oracle access is available, a classical algorithm cannot reliably exploit gradient information or local curvature.
A natural strategy is therefore to use random proposals, together with an acceptance rule that never moves to a worse level. The resulting dynamics can be understood in two stages. First, one asks how likely it is that a single proposal remains inside the current level set \(V_k\). This determines which step sizes are even viable. Second, once the admissible scale of the step size is identified, one analyzes the time required to move from level \(k\) to level \(k+1\), and then sums these costs over all levels. The main result of the appendix is that this leads to an expected hitting time of order $\Theta(d^3/\log d)$.

We begin by explicitly describing the random walk algorithm considered in this section.

\begin{algorithm}[h]
\caption{Random walk with fixed step schedule $\{r_k\}_{k=0}^{d-1}$}
\label{alg:rw}
\begin{algorithmic}[1]
\REQUIRE Dimension $d\in\amsmathbb{N}$; slab half-width $\delta>0$; initial point $x_0\in\RR^d$;
oracle access to $W(\cdot)$; \\step schedule $\{r_k\}_{k=0}^{d-1}$ with $r_k>0$
\ENSURE Terminal point $x_\tau$ with ${W(x_\tau)=-d}$ and stopping time $\tau$

\STATE Initialize $t\gets 0$, $x_t\gets x_0$, $K_t\gets {-W(x_t)}$.
\WHILE{$K_t < d$}
    \STATE Sample a direction $U \sim \mathrm{Unif}(\Sph^{d-1})$.
    \STATE Propose $y \gets x_t + r_{K_t}\,U$.
    \STATE Query the oracle and set $K_y \gets {-W(y)}$.
    \IF{$K_y \ge K_t$}
        \STATE Accept: $x_{t+1}\gets y$, $K_{t+1}\gets K_y$.
    \ELSE
        \STATE Reject: $x_{t+1}\gets x_t$, $K_{t+1}\gets K_t$.
    \ENDIF
    \STATE Update $t\gets t+1$.
\ENDWHILE

\STATE Set $\tau \gets t$ and $x_\tau \gets x_t$.
\RETURN $(\tau, x_\tau)$.
\end{algorithmic}
\end{algorithm}

The analysis proceeds in three steps. We first derive complementary upper and lower bounds on the probability that a single proposal from a point at level \(k\) does not decrease the level. These bounds identify the largest admissible step-size schedule that avoids exponentially small acceptance near the top levels. Finally, with this step-size scale in hand, we reduce the passage from level \(k\) to level \(k+1\) to an effective one-dimensional first-passage problem and sum over \(k\) to obtain the $\Theta(d^3/\log d)$ estimate of the query complexity in Section~\ref{app:rw-complexity}.

\subsection{Upper bound on $\PP(X\in V_k)$}
We first bound the probability that a single proposal is accepted, that is, that it does not decrease the level. The mechanism is a loss constraint. At level \(k\) the walker already sits inside \(k\) slabs and only \(d-k\) slabs remain to be gained, so a proposal that knocks more than \(d-k\) of the occupied coordinates out of their slabs cannot be accepted, however many new slabs it happens to enter. Since a step of radius \(r\) knocks each occupied coordinate out with probability \(r|U_i|/2\delta\), this is a binomial tail estimate, and it is most restrictive near the top of the filtration, where \(d-k\) is small. The proposition below makes this precise and shows that the constraint is already significant at \(r\asymp\delta\): once \(k=d-O(\sqrt d)\), a proposal of radius \(\delta\) is accepted with probability \(e^{-\Theta(\sqrt d)}\).

\begin{prop}[Terminal-band bottleneck constrains small steps]\label{prop:terminal_band}
Fix $d\ge 2$, $\delta>0$, and $k\in\{1,\dots,d-1\}$.  Let $U\sim{\rm Unif}(\Sph^{d-1})$ and consider one proposal
\[
X=x+rU,\qquad r\in(0,2\delta].
\]
Assume ${-W(x)=k}$ and relabel coordinates so that $|x_i|<\delta$ for $i\le k$ and $|x_j|\ge \delta$ for $j>k$.
Assume the \emph{typical-inside model}: conditional on ${-W(x)=k}$, the coordinates $(x_1,\dots,x_k)$ are i.i.d.\ ${\rm Unif}([-\delta,\delta])$
and independent of $U$.  Then there exist absolute constants $c,C>0$ such that if
\begin{equation}\label{eq:terminal_band_condition_clean}
d-k \le C\,\frac{r}{\delta}\,k\,\E|U_1|,
\end{equation}
one has
\begin{equation}\label{eq:terminal_band_bound_clean}
\PP(X\in V_k)
\le
\exp\!\Big(-c\,\frac{r}{\delta}\,k\,\E|U_1|\Big)
+
\exp(-c k).
\end{equation}
Consequently, since $\E|U_1|=\Theta(d^{-1/2})$, if $r\ge\delta$ and $k=d-m$ with $m\le C\sqrt d$, then
$\PP(X\in V_k)\le \exp(-\Theta(\sqrt d))$.
\end{prop}
\begin{proof}
Write $m\coloneqq d-k$.  Define the number of lost inside coordinates
\[
L\;\coloneqq\;\sum_{i=1}^k \mathbf{1}\{|X_i|\ge \delta\}.
\]
If $L=\ell$, then among the first $k$ coordinates, at most $k-\ell$ remain inside $(-\delta,\delta)$.
Even in the \emph{best case} where all remaining $m$ coordinates ($j>k$) land inside the slab after the move,
the total number of inside coordinates of $X$ is at most
\[
(k-\ell)+m=d-\ell.
\]
If $\ell>m$, then $d-\ell<k$, so ${-W(X)<k}$ and therefore $X\notin V_k$.
Hence
\begin{equation}\label{eq:Vk_implies_L}
\{X\in V_k\}\subseteq \{L\le m\},
\qquad\text{so}\qquad
\PP(X\in V_k)\le \PP(L\le m).
\end{equation}
Fix $u\in[-1,1]$ and let $z\sim{\rm Unif}([-\delta,\delta])$.  For $r\in(0,2\delta]$ the set
\(
\{z:\ |z+ru|<\delta\}
\)
is exactly the overlap of two intervals of length $2\delta$:
\[
[-\delta,\delta]\cap[-\delta-ru,\ \delta-ru].
\]
The overlap length is $2\delta-|ru|$, so
\begin{equation}\label{eq:interval_overlap}
\PP(|z+ru|<\delta)=\frac{2\delta-|ru|}{2\delta}=1-\frac{r|u|}{2\delta},
\qquad
\PP(|z+ru|\ge\delta)=\frac{r|u|}{2\delta}.
\end{equation}
Under the typical-inside model, conditional on $U$ the random variables $x_1,\dots,x_k$ are independent,
and each event $\{|x_i+rU_i|\ge\delta\}$ depends only on $x_i$.
Therefore, given $U$, the indicators $\mathbf{1}\{|X_i|\ge\delta\}$ for $i\le k$ are independent Bernoullis with
\begin{equation}\label{eq:pi}
\PP(|X_i|\ge\delta\mid U)=\frac{r}{2\delta}|U_i|\qquad (i\le k).
\end{equation}
Thus $L\mid U$ is a Poisson--binomial sum with conditional mean
\begin{equation}\label{eq:muU}
\mu(U)\;\coloneqq\;\E[L\mid U]
=\frac{r}{2\delta}\sum_{i=1}^k |U_i|.
\end{equation}
Whenever $\mu(U)\ge m$, we may apply the standard lower-tail Chernoff bound\footnote{Recall the lower-tail Chernoff bound $\PP(X\le (1-\d)\mu)\le e^{-\frac{\d^2\mu}{2}}$ where $\mu = \EE[X]$ and $\d\in(0,1)$ is arbitrary. We use $X = L|U$ with $\d = 1-\frac{m}{\mu(U)}$ on the event $\mu(U)\ge m$.}
to get
\begin{equation}\label{eq:chernoff_m}
\PP(L\le m\mid U)
\le
\exp\!\left(-\frac{(\mu(U)-m)^2}{2\mu(U)}\right).
\end{equation}
When $\mu(U)<m$ this inequality is not useful, so we will ensure $\mu(U)$ is typically larger than $m$
by conditioning on a high-probability event for $\sum_{i\le k}|U_i|$.
Define
\[
S_k(U)\;\coloneqq\;\sum_{i=1}^k |U_i|.
\]
This function is $\sqrt{k}$-Lipschitz on $\Sph^{d-1}$:
\[
|S_k(u)-S_k(v)|
\le \sum_{i=1}^k |u_i-v_i|
\le \sqrt{k}\,\|u-v\|_2.
\]
Hence, by L\'evy's concentration lemma on the sphere, there is an absolute constant $c>0$ such that
for every $\varepsilon\in(0,1)$,
\begin{equation}\label{eq:levy}
\PP\!\left(S_k(U)\le (1-\varepsilon)\,\E[S_k(U)]\right)\le \exp(-c\,\varepsilon^2 k).
\end{equation}
By exchangeability of coordinates under ${\rm Unif}(\Sph^{d-1})$ we have $\E[S_k(U)]=k\,\E|U_1|$,
so \eqref{eq:levy} becomes
\begin{equation}\label{eq:levy2}
\PP\!\left(\sum_{i=1}^k |U_i|\le (1-\varepsilon)\,k\,\E|U_1|\right)\le \exp(-c\,\varepsilon^2 k).
\end{equation}
On the complementary event, \eqref{eq:muU} yields the deterministic lower bound
\begin{equation}\label{eq:mu_lower}
\mu(U)\ge \frac{r}{2\delta}(1-\varepsilon)\,k\,\E|U_1|.
\end{equation}
Combine \eqref{eq:Vk_implies_L} with the decomposition
\[
\PP(L\le m)=
\E\!\left[\PP(L\le m\mid U)\,\mathbf{1}_{\{S_k(U)\ge (1-\varepsilon)k\E|U_1|\}}\right]
+\PP\!\left(S_k(U)<(1-\varepsilon)k\E|U_1|\right).
\]
The second term is $\le \exp(-c\,\varepsilon^2 k)$ by \eqref{eq:levy2}.
On the first event, \eqref{eq:mu_lower} ensures $\mu(U)$ is at least a constant multiple of
$\frac{r}{\delta}k\E|U_1|$; if \eqref{eq:terminal_band_condition_clean} holds with $C$ chosen small enough
(depending only on the fixed choice of $\varepsilon$, say $\varepsilon=\tfrac12$), then this lower bound
also ensures $\mu(U)\ge m$, so \eqref{eq:chernoff_m} applies.  Moreover, under the same condition the gap
$\mu(U)-m$ is a constant fraction of $\mu(U)$, so the exponent in \eqref{eq:chernoff_m} is
$-\Omega(\mu(U))=-\Omega\!\big(\frac{r}{\delta}k\E|U_1|\big)$.
This yields \eqref{eq:terminal_band_bound_clean} after absorbing constants, and the stated scaling
consequences follow from $\E|U_1|=\Theta(d^{-1/2})$.
\end{proof}

Proposition~\ref{prop:terminal_band} caps how large the step may be near the top of the filtration: for the proposal to be accepted with non-vanishing probability it forces \(\tfrac{r}{\delta}k\,\EE|U_1|=O(1)\), hence \(r=O(\delta\sqrt d/k)\). It says nothing, however, about whether a step on that scale is actually useful. The next proposition supplies the complementary lower bound, showing that at \(r_k=\Theta(\delta\sqrt d/k)\) the walk does remain inside \(V_k\) with probability bounded below by an absolute constant, so that this scale is not merely admissible but attained.

\subsection{Lower bound on $\PP(X\in V_k)$}
\begin{prop}[$\PP(X\in V_k)$ fixes $r_k$]\label{lem:Vk_accept}
Fix integers $d\ge 3$ and $k\in\{1,\dots,d\}$, and a slab half-width $\delta>0$.
Let $U\sim\mathrm{Unif}(\Sph^{d-1})$ and consider a single proposal
\begin{align}
X = x + r\,U,
\qquad r\in(0,\delta].
\end{align}
Assume $x\in V_k$ and choose an instantaneous orthonormal basis such that
\begin{align}
|x_i|<\delta \quad \text{for all } i\le k.
\end{align}
Assume moreover the following \emph{typical-inside model}: conditional on $x\in V_k$, the coordinates
$(x_1,\dots,x_k)$ are i.i.d.\ $\mathrm{Unif}([-\delta,\delta])$ and independent of $U$.
Define the retention event
\begin{equation}\label{eq:Stay-def}
    \mathsf{Stay}_k\;\coloneqq\;\bigcap_{i=1}^k\{|X_i|<\delta\},
\end{equation}
on which none of the $k$ already-occupied slabs is vacated. Then $\mathsf{Stay}_k\subseteq\{X\in V_k\}$, and the inclusion is strict: $X$ may leave one of the first $k$ slabs and simultaneously enter one of the last $d-k$. Consequently
\begin{equation}\label{eq:Vk_accept_bound}
\PP(X\in V_k)
\;\ge\;
\PP(\mathsf{Stay}_k)
\;\ge\;
\exp\!\left(
-\frac{r}{2\delta}\,k\,\EE|U_1|
-\frac{r^2}{4\delta^2}\,\frac{k}{d}
\right),
\end{equation}
where
\begin{equation}\label{eq:EU1_exact_gamma}
\EE|U_1|
=
\frac{\Gamma(d/2)}{\sqrt{\pi}\,\Gamma((d+1)/2)}
=
\sqrt{\frac{2}{\pi d}}\left(1+O\!\left(\frac{1}{d}\right)\right).
\end{equation}
In particular, for $k\ge\sqrt d/2$ the schedule $r=\delta\sqrt d/(2k+1)$ satisfies $r\le\delta$, and bounding $\EE|U_1|\le(\EE U_1^2)^{1/2}=d^{-1/2}$ gives
\begin{equation}\label{eq:r_delta_over_k_consequence}
\PP(X\in V_k)
\;\ge\;
\exp\!\left(-\frac{k}{2(2k+1)}-\frac{k}{4(2k+1)^2}\right)
\;\ge\;
e^{-1/4-1/36}
\;>\;
\frac34 ,
\end{equation}
so the retention probability is bounded below by an absolute constant on the scale $r=\Theta(\delta\sqrt d/k)$.
\end{prop}
\begin{proof}
If $\mathsf{Stay}_k$ occurs, then at least the first $k$ coordinates of $X$ lie in $(-\delta,\delta)$,
so ${-W(X)\ge k}$, hence $X\in V_k$. Therefore $\PP(X\in V_k)\ge \PP(\mathsf{Stay}_k)$. Fix $u\in[-1,1]$ and let $z\sim \mathrm{Unif}([-\delta,\delta])$. For any $r>0$,
translation invariance of Lebesgue measure gives
\begin{align}
\PP(|z+ru|<\delta)
=
\frac{\mathrm{Leb}\big([-\delta,\delta]\cap[-\delta-ru,\delta-ru]\big)}{2\delta}
=
\max \left\{0,\,1-\frac{r|u|}{2\delta}\right\},
\end{align}
the two intervals being disjoint once $r|u|\ge 2\delta$.
Under the typical-inside model, conditional on $U$ the random variables $x_1,\dots,x_k$ are i.i.d.\
$\mathrm{Unif}([-\delta,\delta])$ and independent, so the events $\{|x_i+rU_i|<\delta\}$ are independent
over $i\le k$ and each has probability $\max\{0,1-\frac{r}{2\delta}|U_i|\}$. Consequently, for every $r>0$,
\begin{equation}\label{eq:Stay_conditional}
\PP(\mathsf{Stay}_k\mid U)
=
\prod_{i=1}^k \max \left\{0,\,1-\frac{r}{2\delta}|U_i|\right\},
\qquad
\PP(\mathsf{Stay}_k)
=
\EE\left[\prod_{i=1}^k {\max \left\{0,\,1-\frac{r}{2\delta}|U_i|\right\}}\right].
\end{equation}
Since $r\le \delta$ by hypothesis, one has $0\le \frac{r}{2\delta}|U_i|\le \frac12$ for all $i$, so the truncation in \eqref{eq:Stay_conditional} is never active and the products there are the plain ones.
The elementary inequality
\begin{align}
\log(1-z)\ge -z-z^2,\qquad z\in[0,1/2],
\end{align}
applied to $z=\frac{r}{2\delta}|U_i|$, yields
\begin{align}
\sum_{i=1}^k\log\!\left(1-\frac{r}{2\delta}|U_i|\right)
\ge
-\frac{r}{2\delta}\sum_{i=1}^k |U_i|
-\frac{r^2}{4\delta^2}\sum_{i=1}^k U_i^2.
\end{align}
Exponentiating both sides gives the pointwise bound
\begin{align}
\prod_{i=1}^k\left(1-\frac{r}{2\delta}|U_i|\right)
\ge
\exp\!\left(
-\frac{r}{2\delta}\sum_{i=1}^k |U_i|
-\frac{r^2}{4\delta^2}\sum_{i=1}^k U_i^2
\right).
\end{align}
Taking expectations and using Jensen's inequality in the form
$\EE[e^{-Y}]\ge e^{-\EE[Y]}$ (since $y\mapsto e^{-y}$ is convex) gives
\begin{align}
\PP(\mathsf{Stay}_k)
\ge
\exp\!\left(
-\frac{r}{2\delta}\,\EE\!\left[\sum_{i=1}^k |U_i|\right]
-\frac{r^2}{4\delta^2}\,\EE\!\left[\sum_{i=1}^k U_i^2\right]
\right).
\end{align}
By exchangeability of coordinates of $U$,
\begin{align}
\EE\!\left[\sum_{i=1}^k |U_i|\right]=k\,\EE|U_1|,
\qquad
\EE\!\left[\sum_{i=1}^k U_i^2\right]=k\,\EE[U_1^2]=\frac{k}{d},
\end{align}
where $\EE[U_1^2]=1/d$ holds exactly for $U\sim\mathrm{Unif}(\Sph^{d-1})$. Substituting these identities
proves \eqref{eq:Vk_accept_bound}. It remains to verify \eqref{eq:EU1_exact_gamma}. The marginal density of $U_1$ is
\begin{align}
f_d(u)=c_d(1-u^2)^{(d-3)/2}\,\mathbf{1}_{\{|u|\le 1\}},
\qquad
c_d=\frac{\Gamma(d/2)}{\sqrt{\pi}\,\Gamma((d-1)/2)}.
\end{align}
Therefore,
\begin{align}
\EE|U_1|
=
2\int_0^1 u f_d(u)\,du
=
2c_d\int_0^1 u(1-u^2)^{(d-3)/2}\,du.
\end{align}
With the substitution $v=u^2$ (so $dv=2u\,du$),
\begin{align}
\EE|U_1|
=
c_d\int_0^1 (1-v)^{(d-3)/2}\,dv
=
c_d\cdot \frac{2}{d-1}.
\end{align}
Using $\Gamma((d+1)/2)=\frac{d-1}{2}\Gamma((d-1)/2)$ yields
\begin{align}
\EE|U_1|
=
\frac{\Gamma(d/2)}{\sqrt{\pi}\,\Gamma((d+1)/2)},
\end{align}
which is the first equality in \eqref{eq:EU1_exact_gamma}. The asymptotic expansion follows from
Stirling's formula applied to the ratio of Gamma functions. Finally, \eqref{eq:r_delta_over_k_consequence} is obtained by substituting $r=\delta\sqrt d/(2k+1)$ into \eqref{eq:Vk_accept_bound}, together with $\EE|U_1|\le(\EE U_1^2)^{1/2}=d^{-1/2}$, which gives the two exponents $\frac{k}{2(2k+1)}\le\frac14$ and $\frac{k}{4(2k+1)^2}\le\frac1{36}$, the latter maximal at $k=1$.
\end{proof}

The previous bounds identify the admissible order of magnitude of the step size. The only scale that simultaneously avoids exponentially small acceptance and keeps the walk mobile near the terminal levels is \(r_k\asymp \delta\sqrt d/k\). We now make this statement explicit.

\subsection{Fixing the step-size schedule}

Both directions below concern the retention event \(\mathsf{Stay}_k\) of \eqref{eq:Stay-def} rather than the acceptance event \(\{X\in V_k\}\). For sufficiency the distinction is harmless: \(\mathsf{Stay}_k\subseteq\{X\in V_k\}\), so a lower bound on retention is a lower bound on acceptance. For necessity it is retention that is the operative quantity, because the walk advances from level \(k\) to level \(k+1\) by keeping the \(k\) slabs it already occupies and entering one more; a step size at which retention is exponentially small cannot drive the filtration upward through that mechanism. We do not rule out the exchange moves in which a proposal sheds \(j\ge 1\) occupied slabs while simultaneously entering \(j+1\) new ones, which also raise the level; bounding those would require controlling the probability of \(j+1\) simultaneous entries, and we do not attempt it here.

\begin{cor}[Fixing the step-size schedule \(r_k=\Theta(\delta\sqrt d/k)\)]
\label{cor:rk_bottleneck}
Fix \(d\ge 2\) and \(\delta>0\), and work under the typical-inside model. For every \(k\ge C\sqrt d\) with \(C\) a large enough absolute constant, the step size must satisfy \(r_k = O\!\left(\delta\,\frac{\sqrt d}{k}\right)\) in order for the \emph{retention} probability \(\PP(\mathsf{Stay}_k)\) to stay bounded below by an absolute constant. On the other hand, the choice \(r_k=\delta\,\frac{\sqrt d}{2k+1}\) satisfies \(r_k=\Theta\!\left(\delta\,\frac{\sqrt d}{k}\right)\) and keeps \(\PP(X\in V_k)\ge\PP(\mathsf{Stay}_k)\) bounded below by an absolute constant. Hence the natural scale of the admissible step-size schedule is
\[
r_k=\Theta\!\left(\delta\,\frac{\sqrt d}{k}\right).
\]
\end{cor}
\begin{proof}
\emph{Necessity.} The computation leading to \eqref{eq:Stay_conditional} gives, for every $r>0$,
\[
  \PP(\mathsf{Stay}_k\mid U)=\prod_{i=1}^k\max\Bigl\{0,\,1-\tfrac{r}{2\delta}|U_i|\Bigr\},
\]
which is non-increasing in $r$; writing $\bar r\coloneqq\min\{r_k,2\delta\}$ and using $1-z\le e^{-z}$,
\[
  \PP(\mathsf{Stay}_k)\;\le\;\EE\exp\Bigl(-\tfrac{\bar r}{2\delta}\sum_{i=1}^k|U_i|\Bigr)
  \;\le\;\exp\Bigl(-\tfrac{\bar r}{4\delta}\,k\,\EE|U_1|\Bigr)+e^{-ck/4},
\]
the last step by splitting on the event of \eqref{eq:levy2} with $\varepsilon=\tfrac12$ and bounding the
conditional expectation by $1$ on its complement. Suppose $\PP(\mathsf{Stay}_k)\ge c_0$ with $k\ge C\sqrt d$
and $C$ large enough that $e^{-cC\sqrt d/4}\le c_0/2$. Then $\frac{\bar r}{\delta}k\,\EE|U_1|=O(1)$, and
since $\EE|U_1|=\Theta(d^{-1/2})$ this reads $\bar r\,k=O(\delta\sqrt d)$. If $r_k>2\delta$ this would force
$k=O(\sqrt d)$, contradicting $k\ge C\sqrt d$ for $C$ large; hence $\bar r=r_k$ and $r_k=O(\delta\sqrt d/k)$.

\emph{Sufficiency.} For $k\ge\sqrt d/2$ the choice $r_k=\delta\sqrt d/(2k+1)$ satisfies $r_k\le\delta$, so
Proposition~\ref{lem:Vk_accept} applies directly. Substituting it into \eqref{eq:Vk_accept_bound} and using
$\EE|U_1|\le(\EE U_1^2)^{1/2}=d^{-1/2}$,
\[
  \PP(X\in V_k)\;\ge\;\exp\Bigl(-\tfrac{k\sqrt d}{2(2k+1)}\EE|U_1|-\tfrac{k}{4(2k+1)^2}\Bigr)
  \;\ge\;e^{-1/4-1/36}\;>\;0.75 .
\]
This settles both directions on the range $k=\Omega(\sqrt d)$, which is the only range the walk visits: as
noted in Appendix~\ref{app:rw-complexity}, a starting point drawn from $Q_L$ already lies in
$K_0=\Theta(\delta d/L)$ slabs, so $k\ge K_0\gg\sqrt d$ throughout the run and in particular $r_k\le\delta$
at every level, which is the hypothesis of Proposition~\ref{lem:Vk_accept}.
\end{proof}

With the natural step-size schedule fixed, the remaining question is how long the walk takes to move from level \(k\) to level \(k+1\). The key simplification is that, conditioned on the already-captured coordinates, the next successful capture is controlled by motion in a single new transverse direction.

\subsection{Query complexity of the adaptive random walk}
\label{app:rw-complexity}

We use the step size estimates of the previous section to derive a
prediction for the query complexity of the adaptive random walk algorithm.
Consider the version of Algorithm~\ref{alg:rw} that rejects proposals outside $Q_L=[-L,L]^d$, with $Q_L\cap V_d\ne\varnothing$. 
We assume $R$ is drawn at random from the Haar measure of $\mathbf{O}(d)$.
For $k= k_0(d)=C\sqrt{d}$, Corollary~\ref{cor:rk_bottleneck} motivates the adaptive step size choice of 
\begin{equation}
r_k=\delta\sqrt d/(2k+1),
\end{equation}
which retains the occupied slabs with constant probability under the typical-inside model.
In what follows, we will analyze the walk starts at a level $k = k_0(d)$, and that levels are gained one at a time; this still yield a lower bound on the total hitting time starting at level $k=0$.
The derivation below makes additional assumptions about the
accepted motion and the distribution of the remaining coordinates.

\paragraph{Diffusion and confinement scales.}
A proposal has
variance $r_k^2/d$ in every slab direction. We approximate the accepted
motion by driftless diffusion with variance per query
\begin{equation}
v_k\asymp\frac{r_k^2}{d}=\frac{\delta^2}{(2k+1)^2}.
\label{eq:rw-heuristic-variance}
\end{equation}
This assumes that rejection changes the diffusion rate by only a
constant factor. 

For a row $a_j$ of $R$, holding all other slab coordinates fixed gives
a box chord of length at most $2L/\|a_j\|_\infty$, with equality at
$x=0$. A fixed random row $R$ of has
$\|a_j\|_\infty\asymp\sqrt{\log d/d}$ with probability tending to one.
We therefore use the effective interval length
\begin{equation}
H_d\asymp L\sqrt{\frac{d}{\log d}}.
\label{eq:rw-heuristic-width}
\end{equation}
This choice assumes that the remaining coordinates explore intervals
comparable to the maximal chord, rather than much shorter sections.

\paragraph{Wait time at level $k$.}
At level $k$, let $m=d-k$ be the number of uncaptured coordinates.
Let $\Delta_k$ be the wait time at level $k$, which is the random variable describing the number of steps before hitting level $k+1$ when is initialized at a typical point in $V_k$.
Approximate their distances from the nearest target boundary by
independent, approximately uniform points of $[0,H_d]$, with this
approximation renewed whenever the level increases. For one reflected
diffusion on this interval, the chance of absorption at zero by time
$t\lesssim H_d^2/v_k$ satisfies
\begin{equation}
F_k(t)\asymp\frac{\sqrt{v_k t}}{H_d}.
\label{eq:rw-heuristic-hit}
\end{equation}
Indeed, a fraction of this order starts within one diffusion length
$\sqrt{v_k t}$ of the target, and such points have a constant chance
of hitting it. 
Independence then gives
\begin{equation}
\PP(\Delta_k>t)\approx[1-F_k(t)]^m
\approx\exp\!\left(-c m\frac{\sqrt{v_k t}}{H_d}\right),
\qquad
\EE\Delta_k\asymp\frac{H_d^2}{v_k m^2}.
\label{eq:rw-heuristic-race}
\end{equation}
The reflecting boundary controls the large-time tail. Thus the factor
$m^{-2}$ comes from the first arrival among coordinates whose starting
distribution has nonzero density at the target. It follows directly
from the survival probability, without assigning an additional speedup
to the nearest coordinate.

\paragraph{Full query complexity estimate.}
Each step proposal uses one query, including rejected proposals.
Summing the modeled waiting times gives
\begin{equation}
\EE Q\asymp\frac{H_d^2}{\delta^2}
\sum_{k=k_0}^{d-1}\frac{(2k+1)^2}{(d-k)^2}
\asymp\frac{L^2}{\delta^2}\frac{d^3}{\log d}.
\label{eq:rw-heuristic-query}
\end{equation}
The last summand is of order $d^2$, while the full sum is at most
$4d^2\sum_{m\ge1}m^{-2}$. This explains why a bounded number of final
levels carries a constant fraction of the predicted query count.

We emphasize that Eq.~\eqref{eq:rw-heuristic-query} is a scaling prediction derived under a number of (reasonable) assumptions.
The strongest assumption is the renewed entrance distribution: conditioning on having avoided a slab can instead deplete the density near its boundary. 
For example, independent reflected Brownian paths on fixed intervals, run with the common adaptive clock $v_k$ and retained after each capture, acquire an extra $\log d$
in their completion time when starting at $k_0=\alpha d$ for fixed
$0<\alpha<1$.
Nevertheless, this predicted scaling is consistent with numerical simulations up to $d=100$ in Figure~\ref{fig:geometry-aware}.
In any case, the random ray tomography algorithm outperforms this adaptive random walk algorithm with a provable $\Theta(d^2\log d)$ query complexity.

\section{Resource Complexity for Hardware Implementation} 
\label{sec:hardware}

In this section, we estimate the resource complexity for CV-QAOA on continuous-variable quantum computers and qubit-based quantum computer.

\subsection{Realization on continuous-variable quantum computer}

Recent experimental advances in continuous-variable (CV) quantum computing have significantly advanced the field toward scalable and fault-tolerant architectures. Notably, researchers have demonstrated high-quality generation of Gottesman–Kitaev–Preskill (GKP) states, a key ingredient for error correction in CV systems, using both bulk optics and integrated photonic platforms. Large-scale entangled cluster states have been realized in time, space, and frequency domains, enabling measurement-based quantum computing with continuous variables. Integrated silicon-nitride photonic chips and high-efficiency detectors have enabled on-chip implementations of CV logic with improved scalability and stability. These developments, combined with progress in fault-tolerant encoding and gate teleportation protocols, mark a transition from proof-of-concept experiments to increasingly viable hardware platforms for universal CV quantum computation.

We now describe the gate complexity $G$ and the space complexity $S$ (number of bosonic q-mode) for implementing the CV-QAOA based on the continuous variable quantum computer architecture.
Note that we can decompose any continuous-variable unitary circuit into a sequence of Gaussian gates and elementary non-Gaussian gates such as cubic phase gates~\cite{Kalajdzievski_2021}. For the QAOA applied to polynomial functions of bounded degree with $m$ monomial terms, the gate complexity of $e^{-i\gamma C}$ is $O(poly(m))$ if the approximate decomposition via commutators approach is used~\cite{Kalajdzievski_2021}.
The gate complexity of $e^{-i\beta B}$ is $O(d)$.
Hence, for a depth-$p$ QAOA with $2p$ parameters, the gate complexity is $G=O(p(poly(m)+d))$. The space complexity is at least $d$, but could be more if we use ancilla q-modes to mediate implementation of non-Gaussian gates.
When ancilla can be reset and reused, the space complexity is $S=d + O(1)$.

We provide more detailed analysis on gate and space complexities below.

\paragraph{Gate Complexity. }
A generic continuous-variable unitary circuit can be decomposed into a sequence of Gaussian gates and elementary non-Gaussian gates~\cite{Kalajdzievski_2021}. We start with the analysis on one layer of the ansatz that involves one kinetic evolution \( e^{-i\beta_j \hat{B}} \) and one potential evolution \( e^{-i\gamma_j \hat{C}} \). The total gate complexity for the Depth-$p$ CV-QAOA will be $p$ times of the cost of the single layer CV-QAOA.

For the kinetic term \( e^{-i\beta_j \hat{B}} \), it is the quadratic phase gate and can be implemented with \( \mathcal{O}(d) \) Gaussian gates. For the potential term \( e^{-i\gamma_j \hat{C}} \), we consider the following two cases:

(1) \quad \( C(\hat{\xv}) \) is a polynomial of degree $l \leq 2$, then it is implementable with Gaussian gates. This is the case for the CV-QAOA on quadratic function optimization. To be specific, consider the potential function
\begin{align}
    C(\xv) = \frac{1}{2} \sum_{i,j} Q_{ij} x_i x_j + \sum_k b_k x_k + c,
\end{align}
where \( Q = Q^\top \in \amsmathbb{R}^{d \times d} \). The corresponding unitary $U = e^{-i \gamma_j C(\hat{\xv})}$
is a Gaussian unitary because it involves at most quadratic and linear functions of the position operators \( \hat{x}_i \). As such, it can be exactly implemented using only Gaussian gates on a continuous-variable quantum computer. To implement \( e^{-i \gamma_j C(\hat{\xv})} \), we proceed in three steps:

\begin{enumerate}
    \item Diagonalize \( Q = O^\top \Lambda O \) via a linear interferometer:
        \( \mathcal{O}(d^2) \) beamsplitters and phase shifters.
    
    \item Apply single-mode quadratic unitaries \( e^{-i \lambda_k \hat{x}_k^2/2} \) in the rotated basis:
        \( \mathcal{O}(d) \) single-mode operations.
    
    \item Implement displacements \( e^{-i b_k \hat{x}_k} \):
         \( \mathcal{O}(d) \) displacement gates.
\end{enumerate}

Combining the above, the total gate complexity to implement \( e^{-i \gamma_j C(\hat{\xv})} \) is $\mathcal{O}(d^2)$
using only Gaussian gates.
    
(2) If \( C(\hat{\xv}) \) is a polynomial of degree \( l > 2 \), then \( e^{-i\gamma_j \hat{C}} \) is a non-Gaussian gate, which requires synthesis from Gaussian operation and cubic phase gate. Assume that $C(\hat{\xv}) = \sum a_k \hat{\xv}^k$ ($\hat{\xv}^k$ is a polynomial with all variables' power sum up to $k$), we consider Trotterization on \( e^{-i\gamma_j \hat{C}} \approx \prod e^{-i \gamma_j a_k \hat{\xv}^k} \). Depending on the polynomials structures, there exist both approximate and exact decomposition approaches to realize $e^{-i \gamma_j a_k \hat{\xv}^k}$ ~\cite{Kalajdzievski_2021}. The complexity of synthesizing $e^{-i \gamma_j a_k \hat{\xv}^k}$ for $k > 2$ from non-Gaussian primitives (e.g., cubic phase plus Gaussians) is $O(poly(k))$ if the approximate decomposition via commutators is used~\cite{Kalajdzievski_2021}.

\paragraph{Qumode Complexity.}
The minimal number of logical q-modes required is \( n \), corresponding to the number of variables in the optimization problem. Additional qumodes may be needed as ancillas for implementing non-Gaussian gates. For example, each cubic phase gate requires one non-Gaussian ancilla mode. In the worst-case (no ancilla reuse), total qumode complexity is \( n + a \), where $a$ is the number of time that the cubic phase gate appears in the synthesized circuit for the non-Gaussian gate. With ancilla reuse (e.g., time-sequenced circuits), the number of ancilla modes can be reduced to a small constant (e.g., 1), at the cost of increased circuit depth and latency.

\subsection{Realization on qubit-based quantum computer}

Recent experimental progress in qubit-based quantum computing has led to notable improvements in scalability, control fidelity, and fault tolerance across multiple platforms on the scale of hundreds of qubits. Researchers have demonstrated more efficient qubit interactions, enabling faster gate operations and more precise control. Advances in error mitigation and quantum error correction have significantly reduced gate error rates, moving closer to the thresholds required for fault-tolerant computation. Breakthroughs in logical qubit encoding and magic state distillation have marked important steps toward universal quantum computing. New qubit modalities, including those based on exotic matter and novel materials, have also expanded the landscape of viable architectures. Together, these developments highlight a steady convergence toward scalable and reliable quantum processors capable of addressing real-world computational challenges.

In this subsection, we consider realizations of CV-QAOA on a near-term qubit-based quantum computer. Suppose the continuous domain $R^d$ is truncated to a box and discretized into $m$ grid points per dimension. We can then map initial state and the continuous-variable unitary using the finite-difference approach on a qubit-based quantum computer.

\paragraph{Space Complexity.} The total number of discrete points is $N = m^d$.
The Hilbert space becomes \( \amsmathbb{C}^{N} \cong (\amsmathbb{C}^2)^{n_q} \) with $
n_q = \lceil \log_2 N \rceil = \lceil d \log_2 m \rceil \quad \text{qubits}$.
This gives the qubit complexity $n_q = \mathcal{O}(d \log m)$. Additional ancilla qubits may be required in realizing the unitary gates in CV-QAOA according to the following discussion. 

\paragraph{Gate Complexity.} 
We start with the kinetic term driver. In the discretized system (finite grid with periodic boundary conditions), \( \hat{p}_i^2 \) is diagonal in the Fourier basis. Thus, the kinetic evolution can be expressed as:
\[
e^{-i\beta_j \hat{B}} = \mathcal{F}^\dagger \cdot e^{-i\beta_j \hat{B}_{\text{diag}}} \cdot \mathcal{F},
\]
where \( \mathcal{F} \) is the multi-dimensional Quantum Fourier Transform (QFT), \( \hat{B}_{\text{diag}} \) is diagonal in Fourier space (contains the eigenvalues of the discrete Laplacian). The diagonal entries are proportional to \( \sum_i k_i^2 \), with \( k_i \in \{0, \ldots, m-1\} \) denoting discretized momentum labels. The above transformation can be realized with QFT and phase kickback. The complexity is $\mathcal{O}(d \cdot \log^2 m)$ per layer. If there are \( p \) QAOA layers, the total cost scales as \( \mathcal{O}(p \cdot d \cdot \log^2 m) \).

Next we analyze the gate complexity of implementing \( e^{-iC(x)} \) on a qubit-based quantum computer. We consider the problem of implementing the unitary \( e^{-iC(x)} \) on a qubit-based quantum computer, where \( C(x) \) is a classically efficiently computable function such as a low-degree polynomial. Let \( x \in \amsmathbb{R}^d \) be a \( d \)-dimensional variable discretized on an \( m \)-point grid per dimension. The input register, which encodes the discretized \( x \), requires \( n = d \log_2 m \) qubits. We assume that \( C(x) \) is a total degree-\( k \) polynomial in \( d \) variables under $r$-bit precision, expressed as
\[
C(x) = \sum_{\alpha} c_\alpha x^\alpha, \quad |\alpha| \leq k,
\]
where \( x^\alpha = \prod_{i=1}^d x_i^{\alpha_i} \) and \( \alpha \in \amsmathbb{N}^d \) is a multi-index. The number of such monomial terms scales as \( T=\binom{d + k}{k} = \mathcal{O}(d^k) \). Using reversible quantum arithmetic circuits and phase kickback, the gate complexity can be bounded by $\mathcal{O}(\mathrm{poly}(d^k, \log m, p, k))$.

\end{appendices}

\end{document}